\RequirePackage{fix-cm}
\documentclass{svjour3}                     
\smartqed  
\usepackage{graphicx}
\usepackage{amsfonts}
\usepackage{amsmath}
\usepackage{amssymb}
\usepackage{bm}
\usepackage{color}
\usepackage{ctable}
\usepackage[multiple]{footmisc}
\usepackage{mathrsfs}
\usepackage{pdflscape}
\usepackage{rotating}
\usepackage{scalerel}
\usepackage{tabulary}
\usepackage{hyperref}
\usepackage{ulem}
\newcommand{\cc}{\mathrm{\Lambda}}
\newcommand{\CA}{\mathcal A}
\newcommand{\CC}{\mathcal C}
\newcommand{\CD}{\mathcal D}
\newcommand{\CE}{\mathcal E}
\newcommand{\CEt}{\CE_{\bm{\partial}_t}}
\newcommand{\CEn}{\CE_{\bm n}}

\newcommand{\CK}{\mathcal K}
\newcommand{\CL}{\mathcal L}
\newcommand{\CM}{\mathcal M}
\newcommand{\CP}{\mathcal P}
\newcommand{\CQ}{\mathcal Q}
\newcommand{\CR}{\mathcal R}

\newcommand{\CT}{\mathcal T}
\newcommand{\CV}{\mathcal V}
\newcommand{\CW}{\mathcal W}

\newcommand{\MD}{M_\CD}
\newcommand{\QD}{\CQ_\CD}
\newcommand{\PD}{\CP_\CD}

\newcommand{\SR}{\mathscr R}
\newcommand{\TD}{\CT_\CD}
\newcommand{\VD}{\CV_\CD}
\newcommand{\VDb}{\CV_\CD^b}
\newcommand{\WD}{\mathcal W_{\mathcal D}}

 1
\newcommand{\bdot}{{\displaystyle \cdot}}

\newcommand{\albe}{{\alpha\beta}}
\newcommand{\mnu}{{\mu\nu}}

\newcommand{\CAt}{\tilde {\mathcal A}}
\newcommand{\epsilont}{\tilde \epsilon}
\newcommand{\omegat}{\tilde \omega}

\newcommand{\sigmat}{\tilde \sigma}
\newcommand{\SRt}{\tilde {\mathscr R}}
\newcommand{\Thetat}{\tilde \Theta}

\newcommand{\average}[1]{\left\langle {\strut #1} \right\rangle_{\CD}}  
\newcommand{\naverage}[1]{\left \langle {\strut #1} \right\rangle} 
\newcommand{\baverage}[1]{\left\langle {\strut #1} \right\rangle^b_{\CD}}
\newcommand{\Eaverage}[1]{\left\langle {\strut #1} \right\rangle^E_{\CD}}
\newcommand{\Rbaverage}[1]{\left\langle {\strut \overset{\phantom{.}}{#1}} \right\rangle^b_\CD} 

\newcommand{\baveragedot}[1]{\left\langle {\strut #1} \right\rangle_{\CD}^{b\ \bdot}}

\newcommand{\gaverage}[1]{\left\langle {\strut #1} \right\rangle_{\mathrm{\Sigma}}}

\newcommand{\ddt}[1]{\frac{{\mathrm d} #1}{{\mathrm d} t}}
\newcommand{\spatiald}{{}^{3 \!} \mathbf{d}} 
\newcommand{\initial}[1]{{#1_{\mathbf{i}}}}

\newcommand{\beq}{\begin{equation}}
\newcommand{\eeq}{\end{equation}}
\newcommand{\beqNo}{\begin{equation*}}
\newcommand{\eeqNo}{\end{equation*}}

\newcounter{thgroupcount}
\newenvironment{thgroup}{%
\setcounter{thgroupcount}{\thetheorem}
\setcounter{theorem}{0}
\refstepcounter{thgroupcount}
\renewcommand{\thetheorem}{\thethgroupcount.\alph{theorem}}
}
{\setcounter{theorem}{\thethgroupcount}}

\newcounter{propgroupcount}

\newcounter{corgroupcount}
\newenvironment{corgroup}{%
\setcounter{corgroupcount}{\thecorollary}
\setcounter{corollary}{0}
\refstepcounter{corgroupcount}
\renewcommand{\thecorollary}{\thecorgroupcount.\alph{corollary}}
}
{\setcounter{corollary}{\thecorgroupcount}}

\numberwithin{equation}{section}

\definecolor{myblue}{rgb}{0.2,0.3,0.7}
\definecolor{darkgreen}{rgb}{0,0.3,0}
\definecolor{mygreen}{rgb}{0,0.5,0}
\definecolor{grey}{rgb}{0.5,0.5,0.5}
\definecolor{orange}{rgb}{1,0.5,0}
\hypersetup{
    colorlinks=true,
    citecolor=blue,
    linkcolor=violet,
    filecolor=magenta,      
    urlcolor=violet,
}

\journalname{GRG}
\begin{document}

\title{On average properties of inhomogeneous fluids in general relativity III: general fluid cosmologies
\thanks{Work supported by ERC advanced Grant 740021--ARTHUS.} 
}

\titlerunning{On average properties of inhomogeneous fluids in general relativity III} 

\author{Thomas~Buchert$\;\boldsymbol{\cdot{}}\;$Pierre~Mourier$\;\boldsymbol{\cdot{}}\;$Xavier~Roy}

\authorrunning{Buchert \textit{et al.}} 

\institute{Thomas Buchert$^1$ \and Pierre Mourier$^{1,2,3}$ \and Xavier Roy$^{1,4}$ \at
$^1$ Univ Lyon, Ens de Lyon, Univ Lyon1, CNRS, Centre de Recherche Astrophysique de Lyon UMR5574, F--69007 Lyon, France \\ 
$^2$ Max Planck Institute for Gravitational Physics (Albert Einstein Institute), Callinstra\ss e 38, D--30167 Hannover, Germany \\
$^3$ Leibniz Universit\"at Hannover, D--30167 Hannover, Germany \\
$^{4}$ Cosmology and Gravity Group, Department of Mathematics and Applied Mathematics, University of Cape Town, 7701 Rondebosch, South Africa \\
\email{buchert@ens-lyon.fr $\cdot$ pierre.mourier@aei.mpg.de $\cdot$ x.roy@gmx.com}
}

\date{Received: 10 December 2019 / Accepted: 8 February 2020}

\maketitle

\begin{abstract}
We investigate effective equations governing the volume expansion of spatially averaged 
portions of inhomogeneous cosmologies in spacetimes filled with an arbitrary fluid. 
This work is a follow-up to previous studies focused on irrotational dust models (Paper~I) and irrotational perfect 
fluids (Paper~II) in flow-orthogonal foliations of spacetime. It complements them by considering arbitrary foliations, arbitrary lapse and shift, and by allowing for a tilted fluid flow with vorticity. As for the first studies, the 
propagation of the spatial averaging domain is chosen to follow the congruence of the fluid, which avoids unphysical 
dependencies in the averaged system that is obtained.
We present two different averaging schemes and corresponding systems of averaged evolution equations providing generalizations of Papers~I and II. The first one retains the averaging operator used in several
other generalizations found in the literature. We extensively discuss relations to these formalisms and pinpoint limitations, in particular regarding rest mass conservation on the averaging domain. The alternative averaging scheme that we subsequently introduce follows the spirit of Papers~I and II and focuses on the fluid flow and the associated $1+3$ threading congruence, used jointly with the $3+1$ foliation that builds the surfaces of averaging. This results in compact averaged equations with a minimal number of cosmological backreaction terms. We highlight that this system becomes especially transparent when applied to a natural class of foliations which have constant fluid proper time slices.
%
%
%
\end{abstract}
\clearpage
\setcounter{tocdepth}{3}
\normalsize\tableofcontents

\clearpage


\section{Introduction}
\label{sec:intro}

A viable cosmological model provides an effective evolution history of the inhomogeneous Universe. 
The procedure of spatially averaging the scalar characteristics of an inhomogeneous model universe yields a system 
of scale-dependent Friedmann-like equations with an effective energy-momentum tensor, featuring so-called backreaction terms 
(see \cite{buchert:av_dust,buchert:av_pf}, respectively referred to as Paper~I and Paper~II hereafter). 
These additional terms contribute to and may potentially replace the dark constituents of the Universe that have to be postulated
as fundamental sources in the standard model of cosmology \cite{buchert:jgrg,buchert:review}.
For recent reviews and references, we direct the attention of the reader to
\cite{ellis:focus,buchert:focus,clarkson:review,kolb:focus,rasanen:focus,buchertrasanen,buchert11}. 

Extensions of this averaging framework to general foliations of spacetime have been investigated 
\cite{futamase:aver1,futamase:aver2,larena:aver,brown:aver,marozzi:aver1,marozzi:aver2,rasanen:lightpropagation,dunsby:aver,smirnov:aver}  
within the $3+1$ formalism or within a four-dimensional approach with spatial averaging slices.
Some limitations and drawbacks can be identified in the formalisms adopted in these papers, and we are going to point them out in specially dedicated sections on the comparison with results in the literature.
The results of the four-dimensional formalism of \cite{marozzi:aver1,marozzi:aver2} turn out to be closest
to the covariant results in $3+1$ form that we develop in the present work (see \cite{asta1} for a modification and extension of the former approach highlighting the relation between both frameworks).

We describe in this paper two general averaging frameworks based on the $3+1$ formalism,
valid for any foliation, for arbitrary lapse and shift, and allowing for a tilted and vortical fluid flow.
We shall emphasize (i) the use of an averaging domain comoving with the congruence of the fluid, and (ii) the Lagrangian point of view,
that has been employed previously, without averaging, for fluids with vorticity \cite{asada:vorticity} and pressure \cite{asada:pressure}. We shall in particular highlight advantages of our second formalism, 
for which we shall additionally demonstrate the interest of (iii) employing the proper volume of the fluid as a volume measure.
The present general investigation is also useful to relax some restricting assumptions of Papers~I and II, to better understand 
the relation to Newtonian averaged cosmologies \cite{buchert:av_Newton}, and to extend the range of applicability of the
effective equations. It will in particular allow for their application to the analysis in terms of spatial averages of relativistic cosmological simulations which make use of non-fluid-orthogonal foliations (see the recent works \cite{macpherson_etal_19,daverio_etal_19}).

The averaged systems that we derive furnish background-free approaches to relativistic cosmologies. Spatial averages can alternatively be interpreted as a \textit{general background cosmology} with a `background' that is not fixed \textit{a priori} \cite{kolb:focus}, but that interacts with the formation of structures.
Fluctuations can then be investigated with respect to the physical average, 
abandoning standard perturbative frameworks where fluctuations
are referred to a fixed reference background, and thus eliminating the need to consider gauge transformations.
A corresponding perturbation scheme that makes structures evolve
on such a physical background has been investigated for dust cosmologies in \cite{roy:generalbackground}.

This paper is organized as follows.
Section~\ref{sec:foliation&fluid} gives a comprehensive outline of the $3+1$ framework and the general fluid content we consider. We here also introduce
the Lagrangian description, the relevance of which we shall emphasize in what follows. We introduce in section~\ref{sec:extrinsic} an averaging framework similar to the commonly used frameworks in relativistic cosmological modeling
(named here \textit{fluid-extrinsic approach}), however with emphasis on a comoving evolution of the averaging domain. We derive the corresponding averaged evolution equations for the domain and comment on the resulting backreaction terms.
We close this section with a detailed discussion of existing results in the literature.
Section~\ref{sec:intrinsic} develops a new perspective on the averaging problem by exposing in detail a 
\textit{fluid-intrinsic approach}. The latter employs a $1+3$ threading of spacetime and focusses on the volume form of the fluid rather than that of the hypersurfaces. This allows for a compact formulation of the effective equations governing regional averages of fluid properties, and it agrees in spirit with what has been presented in Papers~I and II (see \cite{foliationsletter} for an overview and a discussion of foliation dependence of cosmological backreaction).
Section~\ref{sec:discussion_intr} concludes with a discussion of a few cases of interest in order to illustrate
the fluid-intrinsic approach, to compare with Papers~I and II and with the Newtonian framework, and to prepare applications. Perspectives for future studies are also discussed in this section.


\section{Foliation of spacetime and decomposition of the fluid}
\label{sec:foliation&fluid}

We consider in this work a model universe sourced by a single general fluid.
This section sets the definitions and notations for a general $3+1$ spatial foliation of spacetime on the one hand, and for the decomposition 
of the fluid flow and of its energy-momentum tensor with respect to this foliation on the other hand
(see, e.g., \cite{adm,mtw,smarr:foliation,alcub:foliation,gourg:foliation} for more details). The \textit{comoving} and \textit{Lagrangian} pictures are then introduced as natural possible coordinate descriptions adapted to the fluid flow.


\subsection{Description of the geometry} 
\label{subsec:geom}

Our spacetime model is a globally hyperbolic four-dimensional manifold ${\mathcal M}^4$, endowed with the pseudo-Riemannian 
metric tensor $\bf g$ and described by a local system of coordinates $x^\mu = (t, x^i)$,%
\footnote{%
	Greek letters are assigned to spacetime indices, they run in $\{ 0, 1, 2, 3 \}$, and Latin letters refer to 
	space indices, running in $\{ 1, 2, 3 \}$. The signature of the metric is taken as $(- + + +)$, and units 
	are such that $c = 1$. The coordinate system $x^\mu$ is associated to the coordinate vector basis 
	$\{ \bm{\partial}_{x^\mu} \} := \{ \bm{\partial}_t, \bm{\partial}_{x^i} \}$ and its dual exact basis 
	$\{ {\mathbf d} x^\mu \} := \{ {\mathbf d} t, {\mathbf d} x^i \}$. Unless otherwise specified, components of tensorial 
	objects should be understood as expressed in these bases, with arguments $(t, x^i)$.%
	\label{fn:footnote_a}
}
with $\mathbf{g} = g_{\mu \nu} \, \mathbf{d}x^\mu \otimes \mathbf{d}x^\nu$.
The fluid flow is described by a timelike congruence with a unit, future-oriented, timelike tangent $4-$vector field $\bm u$, characterizing its $4-$velocity. The exact physical interpretation of this $4-$velocity may depend on the specific matter model under consideration. It can for instance be defined as an energy frame for the fluid (\textit{i.e.}, an eigenvector of the stress-energy tensor), a barycentric velocity, or the unit vector associated to another conserved current; see \textit{e.g.} \cite[p.72]{ehlers73}, and the related discussion on rest mass density in section \ref{subsubsec:se_fluid} (footnote \ref{fn:restmass}) below. We keep our approach general to this respect so that the most appropriate construction can be chosen for any specific application.

We foliate the spacetime manifold into a family of spacelike hypersurfaces, and we denote by $\bm n$ their unit, timelike, future-oriented normal vector field, which is in general tilted with respect to the $4-$velocity $\bm u$. The foliation can be characterized by a regular scalar function $S$ 
strictly increasing along each flow line, and defined such that each spatial hypersurface is a level set of $S$. 
For simplicity, we choose the time coordinate $t$ as being a strictly increasing function of $S$ (implying the 
reciprocal relation $S = S(t)$), and use it to label the hypersurfaces. The spatial coordinates $x^i$, on the 
other hand, are kept arbitrary.

In such a spacetime coordinate basis, the components of $\bm n$ are written: 
\beq
	n^\mu = \frac{1}{N} \left( 1, - N^i \right) \, ,
	\label{eq:n_vec}
\eeq
and the components of its non-exact dual form $\underline{\bm n}$ read: 
\beq
	n_\mu = - N \, ( 1, 0 ) \, . 
	\label{eq:n_form}
\eeq
The positive lapse function $N$ determines how far consecutive slices are from each other in the slice-orthogonal time direction at each point, while the shift vector 
$\bm N$ generates a spatial diffeomorphism that relates points between successive slices. Following the usual conventions of the $3+1$ formalism, we associate the shift to the coordinate functions defining the propagation of the local spatial coordinates between slices. By definition we have:
\beq
	\bm{\partial}_t = N \bm n + \bm N \, .
\eeq
We shall keep the lapse and shift unspecified for the derivation of the averaged system, thereby preserving the 
four degrees of freedom of the foliation. We shall, however, introduce in section~\ref{subsec:lagrange} convenient foliation and coordinate choices 
that may be adopted for the description of the system (these amount to setting the shift, or both the lapse and the shift).

Spacetime tensors are projected onto the hypersurfaces of the foliation by means of the operator 
$\mathbf{h} = h_{\alpha \beta} \, \mathbf{d} x^\alpha \otimes \mathbf{d} x^\beta$, 
\beq
	h_\mnu := g_\mnu + n_\mu n_\nu \, , \quad 
	h_{\alpha \mu} n^\alpha = 0 \, , \quad 
	h^\mu_{\phantom{\mu} \alpha} h^\alpha_{\phantom{\alpha} \nu} = h^\mu_{\phantom{\mu} \nu} \, , \quad 
	h^\albe h_\albe = 3 \, , 
	\label{eq:proj_h} 
\eeq
whose restriction to the spatial slices defines the spatial Riemannian metric $h_{ij}$, with inverse $h^{i j}$. 
Given this operator and the normal vector $\bm n$, the four-dimensional line element can be decomposed into 
\beq
	{\mathrm d}s^2 
		= g_\albe \, {\mathrm d}x^\alpha {\mathrm d}x^\beta
		= - \left( N^2 - N^k N_k \right) {\mathrm d}t^2 + 2 N_i \, {\mathrm d}x^i \, {\mathrm d}t \,
		  + h_{ij} \, {\mathrm d}x^i {\mathrm d}x^j \, . 
	\label{eq:line_elem}
\eeq
The lapse $N$ also measures, through its spatial variations, the acceleration $\bm a^{(n)}$ of the frames associated 
with $\bm n$:
\beq \label{eq:eulerian_acc}
	a^{(n)}_\mu := n^\alpha \nabla_\alpha n_\mu = \frac{N_{||\mu}}{N} \; ,
\eeq
where $\nabla_\alpha$ denotes the four-covariant derivative, and ${}_{||}$ the three-covariant derivative associated 
with the spatial metric $h_{ij}$.


\subsection{Description of the fluid}
\label{subsec:fluid}

\begin{figure}[htb]
	\center{\includegraphics[scale=1.3]{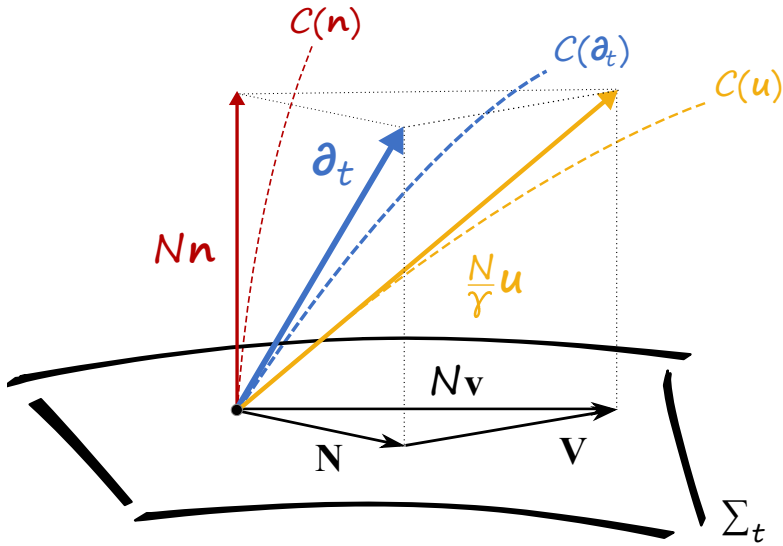}}
	\caption{%
		Representation of the different vector fields at hand, on a spatial hypersurface $\mathrm{\Sigma}_t$.
		$\bm n$ is the vector normal to the hypersurface and it transports the normal frames (defining a congruence 
		$ \CC(\bm n)$); 
		$\bm{\partial}_t$ is the time-vector of the coordinate basis, tangent to the integral curves 
		$\CC(\bm{\partial}_t)$ (with $x^i = const.$); and $\bm u$ is the $4-$velocity of the fluid, tangent to the 
		congruence $\CC(\bm u)$.
		The deviations between $\bm n$ and $\bm{\partial}_t$, on the one hand, and between $\bm{\partial}_t$ 
		and $\bm u$, on the other hand, are identified respectively by $\bm N$ and $\bm V$. The tilt between 
		$\bm u$ and $\bm n$ is given by $\bm v = (\bm N + \bm V) / N$. 
		For a coordinate system comoving with the fluid, we have 
		$\bm V = \bm 0$ and $\bm{\partial}_t = (N / \gamma) \, \bm u$.  Even though the coordinate velocity vanishes 
		in this situation, the fluid can still experience a spatial motion within the hypersurface, given by $\bm v$, and the shift  would be set to $\bm N = N \bm v$.
		Alternatively, in the case of a fluid flow orthogonal to the hypersurfaces, we would have $\bm u = \bm n$, and hence $\bm V = - \bm N$ 
		for any shift.
	}
	\label{fig:schem_vect_3D}
\end{figure}
%


\subsubsection{Decomposition of the $4$-velocity}
\label{subsubsec:u_fluid}

We represent in figure \ref{fig:schem_vect_3D} the different vector fields that we introduce now.
The fluid $4$-velocity vector $\bm u$ can be decomposed in all generality into
\begin{gather}
	{\bm u} = \gamma \left( \bm n + \bm v \right) \, , 
	\label{eq:four_vel} \\ 
	\mathrm{with} \quad 
		n_{\alpha} v^{\alpha}= 0 \, , \qquad 
		\gamma = - n_{\alpha} u^{\alpha} = \frac{1}{\sqrt{1 - v^{\alpha} v_{\alpha}}} \, , 
	\label{eq:lorentz}
\end{gather}
where $\bm v$ (hereafter \textit{Eulerian velocity}) is the spatial velocity of the fluid relative to the 
\textit{normal frames}, which are defined as being locally at rest within the hypersurfaces and transported along 
the normal $\bm n$. The vector $\bm v$ identifies the direction and magnitude of the above-mentioned tilt 
between the normal and the fluid frames. 
The magnitude is equivalently measured by the Lorentz factor $\gamma$ or by the tilt angle $\phi$, defined as 
$\phi := \mathrm{arcosh}(\gamma)$ \cite{ellis:tilted,marozzi:aver2}. 
For a vanishing tilt, $\bm u = \bm n$, we have $\bm v = \bm{0}$, $\gamma = 1$, and $\phi = 0$.

Introducing the spatial \textit{coordinate velocity} of the fluid,
\beq
	\bm V = \frac{{\mathrm d} \bm x}{{\mathrm d}t} \, , \quad 
		\mathrm{with} \quad n_{\alpha} V^{\alpha} = 0 \, , 
	\label{eq:coord_vel}
\eeq
where $\bm x$ is the spatial position of a fluid element in the coordinate system $(t, x^i)$ and 
${\mathrm d}/{\mathrm d}t$ is the derivative with respect to $t$ along the fluid flow lines, we can write 
the Eulerian velocity as (see, e.g., \cite{smarr:foliation,alcub:foliation,gourg:foliation}):
\beq
	\bm v = \frac{1}{N} \left( \bm N + \bm V \right) \, .
	\label{eq:relat_vel}
\eeq
Equation \eqref{eq:four_vel} can then be reformulated in the general form: 
\begin{gather}
	\bm u = \frac{\gamma}{N} \left (N \bm n + \bm N + \bm V \right) \, , 
	\label{eq:four_vel_II} \\
	\mathrm{with} \quad 
		\frac{\gamma}{N} = \frac{1}{\sqrt{N^2 - (N^{\alpha} + V^{\alpha})(N_{\alpha} + V_{\alpha})}} \, . 
	\nonumber 
\end{gather}
In contrast to the Eulerian velocity $\bm v$ which is covariantly defined, the coordinate velocity $\bm V$ depends 
on the way the spatial coordinates propagate between neighboring hypersurfaces; hence it depends on the shift. 
For instance, for a coordinate system comoving with the fluid, which corresponds to a specific shift, we have 
$\bm V = \bm 0$, while for a vanishing tilt, we have $\bm V = - \bm N$, whatever shift is chosen. 

Note that a foliation orthogonal to the fluid, where $\bm n := \bm u$ and $\bm v = \bm 0$ (as considered in Papers~I 
and II), is only possible for a fluid flow with no vorticity. Even for irrotational fluids, introducing a tilt allows us 
to keep the freedom in the construction of the spatial hypersurfaces. 

The components of $\bm u$ and its dual $1-$form $\underline{\bm u}$ are obtained by noticing that any spatial 
vector $\bm \chi$ can be extended to a four-dimensional vector by writing:
\beq
	\chi^\mu = (\chi^0, \chi^i) \, , \quad 
	\mathrm{with} \quad \chi^0 = 0 \, .
	\label{eq:ext_vec}
\eeq
The components of the dual $1-$form of $\bm \chi$ are then deduced from the property 
$n^\alpha \chi_\alpha = 0$ along with expression \eqref{eq:n_vec}:
\beq
	\chi_\mu = (\chi_0, \chi_i) \, , \quad 
	\mathrm{with} \quad \chi_0 = N^k \chi_k \, . 
	\label{eq:ext_form}
\eeq
Applying \eqref{eq:ext_vec} and \eqref{eq:ext_form} to the shift vector and the coordinate velocity, we obtain from 
\eqref{eq:four_vel_II} the component expressions for $\bm u$ and $\underline{\bm u}$: 
\begin{gather}
	u^\mu = \frac{\gamma}{N} \left( 1, V^i \right) \, , 
	\;\; \quad
	u_\mu = \frac{\gamma}{N} \left( - N^2 + N^k \left( N_k + V_k \right), \, N_i + V_i \right) \, . 
	\label{eq:comp_u} 
\end{gather}


\subsubsection{Kinematic variables and acceleration}
\label{subsubsec:kin_fluid}

Let us introduce the operator $\mathbf{b} = b_\albe \, \mathbf{d}x^\alpha \otimes \mathbf{d}x^\beta$ that projects 
tensors onto the local rest frames of the fluid (orthogonal to $\bm u$):
\beq
	b_\mnu := g_\mnu + u_\mu u_\nu 
		\, , \;\;\quad
	b_{\alpha \mu} u^\alpha = 0 
		\, , \;\;\quad 
	b^\mu_{\phantom{\mu} \alpha} b^\alpha_{\phantom{\alpha} \nu} = b^\mu_{\phantom{\mu} \nu} 
		\, , \;\;\quad 
	b^{\alpha \beta} b_{\alpha \beta} = 3 
		\, . 
	\label{eq:proj_f}
\eeq
The projectors $\mathbf{b}$ and $\mathbf{h}$ usually differ because of the tilt between $\bm u$ and $\bm n$. From relations \eqref{eq:proj_f}, we can decompose the $4-$covariant derivative of the $1-$form 
$\underline{\bm u}$ into the $4-$acceleration and the kinematic parts of the fluid \cite{ehlers} as follows:
\begin{gather}
	\nabla_\mu u_\nu 
	= - u_\mu \, a_\nu + \frac{1}{3} \Theta b_\mnu + \sigma_\mnu + \omega_\mnu \, , 
	\label{eq:cov_u} \\ 
	\mathrm{with} \quad 
	a_\mu := u^\alpha \nabla_\alpha u_\mu \, , \quad \; 
	\Theta := \nabla_\alpha u^\alpha \, , \nonumber \\
	\mathrm{and} \quad 
	\sigma_\mnu := b^\alpha_{\phantom{\alpha} \mu} b^\beta_{\phantom{\beta} \nu}  \nabla_{( \alpha} u_{\beta )} 
	- \frac{1}{3} \Theta b_\mnu \, , \quad
	\omega_\mnu := b^\alpha_{\phantom{\alpha} \mu} b^\beta_{\phantom{\beta} \nu} \nabla_{[ \alpha} u_{\beta ]} \, ,
	\label{eq:kin_fluid}
\end{gather}
where the round and square brackets respectively imply symmetrization and anti-symmetrization over the indices enclosed. 
$\bm a$ is the acceleration of the fluid, $\Theta$ its expansion rate,  $\bm \sigma$ its shear tensor, and $\bm \omega$ is 
its vorticity tensor.\footnote{%
	The shear, vorticity and acceleration of the fluid, as seen in the normal frames, can be derived from 
	the projections onto the three-surfaces of the proper shear $\bm \sigma$, proper vorticity $\bm \omega$ and 
	proper acceleration $\bm a$, respectively. For instance, the second would read 
	$h^\alpha_{\phantom{\alpha} \mu} h^\beta_{\phantom{\beta} \nu} \, \omega_\albe %
	= h^\alpha_{\phantom{\alpha} \mu} h^\beta_{\phantom{\beta} \nu} \, %
	b^\delta_{\phantom{\delta} \alpha} b^\xi_{\phantom{\xi} \beta} \nabla_{[ \delta} u_{\xi ]}$,  
	which differs from 
	$h^\alpha_{\phantom{\alpha} \mu} h^\beta_{\phantom{\beta} \nu} \nabla_{[ \alpha} u_{\beta ]}$ when 
	$\bm a$ is not null.%
}
Recall that the rest frames of the fluid are not hyper\-surface-forming if $\bm \omega$ does not vanish. From $\bm \sigma$ and $\bm \omega$, respectively, one can additionally define the squared \textit{rate of shear} $\sigma^2 \equiv \sigma^{\mu \nu} \sigma_{\mu \nu} / 2$ and the squared \textit{rate of vorticity} $\omega^2 \equiv \omega^{\mu \nu} \omega_{\mu \nu} / 2$, which are both positive-definite.


\subsubsection{Stress-energy tensor and conservation laws}
\label{subsubsec:se_fluid}

The stress-energy tensor of the fluid can be decomposed with respect to the fluid rest frames as follows: 
\begin{gather}
	T_\mnu = \epsilon \, u_\mu u_\nu + 2 \, q_{(\mu} u_{\nu)} + p \, b_\mnu + \pi_\mnu \, , \label{eq:se_fluid_u} \\ 
	\mathrm{with} \; 
		\epsilon := u^\alpha u^\beta T_\albe 
			\, , \;
		q_\mu := - b^\alpha_{\phantom{\alpha} \mu} u^\beta T_\albe 
			\, , \;
		p \, b_\mnu + \pi_\mnu := b^\alpha_{\phantom{\alpha} \mu} b^\beta_{\phantom{\beta} \nu} T_\albe \, , \; b^{\mu \nu} \pi_{\mu \nu} = 0
			\, .
	\nonumber 
\end{gather}
$\epsilon$ denotes the energy density of the fluid in its rest frames, $\bm q$ the spatial heat vector, 
$p$ the isotropic pressure, and $\pi_\mnu$ the spatial and traceless anisotropic stress. 
Alternatively, it can be decomposed with respect to the normal frames as 
\begin{gather}
	T_\mnu = E \, n_\mu n_\nu + 2 \, n_{( \mu} J_{\nu )} + S_\mnu \, , \label{eq:se_fluid_n} \\ 
	\mathrm{with} \quad 
		E := n^\alpha n^\beta T_\albe 
			\, , \qquad 
		J_\mu := - h^\alpha_{\phantom{\alpha} \mu} n^\beta T_\albe 
			\, , \qquad 
		S_\mnu := h^\alpha_{\phantom{\alpha} \mu} h^\beta_{\phantom{\beta} \nu} T_\albe 
				\, ,
	\nonumber 
\end{gather}
where $E$ is the energy density of the fluid, $J_\mu$ its momentum density, and $S_\mnu$ its stress density,
all as measured in the normal frames. The isotropic part of $S_{\mu \nu}$ is given by (one third of) the trace $S := g^\albe S_\albe$. 
This second decomposition will be used in section \ref{sec:extrinsic} for the derivation of the averaged equations in the fluid-extrinsic approach.
Using expression \eqref{eq:four_vel}, we can relate the scalar quantities of both decompositions as
\begin{gather}
	E = \gamma^2 \epsilon + ( \gamma^2 - 1 ) \, p 
		+ 2 \, \gamma v^\alpha q_\alpha + v^\alpha v^\beta \pi_\albe \;\, , \label{eq:rel_E} \\ 
	S = ( \gamma^2 - 1 ) \, \epsilon + ( \gamma^2 + 2 ) \, p 
		+ 2 \, \gamma v^\alpha q_\alpha + v^\alpha v^\beta \pi_\albe \;\, . \label{eq:rel_S}
\end{gather}
From the property $\nabla_\beta T^\albe = 0$ along with relations \eqref{eq:se_fluid_u} and \eqref{eq:kin_fluid}, 
we derive the energy conservation law: 
\beq
	u_\alpha \nabla_\beta T^\albe = 0 
		\quad \Leftrightarrow \quad 
	\dot\epsilon + \Theta \left( \epsilon + p \right) 
		= - a_\alpha q^\alpha - \nabla_\alpha q^\alpha - \pi^\albe \sigma_\albe \, , 
	\label{eq:en_cl}
\eeq
and the momentum conservation law:
\begin{gather}
	b_{\mu \alpha} \nabla_\beta T^\albe = 0 \nonumber \\ 
		\Leftrightarrow \quad 
	a_\mu = - \frac{1}{\epsilon + p} 
		\left( b^\alpha_{\phantom{\alpha} \mu} \nabla_\alpha p 
		+ b_{\mu \alpha} \dot{q}^\alpha 
		+ \frac{4}{3} \Theta q_\mu + q^\alpha (\sigma_{\alpha \mu} + \omega_{\alpha \mu}) 
		+ b_{\mu \alpha} \nabla_\beta \pi^\albe \right) \, ,
	\label{eq:mom_cl}
\end{gather}
where the overdot is defined below in section~\ref{subsec:derivatives}. These relations can be complemented by the conservation of the \textit{rest mass density} $\varrho$ of the fluid in its rest frame:\footnote{%
For massive particles, the rest mass density $\varrho$ is defined naturally and the associated definition of $\bm u$ as the barycentric velocity ensures the conservation of the mass current $\varrho \, \bm u$. For massless particles such as a pure photon gas, a conserved \textit{number density} current can be constructed instead in general, defining a $4-$velocity $\bm u$ as its direction and $\varrho$, now a particle number density, as its norm; similar definitions from other conserved currents (such as the baryon current density) are possible, e.g., when nuclear reactions occur and the actual rest mass is not locally conserved \cite[sec. 3.1]{ellis71}, \cite[p.72]{ehlers73}. Any of these definitions may be used in applications as appropriate for the specific matter model at hand, although we will, for simplicity, refer to $\varrho$ as the `rest mass density' and its volume integral as the `rest mass' in what follows. Note that for a general fluid model, a definition of $\bm u$ based on the energy frame may not always be compatible with such constructions based on conserved currents.
\label{fn:restmass}
}%
\beq
	\nabla_\alpha (\varrho u^\alpha) = 0
		\, , \quad \text{or equivalently,} \quad 
	\dot\varrho + \Theta \varrho = 0 \, .
	\label{eq:cons_restmass_density}
\eeq


\subsection{Time derivatives and their relations}
\label{subsec:derivatives}

The existence of two different times (the coordinate time $t$ and the fluid proper time $\tau$) and of three timelike 
congruences (see figure~\ref{fig:schem_vect_3D}) leads to several possible definitions of time derivatives. Those of main interest for the present work are:
\begin{itemize}
	\item[$\bullet$] the \textit{covariant derivative} along the fluid flow lines, denoted by an overdot;
	for any tensor field $\bm{F}$, we have $\dot{\bm{F}} := u^\alpha \nabla_\alpha \bm{F}$;
	\item[$\bullet$] the comoving derivative along the fluid flow lines and according to the proper time $\tau$, or \textit{Lagrangian derivative},
	denoted by ${\mathrm d}/{\mathrm d}\tau$;
	\item[$\bullet$] the comoving derivative along the fluid flow lines and according to the coordinate time $t$, denoted by  $\mathrm{d}/\mathrm{d}t$; 
	\item[$\bullet$] the partial coordinate time derivative along the vector $\bm{\partial}_t$ (see footnote~\ref{fn:footnote_a}), 
	\textit{i.e.} along the integral curves of constant $x^i$, denoted by $\partial_t\big|_{x^i}$.
\end{itemize}
The last three derivatives are related by:
\begin{gather}
	\frac{{\mathrm d} F^{\mu \nu ...}_{\ \ \alpha \beta ...}}{{\mathrm d}t} 
		= \frac{\partial F^{\mu \nu ...}_{\ \ \alpha \beta ...}}{\partial t} \bigg|_{X^i} 
		\! = \frac{\partial F^{\mu \nu ...}_{\ \ \alpha \beta ...}}{\partial t} \bigg|_{x^i} \! + V^i \frac{\partial F^{\mu \nu ...}_{\ \ \alpha \beta ...}}{\partial x^i} \, , 
	\label{eq:der_rel_1} 
	\\[3pt]
	\frac{{\mathrm d} F^{\mu \nu ...}_{\ \ \alpha \beta ...}}{{\mathrm d}\tau} 
		= \frac{\gamma}{N} \frac{{\mathrm d} F^{\mu \nu ...}_{\ \ \alpha \beta ...}}{{\mathrm d}t} \, , 
	\label{eq:der_rel_2}
\end{gather}
for any tensor field %
$\bm F = %
	F^{\mu \nu \ldots}_{\;\; \alpha \beta \ldots} %
	\bm{\partial}_\mu \otimes \bm{\partial}_\nu \otimes \ldots \otimes %
	\boldsymbol{\mathrm{d}} x^\alpha \otimes \boldsymbol{\mathrm{d}} x^\beta \otimes \ldots$
as decomposed in the coordinate system $(t,x^i)$ with arbitrary spatial coordinates $x^i$ (see footnote \ref{fn:footnote_a}). $X^i$, on the other hand, specifically denotes a set of spatial coordinates that are comoving with the fluid flow, that is, these coordinates represent labels for the fluid elements.
For a scalar field $\psi$, the first two derivatives are identical: 
$\dot\psi = u^\alpha \partial_\alpha \psi = {\mathrm d}\psi / {\mathrm d}\tau$.

\begin{proof}
\small{%
	Let us consider the components $F^{\mu \nu \ldots}_{\;\; \alpha \beta \ldots}$ of a tensor field $\bm F$ in 
	the coordinate basis associated with $(t, x^i)$ (see footnote \ref{fn:footnote_a}). For notational ease, 
	we drop in what follows the indices and write $F := F^{\mu \nu \ldots}_{\;\; \alpha \beta \ldots}$. The total 
	coordinate-time derivative of $F$ along any timelike curve $\CC$ can be decomposed in terms of the coordinate 
	partial derivatives as 
	\beq
		\frac{{\mathrm d} F}{{\mathrm d}t} \bigg|_{\CC} 
		= \frac{\partial F}{\partial t} \bigg|_{x^i}
		+ \frac{\partial F}{\partial x^i} \, \frac{{\mathrm d}x^i}{{\mathrm d}t} \bigg|_{\CC} \, .
	\eeq
	Considering the variation along the congruence $\CC(\bm u)$ of the fluid, and therefore making use of 
	definition \eqref{eq:coord_vel}, we obtain:
	\beq
		\ddt{F} 
			:= \frac{{\mathrm d} F}{{\mathrm d}t} \bigg|_{\CC(\bm u)} 
			= \frac{\partial F}{\partial t} \bigg|_{x^i} + \frac{\partial F}{\partial x^i} V^i \, . 
		\label{eq:der_rel_proof}
	\eeq
	Moreover, for the Lagrangian coordinates $X^i$, by definition constant 
	along the fluid flow lines, we have ${(\mathrm d}X^i/{\mathrm d}t) \, |_{\CC(\bm u)} = 0$, and hence 
	$\mathrm{d}F/\mathrm{d}t = \partial_t |_{X^i} F$, which concludes the proof of \eqref{eq:der_rel_1}.
	
	The total derivative of $F$ with respect to the proper time $\tau$ of the fluid along the congruence 
	$\CC(\bm u)$ satisfies 
	\beq
		\frac{\mathrm{d}F}{\mathrm{d}\tau} 
			:= \frac{{\mathrm d} F}{{\mathrm d}\tau} \bigg|_{\CC(\bm u)} 
			= \frac{\mathrm{d}t}{\mathrm{d}\tau} \bigg|_{\CC(\bm u)} \ddt{F} \bigg|_{\CC(\bm u)} \, .
	\eeq
	From the definition of $\bm u$ and its component expression \eqref{eq:comp_u}, we have 
	$(\mathrm{d}t / \mathrm{d}\tau) \, |_{\CC(\bm u)}=u^0 = \gamma/N$, and thus 
	\beq
		\frac{\mathrm{d}F}{\mathrm{d}\tau} = \frac{\gamma}{N} \ddt{F} \, ,
	\eeq
	which proves \eqref{eq:der_rel_2}. Reformulating the right-hand side by means of \eqref{eq:der_rel_proof},
	and using again the component expression of $\bm u$ finally yields:
	\beq
		\frac{{\mathrm d} F}{{\mathrm d}\tau} 
		= u^0 \frac{\partial F}{\partial t} \bigg|_{x^i} + u^i \frac{\partial F}{\partial x^i} 
			= u^\alpha \, \partial_{x^\alpha} F \, ,
	\eeq
	hence ${\mathrm d}/{\mathrm d}\tau = u^\alpha \partial_\alpha$. This operator coincides with the overdot, 
	$\dot{} = u^\alpha \nabla_\alpha$, when applied to a scalar variable. 
	$\square$
}%
\end{proof}


\subsection{Comoving and Lagrangian descriptions}
\label{subsec:lagrange}


\subsubsection{Comoving description}
\label{subsubsec:comov_desc}

For any given foliation, the shift vector can be chosen in such a way that the spatial components \eqref{eq:comp_u} 
of $\bm u$ vanish: by setting $\bm N = N \bm v$, given relation \eqref{eq:relat_vel}, we have $\bm V = \bm 0$. 
This choice corresponds to spatial coordinates propagating along the fluid flow lines, i.e.\ to \textit{comoving} 
(or \textit{Lagrangian}) spatial coordinates. We will refer to the use of these spatial coordinates as a  
\textit{comoving description} of the fluid, and denote them by $X^i$. Note that a comoving description is a 
``weak'' form of a \textit{Lagrangian description} (as introduced below) in that no constraints are set on the time coordinate $t$. 

In the coordinates $(t, X^i)$ of the comoving description, the components \eqref{eq:comp_u} of the fluid velocity read:
\beq
	u^\mu = \frac{\gamma}{N} \, (1, 0) \, , 
		\qquad 
	u_\mu = \left( - \frac{N}{\gamma}, \gamma v_i \right) \, , 
	\label{eq:fluid_wLag}
\eeq
while the line element \eqref{eq:line_elem} reduces to
\beq
	{\mathrm d}s^2 = 
		- \frac{N^2}{\gamma^2} \, {\mathrm d}t^2 + 2 N v_i  \, {\mathrm d}t \,{\mathrm d}X^i  \,
		+ h_{ij} \, {\mathrm d}X^i {\mathrm d}X^j \, .
	\label{eq:line_elem_wLag}
\eeq
The components of the acceleration and kinematic quantities simplify as follows. 
From the anti-symmetric part of \eqref{eq:cov_u} we can write in any coordinate system:
\beq
	\omega_{\mu \nu} 
		= u_{[\mu} a_{\nu]} + \nabla_{[\mu} u_{\nu]}
		= u_{[\mu} a_{\nu]} + \partial_{[\mu} u_{\nu]} \; .
	\label{eq:vorticity_components}
\eeq
In comoving coordinates, the $(0, i)$ components of this expression vanish, given that $\omega_{\alpha i} u^\alpha = 0$. Combining this property with $a_0 = 0$, from $a_\alpha u^\alpha = 0$, we can thus write the spatial components of the acceleration as 
\beq
	a_i = \frac{\gamma}{N} \left(\ddt{}u_i + \frac{\gamma}{N} \partial_i \left(\frac{N}{\gamma}\right) \right) \; ,
	\label{eq:acc_component_comoving}
\eeq
where we also used $u_0 = - N/\gamma$ and $\partial_t = \partial_t\big|_{X^i} = \mathrm{d}/\mathrm{d}t$. 
Inserting \eqref{eq:acc_component_comoving} back into the $(i,j)$ components of \eqref{eq:vorticity_components} 
yields the non-vanishing components of the vorticity:
\beq
	 \omega_{i j} = 
		 \frac{\gamma}{N} u_{[i} \ddt{}u_{j]}
		 + \frac{N}{\gamma} \partial_{[i} \left(\frac{\gamma}{N} u_{j]} \right) \; .
\eeq
The expansion tensor can be related to the Lie derivative $\CL_{\bm u} \mathbf b$ of the projector $\mathbf b$ 
along the fluid flow in any coordinates according to 
\beq
	\left(\CL_{\bm u} \mathbf b \right)_{\mu \nu} 
		= u^\alpha \nabla_\alpha b_{\mu \nu} + b_{\alpha \nu} \nabla_\mu u^\alpha + b_{\mu \alpha} \nabla_\nu u^\alpha 
		= 2 \, u_{(\mu} a_{\nu)} + 2 \, \nabla_{(\mu} u_{\nu)} 
		= 2 \, \Theta_{\mu \nu} \; ,
	\label{eq:expansion_components}
\eeq
where we have used the symmetric part of \eqref{eq:cov_u} for the last equality. The covariant derivatives of 
the second expression can be equivalently replaced by partial derivatives. This provides the non-vanishing comoving-coordinates components of the expansion tensor as %
$\Theta_{ij} %
	= \left(\CL_{\bm u} \mathbf b \right)_{ij} / \, 2 %
	= u^0 \partial_0 b_{ij} / \, 2$, and hence 
\beq
	\Theta_{ij} = \frac{1}{2} \frac{\gamma}{N} \ddt{}b_{i j} \, .
\eeq
The trace and traceless parts are deduced from the above. For convenience, we express them in terms of a 
representative length $\ell$ in the fluid rest frames, defined by $\dot\ell / \ell := \Theta / 3$ \cite{ehlers}: 
\beq
	\Theta 
		= \frac{1}{2} \frac{\gamma}{N} \, b^{ij} \ddt{} b_{ij} {} 
		= \frac{3}{\ell} \frac{\gamma}{N} \ddt{\ell}  \, ; \qquad 
	\sigma_{i j} 
		= \frac{1}{2} \frac{\gamma}{N} \, \ell^2 \ddt{} (\ell^{-2} b_{i j}) \, .
	\label{eq:exp_tensor_comoving}
\eeq
%

\subsubsection{Lagrangian description}
\label{subsubsec:lag_desc}

An appropriate choice of foliation can be introduced that allows for the labelling of the hypersurfaces by a proper time $\tau$ of the fluid \cite{ellis:vorticity,ellis:relatcosmo}. Such a construction identifies a class of foliations which we call \textit{fluid proper time foliations}. It is realized by the level sets of the fluid proper time $\tau$, as defined
from its comoving coordinate-time evolution rate $\mathrm{d}\tau / \mathrm{d}t = N / \gamma$ (see section \ref{subsec:derivatives}) and an 
initial spacelike hypersurface $\Gamma$ (parametrized by an equation $t = t_\Gamma(X^i)$) on which it takes a given constant value 
$\tau_{\mathbf i}$,\footnote{%
	The proper time is not uniquely defined \textit{a priori}, but it is fully determined by the choice of an initial Cauchy surface to build one of its level 
	sets \cite[p.74]{ellis:relatcosmo}. Another proper time function $\tau'$, taking the constant value 
	$\tau'_{\mathbf i}$ on another initial hypersurface $\Gamma'$, would differ from $\tau$ by a function $\varphi$ constant along the fluid flow lines, 
	$\tau'(t,X^i) = \tau(t,X^i) + \varphi(X^i)$. This relation follows by writing
	\begin{equation*}
		\tau' := \tau'_{\mathbf i} + 
			\int_{t_{\Gamma'}(X^i)}^{t} \frac{N(\hat t, X^i )}{\gamma(\hat t, X^i )} 
			\; {\mathrm d}\hat t \, , 
	\end{equation*}
	with $\Gamma'$ parametrized by $t = t_{\Gamma'}(X^i)$, yielding
	\begin{equation*}
		\varphi(X^i) = \tau'_{\mathbf i} - \tau_{\mathbf i} + \int_{t_{\Gamma'}(X^i)}^{t_\Gamma(X^i)} \frac{N(\hat t, X^i )}{\gamma(\hat t, X^i )} \; 
			{\mathrm d}\hat t \, .
	\end{equation*}
	The expressions defining $\tau$ and $\tau'$ are here given in terms of comoving coordinates. They could alternatively be written covariantly, 
	by setting the value of $\tau - \tau_{\mathbf i}$ (resp.\ $\tau' - \tau'_{\mathbf i}$) at a given spacetime event as the total length of the unique 
	fluid flow line joining this event to the hypersurface $\Gamma$ (resp.\ $\Gamma'$). The properties of both proper times and their relation through 
	$\varphi$ of course hold in this description.%
}
\beq
	\label{eq:def_tau}
	\tau(t,X^i) := 
		\tau_{\mathbf i} + 
		\int_{t_\Gamma(X^i)}^{t} \frac{N(\hat t, X^i )}{\gamma(\hat t, X^i )} \; {\mathrm d}\hat t \; .
\eeq
The hypersurface labelled by a given value $\tau$ can equivalently be defined as the image at time $\tau-\tau_{\mathbf{i}}$ of $\Gamma$ by the flow operator defined from the unitary vector field $\bm u$.

The choice of one fluid proper time foliation sets the normal vector $\bm n$ and relates the lapse $N$ to the Lorentz factor, $N \propto \gamma$ with a purely time-dependent proportionality factor.
The fluid proper time can then be used as the time parameter $t$ labelling these hypersurfaces, $t := \tau$, which implies $N=\gamma$. Note that for such foliations, the hypersurfaces 
cannot be fluid-orthogonal, namely a tilt must be present, except in the case of irrotational geodesic flows (\textit{e.g.}, irrotational dust) \cite{ehlers}. In general, such a tilt may be expected to grow with time and become large and highly inhomogeneous on the slices. This may even imply in some cases that not all slices remain everywhere spacelike; hence, when using such a foliation, we shall implicitly restrict our attention to the part of spacetime where the hypersurfaces do remain spatial, if necessary. Within this class of foliation and lapse choices, 
the additional requirement of using comoving spatial coordinates defines a \textit{comoving and synchronous} picture which we call the
\textit{Lagrangian description} of the fluid (see Asada and Kasai~\cite{asada:vorticity} and Asada~\cite{asada:pressure}, inspired by Friedrich~\cite{friedrich}).

In the coordinates $(\tau, X^i)$ of the Lagrangian description, the components \eqref{eq:comp_u} of the 
$4-$velocity and its dual read:
\beq
	u^\mu = (1, 0) \, , 
	\qquad 
	u_\mu = (-1, \gamma v_i) \, , 
	\label{eq:fluid_sLag} 
\eeq
while the line element \eqref{eq:line_elem} takes the form:
\beq
	{\mathrm d}s^2 = 
		- {\mathrm d}t^2 + 2 \gamma v_i \, {\mathrm d}X^i {\mathrm d}t \, 
		+ h_{ij} {\mathrm d}X^i {\mathrm d}X^j \, . 
	\label{eq:line_elem_sLag}
\eeq
The Lagrangian condition $u^\mu = \delta^{\mu} {}_{0}$, as introduced in \cite{friedrich}, is therefore 
equivalent to setting simultaneously $\bm N = N \bm v$ and $N = \gamma$. It implies $g_{00} = -1$ or, equivalently, $N^2 - N^k N_k = 1$. 
In this description, as a special case of a comoving description (with the additional requirement of $N=\gamma$), the nonvanishing components of the fluid acceleration reduce to 
\beq
	a_i = \frac{\mathrm d}{{\mathrm d}\tau} u_i \; , 
\eeq
and those of the kinematic variables become: 
\begin{gather}
	\Theta_{i j} = \frac{1}{2} \frac{\mathrm d}{{\mathrm d}\tau} b_{i j} \; ; \quad \;
	\Theta = \frac{1}{2}\, b^{k l} \frac{\mathrm d}{{\mathrm d}\tau} b_{k l} {} = \frac{3}{\ell} \frac{\mathrm{d}\ell}{\mathrm{d}\tau} \; ; \qquad 
	\sigma_{i j} = \frac{1}{2} \, \ell^2 \frac{\mathrm{d}}{\mathrm{d}\tau} (\ell^{-2} b_{i j}) \; ; \nonumber 
	\\
	\omega_{i j} = 
		u_{[i} \frac{\mathrm d}{{\mathrm d}\tau} u_{j]} + \partial_{[i} u_{j]} \; . \qquad
\end{gather}
In the following derivations of the extrinsic and intrinsic averaging schemes, we will keep the lapse and 
shift unspecified, thereby considering a general description and preserving the four degrees of freedom of 
the foliation. The formulation of the averaged system in the Lagrangian description will be discussed 
later on (see section \ref{subsec:lagrangian_form}) as a particularly insightful special case within the intrinsic scheme. 


\section{Rest mass--preserving scalar averaging: fluid-extrinsic approach}
\label{sec:extrinsic}

In this section we recall the $3+1$ formulation of the Einstein equations with respect to the hypersurfaces of 
normal $\bm n$, and we present a formalism for spatial averaging over a compact domain that lies within the spatial hypersurfaces and that follows the fluid flow, based on the hypersurface volume measure. We then derive 
the corresponding commutation rule and averaged equations for the scalar parts of the Einstein equations, and we discuss some properties of the resulting backreaction terms and their relation to boundary terms. 
At the end of the section, we compare our approach and its results to previous proposals of generalization  of the framework of Papers~I and II that can be found in the literature, discussing in detail the differences and pinpointing
limitations.


\subsection{Dynamical equations}
\label{subsec:3+1}

The $3+1$ decomposition of the Einstein equations with respect to an arbitrary spatial foliation as described above, 
with the cosmological constant $\cc$ included, comprises the following evolution equations \cite{adm,mtw,smarr:foliation,alcub:foliation,gourg:foliation}: 
\begin{align}
	\partial_t\big|_{x^i} \, h_{ij} = 
		& - 2 N \CK_{ij} + N_{i || j} + N_{j || i} \, , 
		\label{eq:evol_h} \\
	\partial_t\big|_{x^i} \, \CK^i_{\phantom{i} j} = 
		& \; N \, \Big( \CR^i_{\phantom{i} j} + \CK \CK^i_{\phantom{i} j} 
			+ 4 \pi G \, \big[ \left( S - E \right) \delta^i_{\phantom{i} j} - 2 \, S^i_{\phantom{i} j} \big]
				- \cc \,\delta^i_{\phantom{i} j} \Big) \nonumber \\ 
		& - N^{|| i}_{\ \; || j} + N^k \CK^i_{\phantom{i} j || k} + \CK^i_{\phantom{i} k}  N^k_{\ || j}
			- \CK^k_{\phantom{k} j}  N^i_{\ || k} \, , 
		\label{eq:evol_K}
\end{align}
together with the momentum and energy constraints: 
\begin{align}
	 \CK^k_{\phantom{k} i || k} -  \CK_{|| i} 
		& = 8 \pi G \, J_i \, , 
		\label{eq:moment_const} \\ 
	\CR + \CK^2 - \CK^i_{\phantom{i} j} \CK^j_{\phantom{j} i} 
		& = 16 \pi G \, E + 2 \cc \; . 
		\label{eq:hamilt_const}
\end{align}
$\CR_{ij}$ and $\CK_{ij} \! := \! - h^\alpha_{\phantom{\alpha} i} \, h^\beta_{\phantom{\beta} j} \nabla_\alpha n_\beta$ 
are the components of the $3-$Ricci tensor and the extrinsic curvature of the hypersurfaces, respectively.
$\CR := h^{ij} \CR_{ij}$ and $\CK := h^{ij} \CK_{ij}$ are their respective traces.

In Appendix \ref{app:3+1_lag} we give the evolution equations for $h_{ij}$ and $\CK^i_{\phantom{i} j}$ along the 
congruence of the fluid, using the derivative ${\mathrm d}/{\mathrm d}t$ instead of $\partial_t\big|_{x^i}$, 
and we specify their expressions in the comoving and Lagrangian descriptions. 


\subsection{Fluid-extrinsic scalar averaging}
\label{subsec:av_proc}


\subsubsection{Comoving-to-reference map}
\label{subsubsec:map}

We introduce a set of Lagrangian (or comoving) spatial coordinates $\bm X = \{ X^i \}$. 
These are chosen so as to coincide with the arbitrary reference spatial coordinates $\bm x = \{ x^i \}$ for all fluid elements on one constant-$t$ hypersurface (corresponding to a time coordinate value $t = t_{\mathbf{i}}$) which will be used to set initial conditions. The comoving coordinates of each fluid element subsequently remain constant 
along its flow line, as opposed in general to its reference coordinates $\bm x$. This arises from the different directions between the 
threading congruence of the fluid $(t, X^i = const.)$, given by $\bm u$, and the arbitrary coordinate congruence $(t, x^i = const.)$, given by $\bm \partial_t$. 	
The two sets of spatial coordinates $\bm x$ and $\bm X$ are related by a family of diffeomorphisms parametrized by the coordinate time $t$, 
\begin{align}
	& {\bf \Phi}_t \; : \;
		\CD_{\bm X} \;  \to  \; \CD_{\bm x} = {\bf \Phi}_t(\CD_{\bm X}) \, , 
		\nonumber \\ 
	 & \qquad \quad \bm X \;  \mapsto  \; \bm x = {\bf \Phi}_t(\bm X) := \bm f(t,\bm X) \, , 
		\label{eq:map} \\ 
	 & \mathrm{with} \quad 
		\bm f(t_{\mathbf i},\bm X) = \bm X  \, , \;\;
		\text{and with a Jacobian} \;\; J (t, \bm X) := \det{\frac{\partial \bm f(t,\bm X)}{\partial \bm X}} \, ,
		\label{eq:prop_map}
\end{align}
where $\CD$ refers to a compact domain lying within the hypersurfaces and being transported along the congruence of the fluid flow 
(hereafter \textit{comoving domain}). This specific transport ensures that the domain encloses the same collection 
of fluid elements at all times (an important feature to which we shall come back in the discussion). We denote the set of spatial coordinate values corresponding to this 
collection at a given time $t$ by $\CD_{\bm x}(t)$, or $\CD_{\bm x}$ for short, in the reference coordinates, and by $\CD_{\bm X}$ (by definition time-independent) in the comoving coordinates. 
The maps ${\bf \Phi}_t$ define on each constant-$t$ hypersurface a coordinate transformation between $\bm x$ and $\bm X$. Note that we assume throughout the regularity of the fluid flow implied by the existence
of congruences and invertible maps (diffeomorphisms), which excludes the description of caustics
that may occur for particular matter models.

From \eqref{eq:map} we reformulate the coordinate velocity \eqref{eq:coord_vel} as 
\beq
	\bm V = \frac{{\mathrm d} \bm x}{{\mathrm d}t} = \frac{{\mathrm d}}{{\mathrm d}t}  \bm f(t,\bm X)
		= \partial_t\big|_{X^i} \, \bm f(t,\bm X) \, , 
	\label{eq:coord_vel_II}
\eeq
where along the direction of the derivative $\partial_t\big|_{X^i}$, given by the fluid flow lines, the comoving spatial 
coordinates $X^i$ are kept fixed. Using \eqref{eq:prop_map} together with \eqref{eq:coord_vel_II} we have the identity: 
\beq
	\partial_t\big|_{X^i} \, J = J \, \partial_i V^i \, ,
	\label{eq:der_jacob}
\eeq
where the indices in $\partial_i$ and in the components $V^i$ refer to the coordinates $x^i$.
We remark that the Newtonian tools developed for the Lagrangian description of structure formation in cosmology can be 
applied to this diffeomorphism without difficulty (see \cite{buchert:lagrange}, \cite{ehlersbuchert} and references therein).


\subsubsection{Volume of a domain and its time-evolution}
\label{subsubsec:vol_lagrange}

The Riemannian volume of the spatial domain $\CD$ within the hypersurfaces (referred to in the following as the \textit{hypersurface volume} or \textit{extrinsic volume} of the domain) is given by
\beq
	\CV_\CD (t) := \int_{\CD_{\bm x}} n^\mu {\mathrm d} \sigma_\mu
		= \int_{\CD_{\bm x}} \sqrt{h (t, x^i)} \, {\mathrm d}^3 x \, , 
	\label{eq:vol_D}
\eeq
where $h$ is the determinant of the spatial metric, $h := \det(h_{ij})$, and ${\mathrm d}\sigma_\mu$ is the oriented spatial volume element on the slices,
${\mathrm d}\sigma_\mu := - n_\mu \sqrt{h} \, {\mathrm d}^3 x$. 
We seek the coordinate-time variation of \eqref{eq:vol_D} along the fluid flow lines, namely we search 
for the expression of 
\beq
	\frac{\mathrm d}{{\mathrm d}t} \int_{\CD_{\bm x}} \sqrt{h (t, x^i)} \, {\mathrm d}^3 x \, . 
	\label{eq:vol_D_var}
\eeq
The operators ${\mathrm d}/{\mathrm d}t$ and $\int_{\CD_{\bm x}} \! \cdot \; {\mathrm d}^3x$ do not commute 
in general since the endpoints of the integral, determined by the spatial region $\CD_{\bm x}$, depend themselves 
on time. The fluid is moving with respect to the coordinate system $(t, x^i)$, and the domain of integration is 
attached to the fluid\footnote{%
	For the same reason, the operators ${\mathrm d}/{\mathrm d}\tau$ and $\partial_t |_{x^i}$ do not commute 
	either with $\int_{\CD_{\bm x}} \! \cdot \; {\mathrm d}^3x$.%
} (see figure~\ref{fig:var_domain}). We need to reformulate the integrand to get rid of this time-dependence. 

\begin{figure}[htb]
	\center{\includegraphics[scale=1.3]{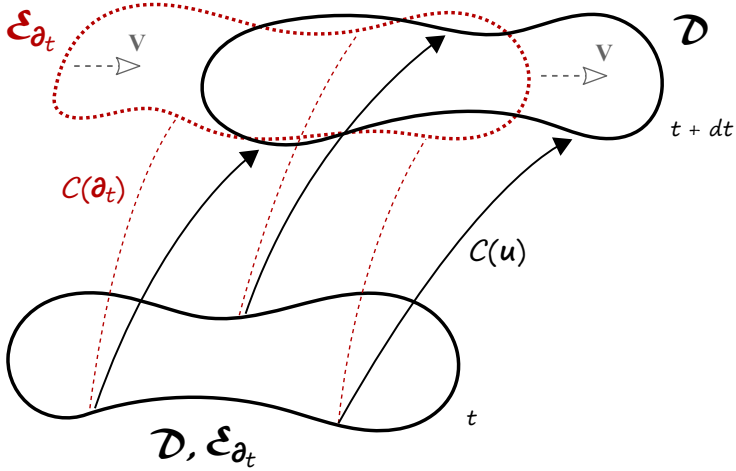}}
	\caption{%
		Representation of the motion of a compact domain $\CD$ between neighboring hypersurfaces. 
		$\CD$ is transported along the congruence of the fluid $\CC(\bm u)$, with $X^i = const.$, and contains 
		by construction the same collection of fluid elements throughout its evolution. 
		We introduce in this figure another compact domain, $\CEt$, carried along the congruence 
		$\CC(\bm{\partial}_t)$, 
		with $x^i = const.$, that coincides with $\CD$ at time $t$. $\CEt$ encloses the same collection of fluid 
		elements as $\CD$ at that time. At $t+{\mathrm d}t$, the two domains do not coincide anymore as the fluid 
		undergoes a spatial motion of velocity $\bm V$ in the coordinate system $(t, x^i)$  (hence 
		${\mathrm d}/{\mathrm d}t$ and $\int_{\CD_{\bm x}} \! \cdot \; {\mathrm d}^3x$ do not commute).
		This motion induces a flux of fluid particles across the boundary of $\CEt$.
		In the comoving and Lagrangian descriptions, the congruences $\CC(\bm{\partial}_t)$ and $\CC(\bm u)$ are 
		identical and this flux does not occur. A similar distinction would have to be made between $\CD$ and a domain transported along the flow of the hypersurfaces normal vector $\bm n$, with a flux of fluid particles accross the boundaries of the latter, except in the absence of tilt.}
	\label{fig:var_domain}
\end{figure}

To this aim, we consider the family of maps ${\bf \Phi}_t = \bm f(t,\cdot)$ introduced above to change the coordinates 
from $x^i$ to $X^i$. We have:
\beq
	x^i 
		= f^i(t,\bm X) \, , \qquad 
	{\mathrm d}^3 x 
		= \det \left( \frac{\partial \bm f(t,\bm X)}{\partial \bm X} \right) \, {\mathrm d}^3X 
		= J (t,\bm X) \, {\mathrm d}^3 X \, , 
	\label{eq:coord_change}
\eeq
while the region of integration transforms as $\CD_{\bm x} \to \CD_{\bm X} = {\bf \Phi}_t^{-1} (\CD_{\bm x})$. 
Inserting \eqref{eq:coord_change} into \eqref{eq:vol_D}, we get: 
\beq
	\CV_\CD (t) = \int_{\CD_{\bm X}} \sqrt{h (t, f^i(t,\bm X))} \, J(t,\bm X) \, {\mathrm d}^3X \, . 
	\label{eq:vol_D_lag}
\eeq
The invariance of the volume element $\sqrt{h(t,x^i)} \, {\mathrm d}^3x$ (here integrated over the same collection of 
fluid elements) with respect to changes of spatial coordinates appears here by noticing that 
$\sqrt{h (t, f^i(t,\bm X))} \, J(t,\bm X)$ above corresponds to the square root of the determinant of the components in the coordinate 
system $(t,\bm X)$ of the spatial metric $\bf h$. Obviously, the fluid is at rest in this coordinate system, allowing 
for the commutation of ${\mathrm d}/{\mathrm d}t = \partial_t\big|_{X^i}$ and $\int_{\CD_{\bm X}} \! \cdot \; {\mathrm d}^3X$.\footnote{%
	In contrast to the operator ${\mathrm d}/{\mathrm d}t$, the operator ${\mathrm d}/{\mathrm d}\tau$ does 
	not commute in general with $\int_{\CD_{\bm X}} \! \cdot \; {\mathrm d}^3X$, since 
	${\mathrm d}/{\mathrm d}\tau = (\gamma / N) \; {\mathrm d}/{\mathrm d}t$ depends on the spatial coordinates. 
} We can now write: 
\beq
	\frac{\mathrm d}{{\mathrm d}t} \CV_\CD =
		\int_{\CD_{\bm X}} \frac{\mathrm d}{{\mathrm d}t} \left(\sqrt{h(t,f^i(t,\bm X))} \, J \right) \, {\mathrm d}^3X \, , 
	\label{eq:vol_D_var_II}
\eeq
and, transforming the coordinates back to $x^i$ with the help of ${\bf \Phi}_t^{-1}$, we obtain: 
\beq
	\frac{\mathrm d}{{\mathrm d}t} \CV_\CD 
		 = \int_{\CD_{\bm x}} \frac{\mathrm d}{{\mathrm d}t} \left( J \sqrt{h} \right) J^{-1} \, {\mathrm d}^3 x 
		 = \int_{\CD_{\bm x}} \left( \frac{\mathrm d}{{\mathrm d}t} \sqrt{h}
			+ \sqrt{h} \, J^{-1} \frac{\mathrm d}{{\mathrm d}t} J \right) {\mathrm d}^3 x \, . 
	\label{eq:vol_D_var_III}
\eeq
Using the relations \eqref{eq:der_rel_1} and \eqref{eq:der_jacob}, this implies: 
\begin{align}
	\frac{\mathrm d}{{\mathrm d}t} \CV_\CD 
		& = \int_{\CD_{\bm x}} \left( \partial_t\big|_{x^i} \sqrt{h} 
			+ V^k \partial_k \sqrt{h} + \partial_k V^k \sqrt{h} \right) {\mathrm d}^3 x \nonumber \\ 
		& = \int_{\CD_{\bm x}} \left( \frac{1}{2} h^{ij} \partial_t\big|_{x^i} h_{ij} 
			+ \frac{1}{2} h^{ij} V^k \partial_k h_{ij} + \partial_k V^k \right) \sqrt{h} \, {\mathrm d}^3 x \, . 
	 \label{eq:vol_D_var_IV}
\end{align}
From the trace of the evolution equation \eqref{eq:evol_h} and noticing that 
\beq
	\frac{1}{2} h^{ij} \partial_k h_{ij} \,V^k  + \partial_k V^k =  V^k_{\phantom{k}|| k} \, , 
	\label{eq:vol_D_var_V}
\eeq
we finally end up with the expression of the coordinate-time variation of the hypersurface volume
(see Appendix \ref{app:3+1_lag} for an alternative derivation using instead the $3+1$ evolution equations along the congruence of the fluid): 
\begin{eqnarray}
	\frac{\mathrm d}{{\mathrm d}t} \CV_\CD 
		& = & 
			\int_{\CD_{\bm x}} \left( -N \CK + N^i_{\ || i} + V^i_{\ || i} \right)
			\sqrt{h} \, {\mathrm d}^3 x  
		\nonumber \\
		& = & 
			\int_{\CD_{\bm x}} \left( -N \CK + \big( N v^i \big)_{|| i} \right) \sqrt{h} \, {\mathrm d}^3 x \, ,
		\label{eq:vol_D_final}  
\end{eqnarray}
where we used relation \eqref{eq:relat_vel} for the last equality.


\subsubsection{Averaging and commutation rule}
\label{subsubsec:com_rule}

We define the \textit{hypersurface-volume (or extrinsic)  average} of any scalar $\psi$ on a compact comoving domain $\CD$ lying within the arbitrary spatial hypersurfaces as 
\beq
	\average{\psi} (t) :=  \frac{1}{\VD} \int_\CD \psi \, n^\mu {\mathrm d}\sigma_\mu  = \frac{1}{\CV_\CD} \int_{\CD_{\bm x}} \psi (t, x^i) \, \sqrt{h (t, x^i)} \, {\mathrm d}^3x \, . 
	\label{eq:spat_aver}
\eeq
Applying this definition on \eqref{eq:vol_D_final}, we can write the rate of change of $\CV_\CD$ as 
\beq
	\frac{1}{\CV_\CD} \frac{\mathrm d}{{\mathrm d}t} \CV_\CD = \average{ -N \CK + \big( N v^i \big)_{|| i} } \, ,
	\label{eq:vol_rate}
\eeq
and express the coordinate-time derivative of the averaged scalar $\psi$ in the form:
\beq
	\frac{\mathrm d}{{\mathrm d}t} \average{\psi} = 
		- \average{ -N \CK + \big( N v^i \big)_{|| i} } \average{\psi} 
		+ \frac{1}{\CV_\CD} \frac{\mathrm d}{{\mathrm d}t}
			\int_{\CD_{\bm x}} \psi (t, x^i) \, \sqrt{h (t, x^i)} \, {\mathrm d}^3x \, .
	\label{eq:com_rule_I}
\eeq
The second term on the right-hand side is evaluated by following the same procedure as above: 
we perform a coordinate change by means of the maps ${\bf \Phi}_t$,
\begin{align}
	\frac{\mathrm d}{{\mathrm d}t} \int_{\CD_{\bm x}} \!\psi (t, x^i) \sqrt{h (t, x^i)} \, {\mathrm d}^3 x 
		& = \frac{\mathrm d}{{\mathrm d}t}
		     \int_{\CD_{\bm X}} \!\psi(t, f^i(t,\!\bm X)) \sqrt{h(t, f^i(t,\!\bm X))} \, J(t,\bm X) \, {\mathrm d}^3 X \nonumber \\ 
		& = \!\int_{\CD_{\bm X}}\! \frac{\mathrm d}{{\mathrm d}t} \!
			\left( \psi(t, f^i(t,\!\bm X)) \sqrt{h(t, f^i(t,\!\bm X))} \, J(t,\bm X) \right) {\mathrm d}^3 X \, , 
		\nonumber 
\end{align}
and, transforming back to the reference coordinates, expanding the integrand, and using once again the definition 
\eqref{eq:spat_aver}, we end up with 
\beq
	\frac{1}{\CV_\CD} \frac{\mathrm d}{{\mathrm d}t} \int_{\CD_{\bm x}} \psi \, \sqrt{h} \, {\mathrm d}^3 x 
		= \average{ \frac{\mathrm d}{{\mathrm d}t} \psi } + \average{ \left( -N \CK + \big( N v^i \big)_{|| i} \right) \psi } \, . 
	\label{eq:com_rule_III} 
\eeq
Plugging this equation into \eqref{eq:com_rule_I}, we finally obtain the commutation rule for extrinsic averages
over a spatial \textit{comoving domain}. We formulate this new result in the form of a \textit{Lemma}.

\begin{lemma}[Commutation rule for extrinsic volume averages] \label{lemma:com_rule} \\
\vspace{-7pt}	
	
	The commutation rule between extrinsic spatial averaging on a compact domain $\CD$, lying within a constant-$t$ hypersurface and comoving with the fluid, and comoving differentiation with respect to the coordinate time reads, for any 
	$3+1$ foliation of spacetime and for any scalar $\psi$:
	\beq
		\frac{\mathrm d}{{\mathrm d}t} \average{\psi} = 
			\average{ \frac{\mathrm d}{{\mathrm d}t} \psi } 
			 - \average{ -N \CK +  \big( N v^i \big)_{|| i} } \average{\psi} 
			 + \average{ \left( -N \CK + \big( N v^i \big)_{|| i} \right) \psi } \, . 
		\label{eq:com_rule_final}
	\eeq
\end{lemma}
This commutation rule is independent of the shift vector, and hence is independent of the propagation of the spatial 
coordinates. This feature is inherited from the coordinate-independent definition of the propagation of the domain of averaging obtained by requiring it to be comoving with the fluid.

Note that, as shown in Appendix~\ref{app:usual_aver_intrinsic} (Eq.\eqref{eq:tilted_exp_rate_app} therein), the local terms appearing in the rate of change of the volume \eqref{eq:vol_rate} can be equivalently expressed in terms of the lapse, tilt, and fluid expansion rate as
\begin{equation}
\label{eq:rel_volume_rates}
- N \CK + (N v^i)_{||i} = \frac{N}{\gamma} \Theta - \frac{1}{\gamma} \ddt{\gamma} \; .
\end{equation}
The commutation rule can thus alternatively be written under the following form for any scalar $\psi$:
\begin{equation}
 \ddt{} \average{\psi} = \average{\ddt{}\psi} - \average{\frac{N}{\gamma} \Theta - \frac{1}{\gamma} \ddt{\gamma}} \average{\psi} + \average{\left( \frac{N}{\gamma} \Theta - \frac{1}{\gamma} \ddt{\gamma} \right) \psi} \; ,
 \label{eq:com_rule_final_theta}
\end{equation}
which will be useful when applied to fluid rest frame variables such as $\epsilon$ or $\varrho$.


\subsection{Conservation of the fluid rest mass}
\label{subsec:mass_cons}

We introduce the conserved fluid rest mass flux vector $\bm M$,
\beq
	M^\mu := \varrho u^\mu \, , \quad \;\; 
	\nabla_\mu M^\mu = 0 \, ,
	\label{eq:rest_mass_vector}
\eeq
from the (conserved) rest mass density $\varrho$. The rest mass of the fluid within the domain $\CD$ is given by the flow of $\bm M$ through $\CD$ (see also the remarks on the interpretation of the rest mass, rest mass density and flux vector in footnote \ref{fn:restmass}):
\beq
	M_\CD 
	:= \int_\CD M^\mu \mathrm{d} \sigma_\mu 
	 = \int_\CD - \varrho u^\mu n_\mu \sqrt{h} \, \mathrm{d}^3 x
	= \CV_\CD \average{\gamma \varrho} \, ,
	\label{eq:rest_mass_scalar}
\eeq
with the oriented spatial volume element $\mathrm{d} \sigma_\mu = - n_\mu \sqrt{h} \, \mathrm{d}^3 x$, and where we used 
$-u^\mu n_\mu = \gamma$, Eq.~\eqref{eq:lorentz}.

The conservation of this rest mass can be seen by integrating the conservation equation of $\bm M$ 
over the spacetime tube $\mathscr T$ swept by the domain $\CD$ between two hypersurfaces at times $t_1$ and $t_2 > t_1$:
\beq
	0 
	= \int_{\mathscr T} \nabla_\mu M^\mu \sqrt{g} \, {\mathrm d}^4 x
	= \oint_{\partial \! \mathscr T} M^\mu {\mathrm d} \eta_\mu \ ,
\eeq
where $g := |\det(g_{\mu \nu})|$ and ${\mathrm d} \eta_\mu$ is the outward-oriented volume element on the boundary $\partial \mathscr T$ of $\mathscr T$.
Introducing the timelike part $\mathscr A$ of $\partial \mathscr T$, with $\bm A$ its outward-oriented unit normal vector (see figure~\ref{fig:flowtube}) and 
${\mathrm d}V_{\!\mathscr A}$ its volume $3-$form, we rewrite the above as: 
\begin{eqnarray}
	0 & = & \int_{\CD_{t_2}} \gamma \varrho \, \sqrt{h} \, {\mathrm d}^3 x - \int_{\CD_{t_1}} \gamma \varrho \, \sqrt{h} \, {\mathrm d}^3 x
	+ \int_{\mathscr A} \, M^\mu A_\mu \, {\mathrm d}V_{\!\mathscr A} \nonumber \\
	& = & M_{\CD_{t_2}} - M_{\CD_{t_1}} + \int_{\mathscr A} \, \varrho \, u^\mu A_\mu \,  {\mathrm d}V_{\!\mathscr A} \, . 
\end{eqnarray}
The last term cancels out precisely because the domain propagates along the fluid flow lines so that the normal vector $\bm A$ is orthogonal to 
$\bm u$ everywhere on the boundary $\mathscr A$. We therefore end up with the conservation of the rest mass within $\CD$: $M_{\CD_{t_2}} = M_{\CD_{t_1}}$.

\begin{figure}[!ht]
\center{
		\includegraphics[scale=0.8]{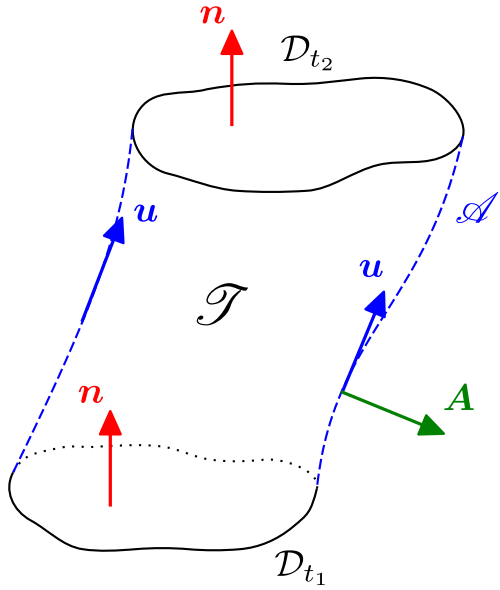}
	}
	\caption{%
		Representation of the flow tube $\mathscr{T}$ and the various vectors and subsets of $\partial \mathscr{T}$ used in the proof of the conservation of the fluid rest mass within $\CD$ in section \ref{subsec:mass_cons}. (We here illustrate orthogonality in terms of an Euclidean spacetime metric rather than the Lorentzian metric $\mathbf{g}$ for visualization purposes.)
	}
	\label{fig:flowtube}
\end{figure}

Alternatively, one can make use of the local continuity equation \eqref{eq:cons_restmass_density} for $\varrho$, equivalent to the conservation of $\bm M$ \eqref{eq:rest_mass_vector}, rewritten in terms of a coordinate-time comoving derivative:
\begin{equation}
	\frac{\mathrm d}{\mathrm dt} \varrho + \frac{N}{\gamma} \Theta \varrho = 0 \, .
	\label{eq:cons_restmass_t}
\end{equation}
Applying the commutation rule expressed in terms of $\Theta$, Eq.~\eqref{eq:com_rule_final_theta}, and the corresponding form of the volume expansion rate, $\VD^{-1} \, \mathrm{d}\VD/\mathrm{d}t = \langle (N/\gamma)\,\Theta - \gamma^{-1}\,\mathrm{d}\gamma/\mathrm{d}t \rangle_\CD$, to the average of the above local continuity equation multiplied by $\gamma$ then gives $\mathrm{d}(\VD \average{\gamma \varrho}) / \mathrm{d}t = 0$, recovering the conservation of $M_\CD$. 


\subsection{Averaged inhomogeneous cosmologies in the extrinsic approach}
\label{subsec:av_cosmo_extrinsic}

We introduced in the previous subsections a scalar averaging procedure on a compact spatial domain comoving with the fluid, based on the hypersurface volume. We derived the corresponding 
commutation rule and showed the preservation of the total fluid rest mass within the comoving domain. Both hold for any foliation of spacetime. 
By means of this formalism, and from the $3+1$ Einstein equations given in section \ref{subsec:3+1}, we now give an (under-determined) set of scalar balance equations describing the effective dynamics of spatially averaged (in terms of hypersurface-volume averages) comoving and compact regions of inhomogeneous cosmologies. 


\subsubsection{Averaged evolution equations}
\label{subsubsec:av_evol}

Following the original proposal of \cite{buchert:av_Newton} (used in Papers~I and II), we define the \textit{hypersurface-volume (or extrinsic) effective scale factor} $a_\CD$ of the comoving domain $\CD$, 
\beq
	a_\CD (t) := \left( \frac{\CV_\CD (t)}{\CV_{\CD_\mathbf{i}}} \right)^{1/3}  \; ,
\eeq
where $\CV_{\CD_\mathbf{i}}$ refers to the volume at the initial time $t_\mathbf{i}$. The hypersurface-volume expansion rate 	\eqref{eq:vol_rate} then reads:
\begin{equation}
	\frac{1}{a_\CD} \frac{{\mathrm d} a_\CD}{{\mathrm d}t}
		\, = \, \frac{1}{3} \average{ - N \CK + \big( N v^i \big)_{|| i} }  \; .
	\label{eq:evol_aD_extrinsic}
\end{equation}
We can now average two scalar Einstein equations: the trace of $N \, \times \,$\eqref{eq:evol_K}, and $N^2 \times \,$\eqref{eq:hamilt_const}. Upon using the commutation rule \eqref{eq:com_rule_final} along with relation \eqref{eq:evol_aD_extrinsic} above, we obtain 
the effective evolution equations for a comoving region of an inhomogeneous model universe in the fluid-extrinsic averaging procedure, that we formulate in the form of a \textit{Theorem}.\\
\begin{thgroup}
\label{ths:av_extrinsic}
\begin{theorem}[Extrinsically averaged evolution equations] \label{th:av_evol} \\
\vspace{-7pt}

The evolution equations for the extrinsic effective scale factor of a compact spatial domain $\CD$ comoving with a general fluid flow read, for any $3+1$ spatial foliation of spacetime:
\begin{align}
	3 \, \frac{1}{a_\CD} \frac{{\mathrm d}^2 a_\CD}{{\mathrm d}t^2} 
		& = - 4 \pi G \average{ N^2 \left( \epsilon + 3 p \right)} + \average{ N^2} \cc + \CQ_\CD + \CP_\CD + \frac{1}{2} \CT_\CD \, , 
		\label{eq:av_raych} \\ 
	3 \left( \frac{1}{a_\CD} \frac{{\mathrm d} a_\CD}{{\mathrm d}t} \right)^2 
		& = 8 \pi G \average{ N^2 \epsilon} + \average{ N^2} \cc
			- \frac{1}{2} \average{N^2 \, \CR} - \frac{1}{2} \CQ_\CD - \frac{1}{2} \CT_\CD \, , 
		\label{eq:av_hamilt}
\end{align}
with $\CQ_\CD$, $\CP_\CD$ and $\CT_\CD$ respectively the (extrinsic) kinematical backreaction, dynamical backreaction, and stress-energy backreaction, defined as follows:
\begin{align}
	\CQ_\CD := 
	& \, \average{ N^2 \left( \CK^2 - \CK_{ij} \CK^{ij} \right) } - \frac{2}{3} \average{ - N \CK + \big( N v^i \big)_{|| i} }^2 \, ,
	\label{eq:kin_back}
\\
	\CP_\CD := 
	& \, \average{ \left( \big( N v^i \big)_{|| i} \right)^2 
		+ \frac{\mathrm d}{{\mathrm d}t} \left( \big( N v^i \big)_{|| i} \right) 
		- 2 N \CK \, \big( N v^i \big)_{|| i} 
		- N^2 v^i \CK_{|| i} } \nonumber
\\ 
	& \; + \average{N  N^{|| i}_{\ \; || i} - \CK \, \frac{{\mathrm d} N}{{\mathrm d}t}} \, , 
	\label{eq:dyn_back}
\\ 
	\CT_\CD := 
		& \, - 16 \pi G \average{ N^2 \left( ( \gamma^2 - 1 ) ( \epsilon + p ) 
		+ 2 \, \gamma v^\alpha q_\alpha + v^\alpha v^\beta \pi_\albe \right) } \, . 
	\label{eq:SE_back}
\end{align}
\end{theorem}
\medskip\noindent
{\bf Remarks to \textit{Theorem} \ref{th:av_evol}:}
Care should be taken in the interpretation of the system (\eqref{eq:av_raych},\eqref{eq:av_hamilt}).
These equations are globally invariant under the remaining coordinate freedoms, that is, (i) under any change of the spatial coordinates, 
or (ii) under a change of the time coordinate of the form $t \mapsto T(t)$ with $\mathrm{d}T / \mathrm{d}t > 0$ and accordingly of the lapse as 
$N \mapsto N' = N \,(\mathrm{d}T / \mathrm{d}t)^{-1}$ (which corresponds to a re-parametrization of the hypersurfaces).
However, individual terms, as well as each equation side taken separately, are invariant under changes of spatial coordinates only.
A time-coordinate change as above would rescale most terms, such as $\QD$, $\TD$ or
$3 \,( (1/a_\CD) \; \mathrm{d} a_\CD / \mathrm{d}t )^2$, by the time-dependent factor $( \mathrm{d}T / \mathrm{d}t )^{-2}$ (strictly preserving their sign).
The terms $\PD$ and $(3 / a_\CD) \, \mathrm{d}^2 a_\CD / \mathrm{d}t^2$ would undergo an affine transformation,
with this same rescaling plus an additional term (the same term for both, thus preserving the equation globally)
proportional to $(\mathrm{d} a_\CD / \mathrm{d}t)\, (\mathrm{d}^2 T / \mathrm{d}t^2)$, so that even their
sign can be arbitrarily changed in a time-dependent manner.

Accordingly, depending on what $t$ represents, the left-hand sides of equations (\eqref{eq:av_raych},\eqref{eq:av_hamilt})
may not follow an interpretation similar to the corresponding $3\, (\dot a / a)^2$ and $3 \, \ddot a / a$ of the standard Friedmann equations. 
These are unambiguously expressed as derivatives with respect to the common proper time of the comoving fluid.\footnote{%
	One could in the same way parametrize the Friedmann model by a different time coordinate while staying within the homogeneous foliation,
	and similarly get rescaled terms and an arbitrarily altered acceleration term (see, \textit{e.g.}, the system of equations (20) in Paper~II 
	\cite{buchert:av_pf} or the system of equations (40) in \cite{rza5}).
	The usual form of the Friedmann equations removes this freedom by choosing the proper time as the most natural time parameter in this situation.
	As, additionally, the spatial coordinates generally used in this framework are comoving with the fluid content, this picture corresponds to what we termed
	in this work a Lagrangian description.
}
Without a well-specified choice for $t$, conclusions may only be drawn on quantities that are invariant under the change of time coordinate expressed above. Such invariants include the sign of each term (except the dynamical backreaction and, importantly, the scale factor acceleration),
as we shall discuss, \textit{e.g.}, for the stress-energy backreaction in section \ref{subsubsec:back_remarks}. They also include effective dimensionless ``$\mathrm{\Omega}$'' parameters that may be defined
for a non-stationary $a_\CD$ ($\mathrm{d}a_\CD / \mathrm{d}t \neq 0$) by dividing each term of Eq.~\eqref{eq:av_hamilt} by $3 \,[ (1/a_\CD) \; \mathrm{d} a_\CD / \mathrm{d}t ]^2$.

The generality of \textit{Theorem} \ref{th:av_evol} allows one to choose the most suitable definition for $t$ in any specific application. 
The Friedmannian interpretation of $t$, $\mathrm{d}/\mathrm{d}t$ and $\mathrm{d}^2/\mathrm{d}t^2$ can be recovered for some choices that are applicable to general settings.
This is the case for instance for a Lagrangian description where $t$ is a proper time for all fluid elements, involving a choice of foliation (see section~\ref{subsec:lagrangian_form} for an example of application of this description). 
One could also set $t$ within any choice of foliation
such that it coincides with the proper time along some given, single timelike worldline, for instance taken to model 
the worldline of an observer on Earth. Once a specification of the time label is performed, each term of the above equations, including the acceleration 
term $(3 / a_\CD) \; \mathrm{d}^2 a_\CD / \mathrm{d}t^2$ or its sign, can be interpreted in direct relation to the physical meaning of the chosen $t$.


\subsubsection{Integrability and energy balance conditions}
\label{subsubsec:av_cond}

We proceed by deriving the \textit{integrability condition} for the system of equations of \textit{Theorem} \ref{th:av_evol}, which provides the relation that has to hold for \eqref{eq:av_hamilt} to be the integral of \eqref{eq:av_raych}. This condition is obtained by taking 
the comoving coordinate-time derivative of \eqref{eq:av_hamilt}, and by inserting the set of equations \eqref{eq:av_raych} and \eqref{eq:av_hamilt} 
back into the result.\footnote{%
	Alternatively, we can derive the integrability condition directly from the Einstein equations. For this we can derive the local
	evolution equations for the square of the trace-free part of the extrinsic curvature and for the scalar $3-$curvature 
	using \eqref{eq:evol_K} and \eqref{eq:hamilt_const}. Averaging these equations and combining them, we also obtain the integrability condition above (\textit{cf.} \cite[appendix]{GBC}).
}
Complementing this condition by the average of the energy conservation equation \eqref{eq:en_cl}, we write the second part of the above \textit{Theorem} in the following.

\begin{theorem}[Integrability and energy balance conditions] \label{th:av_cond} \\
\vspace{-7pt}

A necessary condition of integrability of equation \eqref{eq:av_raych} to yield equation \eqref{eq:av_hamilt} is given by the relation:
\small{%
\begin{gather}
	\frac{\mathrm d}{{\mathrm d}t} \CQ_\CD + \frac{6}{a_\CD} \ddt{a_\CD} \CQ_\CD 
	+ \frac{\mathrm d}{{\mathrm d}t} \average{N^2 \CR} + \frac{2}{a_\CD} \ddt{a_\CD} \average{N^2 \CR} 
	+ \frac{\mathrm d}{{\mathrm d}t} \CT_\CD + \frac{4}{a_\CD} \ddt{a_\CD} \left(\CT_\CD + \CP_\CD \right) \nonumber \\
	= 16 \pi G \left( \frac{\mathrm d}{{\mathrm d}t} \average{N^2 \epsilon}
	+ \frac{3}{a_\CD} \ddt{a_\CD} \average{N^2 \left(\epsilon + p \right)} \right) + 2 \cc \frac{{\mathrm d}}{{\mathrm d}t} \average{N^2} \, , 
	\label{eq:int_condition}
\end{gather}
}%
\normalsize{%
where the source part on the right-hand side satisfies the averaged energy conservation law:}
\small{%
	\begin{multline}
		\!\!\!\!\! \frac{\mathrm d}{{\mathrm d}t} \average{N^2 \epsilon}
			+ \frac{3}{a_\CD} \ddt{a_\CD} \average{N^2 \left( \epsilon + p \right)}  
			= \average{\frac{N}{\gamma} \Theta} \! \average{N^2 p}
			- \average{\frac{N}{\gamma} \Theta \, N^2 p} \!\!
			- \average{\frac{1}{\gamma} \frac{{\mathrm d} \gamma}{{\mathrm d}t}} \! \average{N^2 p} \\ 
			+ \average{\left( 2 \frac{1}{N} \frac{{\mathrm d} N}{{\mathrm d}t}
			- \frac{1}{\gamma} \frac{{\mathrm d} \gamma}{{\mathrm d}t} \right) N^2 \epsilon}
			- \average{ \frac{N^3}{\gamma} \left( q^\alpha a_\alpha + \nabla_\alpha q^\alpha + \pi^\albe \sigma_\albe \right) } \, . 
		\label{eq:av_en_cons}
	\end{multline}
}%
\normalsize{%
This conservation law can be complemented by the conservation of the fluid rest mass, $\mathrm{d}M_\CD / \mathrm{d}t = 0$, which may be rewritten as follows:}
\small{\begin{equation}
 \ddt{} \average{\gamma \varrho} + \frac{3}{a_\CD} \ddt{a_\CD} \average{\gamma \varrho} = 0 \; .
\end{equation}}
\end{theorem}
\end{thgroup}

\begin{proof}
\small{The local energy conservation law \eqref{eq:en_cl} implies: 
\beq
	\frac{\mathrm d}{{\mathrm d}t} \big( N^2 \epsilon \big) + \frac{N}{\gamma} \Theta \left( N^2 \left( \epsilon + p \right) \right) 
		= 2 \frac{1}{N} \frac{{\mathrm d} N}{{\mathrm d}t} N^2 \epsilon 
		- \frac{N^3}{\gamma} \left( q^\alpha a_\alpha + \nabla_\alpha q^\alpha + \pi^\albe \sigma_\albe \right) \, . 
	\label{eq:en_cl_scaled}
\eeq
Relation \eqref{eq:av_en_cons} is then recovered by averaging the local equation \eqref{eq:en_cl_scaled} and applying the commutation rule expressed in terms of $\Theta$, Eq.\eqref{eq:com_rule_final_theta}.} $\square$
\end{proof}

We present as \textit{Corollary}~\ref{cor:extrinsic_all} in Appendix~\ref{app:usual_aver_intrinsic} an equivalent formulation of the system of equations of \textit{Theorem}~\ref{ths:av_extrinsic}, focussing explicitly on the kinematic and dynamical variables of the fluid rather than on the geometric properties of the hypersurfaces (such as their intrinsic and extrinsic curvatures).

The system of equations of this theorem could also be rewritten in more compact ways, as we shall illustrate
for similar equations obtained within an alternative averaging approach in section~\ref{sec:intrinsic}.
We will keep it under the current form, as it
is already sufficient to discuss important properties and relations to the literature, to which we turn now.


\subsection{Discussion}
\label{subsec:discussion}

We summarize in the first part of this subsection the framework of our study thus far. We then discuss the backreaction terms that were defined, investigate boundary effects and boundary-free global domains, and finally discuss relations to the literature for global and general domains successively.


\subsubsection{Summary}
\label{subsubsec:gen_remarks}

We have worked with three independent sets of worldlines: the normal congruence along $\bm n$, 
everywhere orthogonal to the hypersurfaces of constant coordinate time $t$; the congruence of the \textit{coordinate frames} along $\bm{\partial}_t$, \textit{i.e.} at constant $x^i$; and the threading congruence of the \textit{comoving frames} (or, equivalently, the fluid rest frames) 
along $\bm u$. 
The deviations between $\bm n$ and $\bm{\partial}_t$, on the one hand, and between $\bm n$ and $\bm u$, 
on the other hand, are identified respectively by the vector fields $\bm N$ and $\bm v$, while that between $\bm{\partial}_t$ and $\bm u$ is pinpointed by $\bm V$ (see figure~\ref{fig:schem_vect_3D}).

This general configuration allows for a fluid flow with vorticity and tilted with respect to the normal of the three-surfaces,
and for an arbitrary propagation of the spatial coordinates. Also, the lapse function is left unspecified, preserving the  freedom in the construction of the spatial slices and in their time parametrization. 

We have considered a compact spatial domain $\CD$, lying within the hypersurfaces and transported along the fluid flow lines,
thus enclosing by construction the same collection of fluid elements throughout the evolution.
In the generic situation, this domain undergoes a spatial motion in the coordinate system $(t,x^i)$, since the integral 
curves of $\bm{\partial}_t$ and $\bm u$ do not coincide (see figure~\ref{fig:var_domain}). 

Within this framework, we have defined an averaging formalism that is based on hypersurface-volume averaging, and we have established the general commutation rule (formula \eqref{eq:com_rule_final})
between the corresponding spatial averaging operation and differentiation with respect to the coordinate time along the fluid flow lines. 
We have then derived in \textit{Theorem}~\ref{ths:av_extrinsic} a set of scalar equations describing the regional dynamics of portions of an inhomogeneous fluid spatially averaged in this way. The results obtained hold for a general fluid and for a general foliation of spacetime and,
in particular, are independent of the propagation of the spatial coordinates. In such a general foliation, however, we have stressed the risk of too hastily interpreting these results, in particular of interpreting the time acceleration term
in the same way as the proper-time acceleration term $\ddot{a}/a$ of the standard Friedmann equations: its meaning strongly depends on the interpretation of the chosen time parameter $t$ itself.
We have also highlighted the Lagrangian foliation and coordinates choice as a transparent setting that allows one to recover the usual interpretation.


\subsubsection{Comments on the backreaction terms}
\label{subsubsec:back_remarks}

The kinematical backreaction $\QD$ \eqref{eq:kin_back} and the dynamical backreaction $\PD$ \eqref{eq:dyn_back} 
generalize the expressions given in Paper~II. The emphasis is set here on the geometric variables of the foliation 
($\CK$, $\CK_{i j}$, $\CR$, \textit{etc}.), rather than on the kinematic variables of the fluid ($\Theta$, $\Theta_{i j}$, \textit{etc}.; see 
Appendix~\ref{app:usual_aver_intrinsic} for a formulation in terms of the latter). 
These two sets of variables are identical in the fluid-orthogonal approach of Paper~II, but they differ in the 
present framework. Differences with the setup of Paper~II can be made explicit in the kinematical backreaction 
term by reformulating it as
\begin{equation}
	\QD = \frac{2}{3} \left( \average{N^2 \CK^2} - \average{- N \CK + (N v^i)_{||i}}^2 \right) 
	- 2 \average{N^2 \rm \CK_{tl}^2} \; ,
\end{equation}
where the traceless part of the extrinsic curvature defines the squared rate of shear of the normal congruence,
$\CK_{\rm tl}^2 := \left(\CK_{i j} - (\CK / 3) \, h_{i j} \right) \left(\CK^{i j} - (\CK/3) \, h^{i j} \right) / 2$,
and the trace $\CK$ gives (up to a sign change) the expansion rate of the normal congruence. This formulation is reminiscent of Paper~II. However, it is no longer expressed in terms of kinematic variables, and it highlights an additional contribution 
$(N v^i)_{||i}$ from the Eulerian velocity (or, equivalently, the tilt). We can also notice additional terms due to 
the Eulerian velocity in the expression of the dynamical backreaction \eqref{eq:dyn_back}.

If the fluid is vorticity-free, we can choose a fluid-orthogonal foliation, namely we can set $\bm n = \bm u$ as 
in Paper~II and, thus, have $\bm v = \bm 0$ and $\gamma = 1$. In this configuration the geometric and kinematic 
variables coincide, as well as the accelerations associated to both frames ($\bm a$ and $\bm a^{(n)}$), and we recover the expressions of $\CQ_\CD$ and $\CP_\CD$ given in Paper~II. 
As this setting also implies the vanishing of the stress-energy backreaction $\CT_\CD$, we formally recover 
the same set of 
evolution equations for the effective scale factor (up to the additional inclusion of the cosmological constant 
contribution). This could have been expected, but note that here, in contrast to Paper~II, we allow for a non-vanishing 
shift vector and a non-perfect fluid.
As already discussed, and as for the commutation rule \eqref{eq:com_rule_final}, the shift does not contribute because 
local evolutions are studied along the fluid flow lines, and the spatial domain of averaging is comoving with the fluid. 
Hence, the shift vector neither plays a dynamical role locally nor on average. However, even though nonperfect-fluid effects are not formally present in the evolution equations for the effective scale factor, they still influence 
the dynamics through the local and average evolution of the energy density (see equations \eqref{eq:en_cl} and 
\eqref{eq:av_en_cons}).

In addition to contributing to the kinematical and dynamical backreaction terms, the tilt also yields the additional backreaction term $\CT_\CD$,
which we named \textit{stress-energy backreaction}, and which can be interpreted in the following ways.

Firstly, it measures the difference between the energy of the fluid as seen in its rest frames and as seen in the normal frames. In this sense, it is thus (up to an overall negative factor) an average measure of the kinetic energy 
of the fluid in the normal frames. Indeed, using relation \eqref{eq:rel_E} we can write
\begin{equation}
 	( \gamma^2 - 1 ) ( \epsilon + p ) + 2 \, \gamma v^\alpha q_\alpha + v^\alpha v^\beta \pi_\albe = E - \epsilon
		= T_{\mu \nu} n^{\mu} n^{\nu} - T_{\mu \nu} u^\mu u^\nu \, ,
\end{equation}
so that
\begin{equation}  \label{eq:SE_back_energy_difference}
 	\CT_\CD 
		= -16 \pi G \average{N^2 (E - \epsilon)} 
		= -16 \pi G \average{N^2 (T_{\mu \nu} n^{\mu} n^{\nu} - T_{\mu \nu} u^\mu u^\nu)} \, .
\end{equation}

Secondly, it also expresses the difference between the isotropic pressures measured in each of both frames, since combining relations \eqref{eq:rel_E} and \eqref{eq:rel_S} gives
\begin{equation}
 	E - \epsilon = S - 3 p = T_{\mu \nu} h^{\mu \nu} - T_{\mu \nu} b^{\mu \nu} \; .
\end{equation}
These two interpretations arise almost by definition: the backreaction term $\CT_\CD$ has been indeed introduced to express the dynamics of the averaging domain as sourced by averages
of scalar dynamical quantities of the fluid as seen in its rest frames, $\epsilon$ and $p$ (recall equations
\eqref{eq:av_raych}--\eqref{eq:av_hamilt}), rather than by the quantities measured in the normal frames, $E$ and $S$. Only the former
correspond to intrinsic thermodynamical quantities of the fluid that are directly described by its equation of state. 

Thirdly, as will be shown in section \ref{subsubsec:global_domain}, it corresponds to the `bulk' tilt contribution in that it survives for a boundary-free domain, while the tilt contributions to $\QD$ and $\PD$ are boundary terms. 

The last expression in equation \eqref{eq:SE_back_energy_difference} shows that our stress-energy backreaction corresponds (up to a numerical factor)
to the `fluid corrections' terms introduced by Brown \textit{et al.} in \cite{brown:aver}, while the first form \eqref{eq:SE_back}
is sufficient to identify it with the (unnamed and slightly more general) $\langle F \rangle$ term appearing in R\"as\"anen's equations in
\cite{rasanen:lightpropagation}, and to see that it reduces to the `tilt effects' noticed by Gasperini \textit{et al.} in \cite{marozzi:aver2}
in the particular case of a perfect fluid, still up to numerical factors.

The sign of $\CT_\CD$ will usually be constrained and will remain negative, consistently with the interpretation of $- \CT_\CD$ as a measure of kinetic energy, so that this backreaction
will contribute as a deceleration term to the effective acceleration equation \eqref{eq:av_raych}. This constraint is expressed by the following \textit{Proposition}.\\

\begin{proposition}[sign of the stress-energy backreaction]\\
\label{prop:SE_back}
\vspace{-7pt}

If the matter stress-energy tensor satisfies the Null Energy Condition, then
\begin{equation}
	(\gamma^2 - 1) (\epsilon + p) + v^\mu v^\nu \pi_{\mu \nu} \geq 0 \; ,
\end{equation}
and the following assumptions on the heat vector $\bm q$ separately impose $\CT_\CD \leq 0$:
\begin{itemize}
\item [(i)] a vanishing heat vector, $\bm q = \bm 0$ (this includes the case of a perfect fluid, for which the constant sign
of the corresponding `tilt effects' was already noticed in \cite{marozzi:aver2}), which is equivalent to defining the 
fluid $4-$velocity as an eigenvector of the stress-energy tensor; or,
\item [(ii)] a preferred mutual spatial orientation between $\bm v$ and (the projection onto the hypersurfaces of) $\bm q$ ensuring
$N^2 \gamma \, q_\mu v^\mu \geq 0$, locally or on average; or,
\item [(iii)] on the contrary and more realistically, a variable orientation of the heat vector de-correlated from that of $\bm v$ and from the value
of the lapse $N$ and Lorentz factor $\gamma$, so that the variable-sign term $N^2 \gamma \, q_\mu v^\mu$ is averaged out while the other terms all add up
positively: $\left| \average{N^2 \gamma \, q_\mu v^\mu} \right| \ll \average{(\gamma^2 - 1) (\epsilon + p) + v^\mu v^\nu \pi_{\mu \nu}}$.
\end{itemize}
\end{proposition}

\begin{proof}
\small{
Noting that $(b^\mu_{\ \nu} v^\nu) (b_{\mu \rho} v^\rho) = b_{\mu \nu} v^\mu v^\nu = \gamma^2 - 1 = \gamma^2 v^2$, $v := \sqrt{v^\alpha v_\alpha}$, one can define two
future-pointing null vectors ${\bm k}_+, {\bm k}_-$ as
$k^\mu_\pm := \gamma v \, u^\mu \mp b^\mu_{\ \nu} n^\nu = \gamma v \, u^\mu \pm b^\mu_{\ \nu} v^\nu$.
The projections of the stress-energy tensor onto these vectors yield:
\begin{equation}
 T_{\mu \nu} k^\mu_\pm k^\nu_\pm = (\gamma^2 - 1) (\epsilon + p) + \pi_{\mu \nu} v^\mu v^\nu \mp 2 \gamma v \, q_\mu v^\mu \; .
\end{equation}
According to the Null Energy Condition (which we recall is a condition of positiveness of the projection $T_{\mu \nu} k^\mu k^\nu$
for any future-oriented null vector $\bm k$), both projections are positive, hence
\begin{equation}
 (\gamma^2 -1) (\epsilon + p) + \pi_{\mu \nu} v^\mu v^\nu \; \geq \; 2 \gamma v | q_\mu v^\mu | \; \geq \; 0 \; .
\end{equation}
Recalling that
$\CT_\CD = - 16 \pi G \average{ N^2 \left( ( \gamma^2 - 1 ) ( \epsilon + p ) + \pi_\mnu v^\mu v^\nu + 2 \gamma \, q_\mu v^\mu \right) }$
(equation \eqref{eq:SE_back}), and since $2 \gamma v < 2 \gamma$, even the (stronger) first inequality is insufficient to conclude on the sign
of $\CT_\CD$ without further assumptions on $\bm q$. This was to be expected since the same reasoning could be applied similarly after interchanging the roles played by $\bm u$ and $\bm n$ (that is, using the normal-frame decomposition of the stress-energy tensor,
which replaces for instance $\bm q$ by $\bm J$, and using the null vectors $k'^\mu_\pm := \gamma v \, n^\mu \mp h^\mu_{\ \nu} u^\nu$
instead of $k^\mu_\pm$), which exchanges $\CT_\CD$ and $- \CT_\CD$. This symmetry in the roles played by $\bm u$ and $\bm n$ is broken
by the possibility of constraining $\bm q$, which is an intrinsic property of the fluid, through physical assumptions
(e.g. assuming a perfect fluid), while this is not possible for the foliation-dependent vector $\bm J$.
$\square$
}
\end{proof}
Note that the same result holds under any of the other standard (Weak, Strong, Dominant) Energy Conditions as they all imply the
Null Energy Condition \cite{hawking:structure,wald:relativity}. 


\subsubsection{Boundary terms and global averages}
\label{subsubsec:global_domain}

As previously illustrated (see figure \ref{fig:var_domain}), the spatial motion of $\CD$ in the coordinate system $(t,x^i)$ induces 
a flux of fluid elements with velocity $\bm V$ across the boundary of the domain $\CEt$, coinciding at some instant with $\CD$ 
and transported along the congruence of $\bm{\partial}_t$. In the same line of thoughts, there also exists a flux of fluid elements with
velocity $\bm N + \bm V = N \bm v$ across the boundary of the domain $\CEn$, coinciding with $\CD$ at some instant
and carried along the normal congruence.

The first boundary effect is related to the choice of the spatial coordinates, and it can be made to vanish by adopting a comoving picture. 
The second one is generated by the tilt, that is, the deviation of the fluid $4-$velocity with respect to $\bm n$, that translates into
a tilted motion of the comoving domain boundaries with respect to the normal of the slices. It will be present in general unless
the foliation is fluid-orthogonal,  a foliation choice which is not possible if the fluid has non-vanishing vorticity.
It is this second, coordinate-independent effect that impacts on the time variation of the hypersurface volume of the domain, as one can see upon
writing expression \eqref{eq:vol_D_final} as
\begin{equation}
\small{	\frac{\rm d}{{\rm d}t} \CV_\CD 
		 = \int_{\CD_{\bm x}} - N \CK \sqrt{h} \, {\mathrm d}^3 x
			+ \int_{\CD_{\bm x}} \big( N v^i \big)_{|| i} \sqrt{h} \, {\mathrm d}^3 x 
		 = \int_{\CD_{\bm x}}- N \CK \sqrt{h}\, {\mathrm d}^3 x + \oint_{\partial\CD_{\bm x}} N v^i {\varkappa}_i\,
			{\mathrm d}\varsigma\, , 
	\label{eq:vol_D_VI}}
\end{equation}
where we have used Gauss' theorem for the second equality. Above, $\boldsymbol{\varkappa}$ is the outward-pointing unit normal vector to the
boundary $\partial\CD$, whose surface element is denoted by ${\mathrm d}\varsigma$. This rewriting allows to clearly see how the tilt,
as measured by $\bm v$, contributes as a boundary flux term to the evolution of the domain's volume.

Similar tilt-related boundary terms affect the commutation rule \eqref{eq:com_rule_final} and the evolution equations for the effective 
scale factor \eqref{eq:av_raych}--\eqref{eq:av_hamilt}. They arise from the averages of covariant spatial three-divergences, which are
boundary terms as implied by Gauss' theorem:
\begin{equation}
 \average{A^i_{\ || i}} = \frac{1}{\VD} \int_{\CD_{\bm x}} A^i_{\ || i} \, \sqrt{h} \, {\mathrm d}^3 x
                        = \frac{1}{\VD} \oint_{\partial\CD_{\bm x}} A^i {\varkappa}_i \, {\mathrm d}\varsigma \; ,
\label{eq:gauss_theorem}
\end{equation}
for any spatial vector field $\bm A$.
These effects cannot be neglected in general; for a given fluid flow, their contribution entirely depends on the way
the slices are constructed, which locally affects the lapse and tilt amplitudes, and on the choice of the domain of interest (locally defining a specific boundary).

As an example, let us consider the commutation rule \eqref{eq:com_rule_final}. Successively applying \eqref{eq:gauss_theorem} to
$\bm A = N \bm v = \bm N + \bm V$ and $\bm A = \psi N \bm v$, we can rewrite it for 
any scalar $\psi$ under the following forms:
\begin{eqnarray}
 \frac{\mathrm d}{{\mathrm d}t} \average{\psi} & = & \average{ \frac{\mathrm d}{{\mathrm d}t} \psi } + \average{N \CK } \average{\psi} 
	 - \average{ \left( N \CK - \big( N v^i \big)_{|| i} \right) \psi } \nonumber \\
     & & \qquad - \average{\psi} \frac{1}{\VD} \oint_{\partial\CD_{\bm x}} N v^i {\varkappa}_i \, {\mathrm d}\varsigma  \\
 &=& \average{ N n^\mu \partial_\mu \psi} + \average{N \CK } \average{\psi} - \average{ N \CK \psi} \nonumber \\
     & & + \frac{1}{\VD} \oint_{\partial\CD_{\bm x}} \psi N \, v^i {\varkappa}_i \, {\mathrm d}\varsigma
	 - \average{\psi} \frac{1}{\VD} \oint_{\partial\CD_{\bm x}} N v^i {\varkappa}_i \, {\mathrm d}\varsigma \; ,
\label{eq:com_rule_boundary_terms}
\end{eqnarray}
where the second expression makes use of the total coordinate-time derivative with respect to $\bm n$, instead of $\bm u$ (as in the first expression), replacing $\mathrm{d}/\mathrm{d}t$ by $N n^\mu \partial_\mu$. 

For simplicity, we do not make the boundary contributions explicit in 
the evolution equations for $a_\CD$, although this could be done in the same 
manner. Instead, we illustrate their effect by comparing the set of averaged 
equations in the generic case to a restricted situation where all boundary 
terms cancel out. We consider to this aim the case of topologically 
closed spatial sections (that is, we assume that the hypersurfaces are 
compact three-dimensional manifolds without boundaries), and we extend 
the averaging domain to the whole compact boundary-free hypersurface, which we 
denote by $\mathrm{\Sigma}$.
From \eqref{eq:vol_D_VI}, the evolution of the domain volume becomes in this case:
\begin{equation}
 \frac{1}{\CV_\mathrm{\Sigma}} \frac{{\mathrm d} \CV_\mathrm{\Sigma}}{{\mathrm d}t} = - \gaverage{N \CK} \; ,
\end{equation}
so that the scale factor here satisfies $({\mathrm d} a_\mathrm{\Sigma} / {\mathrm d}t) / a_\mathrm{\Sigma} = - \gaverage{N \CK} / 3$.
Then, from \eqref{eq:com_rule_boundary_terms}, the \textit{extrinsic commutation rule for a global boundary-free averaging domain} can be written under the following equivalent forms:
\begin{eqnarray}
	\frac{\mathrm d}{{\mathrm d}t} \gaverage{\psi} & = & \gaverage{ \frac{\mathrm d}{{\mathrm d}t} \psi } 
		 + \gaverage{ N \CK } \gaverage{\psi} - \gaverage{ \left( N \CK - \big( N v^i \big)_{|| i} \right) \psi }  \; ; \nonumber \\ 
	\frac{\mathrm d}{{\mathrm d}t} \gaverage{\psi} & = & \gaverage{ N \, n^\mu \partial_\mu \psi} + \gaverage{ N \CK } \gaverage{\psi} - \gaverage{ N \CK \, \psi} \; , 
\label{eq:com_rule_global}
\end{eqnarray}
for any scalar $\psi$.

Applying \textit{Theorem} \ref{th:av_evol} to a global domain on topologically closed hypersurfaces ($\CD = \mathrm{\Sigma}$), we infer that the system of evolution equations \eqref{eq:av_raych}--\eqref{eq:av_hamilt} for the extrinsic effective scale factor remains formally unchanged as written,
while the global backreaction terms reduce to the following:
\begin{align}
	\CQ_{\mathrm{\Sigma}} =
		& \, 
		\frac{2}{3} \gaverage{ N^2 \CK^2 - \gaverage{ N \CK }^2 } - 2  \gaverage{ N^2 \CK_{\mathrm{tl}}^2}
		\; ; 
		\label{eq:kin_back_global} \\
	\CP_{\mathrm{\Sigma}} = 
		& \, - \gaverage{ N \CK \, n^\mu \partial_\mu N} - \gaverage{ N^{|| i} \, N_{|| i} }  \; ;
		\label{eq:dyn_back_global} \\ 
	\CT_{\mathrm{\Sigma}} = 
		& \, - 16 \pi G \gaverage{N^2 \left( ( \gamma^2 - 1 ) ( \epsilon + p ) 
		+ 2 \, \gamma v^\alpha q_\alpha + v^\alpha v^\beta \pi_\albe \right) } \; ,
		\label{eq:SE_back_global}
\end{align}
thanks to the vanishing of the averages of spatial covariant divergences (which are boundary terms) on $\mathrm{\Sigma}$. In particular,
for the calculation of the expression of $\CP_\mathrm{\Sigma}$ from the general $\PD$ \eqref{eq:dyn_back}, successive uses of this property provide the following equivalent expressions:
\begin{align}
 \CP_\mathrm{\Sigma} = 
	& \, \gaverage{N N^{|| i}_{\phantom{||i}|| i} - \CK \,\frac{{\mathrm d} N}{{\mathrm d}t}} - \gaverage{N \big(\CK \, N v^i \big)_{|| i}} \; ;
\label{eq:dyn_back_global_var1} \\
 \CP_\mathrm{\Sigma} = & \, - \gaverage{ N^{|| i} \, N_{|| i} + \CK \, \frac{{\mathrm d} N}{{\mathrm d}t}} + \gaverage{\CK \, N v^i N_{|| i}} \; .
\label{eq:dyn_back_global_var2}
\end{align}
The backreaction formulae \eqref{eq:kin_back_global}--\eqref{eq:SE_back_global} can be compared with the expressions in the general case,
\eqref{eq:kin_back}--\eqref{eq:SE_back}: the differences are the boundary contributions to the backreactions, erased when $\CD = \mathrm{\Sigma}$.
These include all explicit contributions of the tilt vector $\bm v$ to the kinematical and dynamical backreactions, which have disappeared
in the above expressions \eqref{eq:kin_back_global}--\eqref{eq:dyn_back_global}. The alternative expressions
\eqref{eq:dyn_back_global_var1}--\eqref{eq:dyn_back_global_var2} for the dynamical backreaction when $\CD = \mathrm{\Sigma}$ show, nevertheless, that the
tilt vector still manifests itself through the difference between coordinate-time total derivatives along the fluid flow ${\rm d}/{\rm d}t$
and along the hypersurface-orthogonal (normal) flow $N n^\mu \partial_\mu$, here applied to the lapse $N$. Moreover, the existence of a tilt still
influences the dynamics of the extrinsic effective scale factor through the stress-energy backreaction, which is unchanged whether the domain has boundaries
or not. Indeed, the stress-energy backreaction is not a boundary effect but instead a manifestation of, e.g., the \textit{local} difference between
the rest frame energy of the fluid and its energy as measured in the normal frames.

The integrability condition and the averaged energy conservation law for an hypersurface-volume average performed over a closed hypersurface 
are, respectively, deduced from relations \eqref{eq:int_condition} and \eqref{eq:av_en_cons} without change. The same terms are involved,
since no explicit three-divergence term appears in these two expressions. However, the backreactions appearing in
the integrability condition should again be replaced by their simplified expressions above.


\subsubsection{Relations to the literature: global averages}
\label{subsubsec:literature_global}

The averaged equations and the commutation rule that we obtained in the particular case $\CD = \mathrm{\Sigma}$ are equivalent to those derived by R\"as\"anen in \cite{rasanen:lightpropagation},\footnote{%
This is not obvious at first glance, due to
a different choice of the scalars that have been averaged, i.e. in contrast to our case the averaged
quantities in \cite{rasanen:lightpropagation} do not involve the factor $N^2$. Hence, the averaged equations do not appear identical to those
obtained in the present work. To see that they are equivalent, the use of the corresponding local equations is necessary. The notations
also differ (mostly because the description adopted in \cite{rasanen:lightpropagation} is explicitly $4-$covariant);
one should take care in particular of the fact that in \cite{rasanen:lightpropagation} the notation $\partial_t$ is used for
the coordinate-time covariant derivative along $\bm n$ (i.e. $N n^\mu \nabla_\mu$ in the notations of the present work),  rather than for the  coordinate-time partial derivative $\partial_t |_{x^i}$.
}
where all averages were taken on the whole boundary-free hypersurface
(which was not assumed to be topologically closed and compact;
instead, the existence of the averages was implied by an assumption
of statistical homogeneity of the spatial hypersurfaces).
The above average equations for the $\CD = \mathrm{\Sigma}$ case
are also identical to those obtained by Tanaka and Futamase
in \cite{futamase:aver2} (following from \cite{futamase:aver1} and supplementing
the corresponding equations with the contributions of the cosmological constant),
while the commutation rule was not explicitly given in these papers.
Periodic boundary conditions were assumed, so that the situation considered
was equivalent to a global averaging over hypersurfaces with a closed
$3$-torus topology. The vanishing shift considered in these papers does not affect
the results since, as seen above, this vector neither contributes to the local nor to the average dynamics.

One also recovers the same averaged equations and commutation rule as in 
section~\ref{subsubsec:global_domain} above by restricting in the same way
the expressions obtained by Brown \textit{et al.} in \cite{brown:aver} to the compact boundary-free domain case (whereas it is not the case
for the results of Larena in \cite{larena:aver} due to the different choice of scale factor). More surprisingly, the averaged and commutation
relations derived by Gasperini \textit{et al.} in \cite{marozzi:aver2} (or by Smirnov in \cite{smirnov:aver} within the same formalism) remain formally similar to the equations we get in our boundary-free $\CD = \mathrm{\Sigma}$ case hereabove, even when applied to a general domain.
This originates from the different
propagation of the averaging domain, which in \cite{marozzi:aver2,smirnov:aver} is chosen to be along the flow of $\bm n$;
accordingly, the natural time derivative in their approach is $N n^\mu \partial_\mu$ (in the notations of the present work). This similitude in the equations
(or, equivalently, the fact that the averaged equations and commutation rule of \cite{marozzi:aver2,smirnov:aver} are formally unchanged
by restricting them to the case $\CD = \mathrm{\Sigma}$) indeed shows that boundary terms only occur when the domain's boundaries follow a tilted flow
with respect to the normal to the hypersurfaces in which the domain is embedded. There is no such tilt in the domain propagation in \cite{marozzi:aver2,smirnov:aver},
hence boundary terms are absent, despite the non-vanishing local tilt vector between the fluid and normal flows. As in our case, this local tilt still
influences the dynamics \textit{via} the difference in energy density and pressure between the local frames orthogonal to each of these flows.


\subsubsection{Relations to the literature: transport of the averaging domain}
\label{subsubsec:literature_transport}

In the more generic case of an averaging domain not covering the whole hypersurface, its time propagation
needs to be specified. Three choices in particular, determined by the three congruences we introduced 
(see figure~\ref{fig:schem_vect_3D}), may appear as `natural' definitions of the transport of the averaging domain.

The first choice is to assume a domain evolving along the congruence of the coordinate 
frames $\bm{\partial}_t$. This is the situation implicitly considered by Larena  
\cite{larena:aver} and Brown \textit{et al.} \cite{brown:aver} (see also the respective 
applications of these papers in \cite{larena:aver_app} and 
\cite{brown:aver_app1,brown:aver_app2}).
Such a construction picks up two important issues: first, given a particular choice of shift, 
the vectors $\bm{\partial}_t$ and $\bm u$ will not be collinear in general, hence there will 
be a flow of fluid elements across the domain boundary. This calls the physical relevance 
of the averaged system into question as the domain will not encompass the same collection 
of fluid elements throughout its evolution, i.e. it will not conserve its rest mass 
content. Second, for the same spacetime and the same 
foliation, the location of the domain at a given time will depend on the choice of the shift 
vector, as it determines the direction of $\bm{\partial}_t$. This leads to an unphysical 
dependency of the averaged system (hence, of all spatial average properties) on the choice 
of the spatial coordinates and on the way they propagate. 

The second choice is to assume a spatial domain evolving along the integral curves 
of the normal frames $\bm n$. This is the configuration considered by Gasperini \textit{et al.} \cite{marozzi:aver2}
(see also the follow-up paper \cite{marozzi:aver3}). Their averaging formalism, as 
introduced in \cite{marozzi:aver1}, is based on the construction of a spacetime window 
function characterizing the averaging domain to be considered, and is written in manifestly 
$4-$covariant form. While this formalism is suitable for a freely specifiable propagation 
of the domain boundaries, the averaged system of equations derived in \cite{marozzi:aver2},
both in $4-$covariant and $3+1$ forms, has assumed a transport along $\bm n$ (see equation (3.2) therein).\footnote{%
	Accordingly, and in contrast to a statement of \cite{marozzi:aver2}, the resulting averaged system of equations, as expressed in $3+1$ form, is \textit{not}  identical to that of Brown \textit{et al.} \cite{brown:aver} for a non-vanishing shift, as in this latter study the domain is transported along $\bm{\partial}_t$. This becomes true if a vanishing shift is chosen, due to the proportionality of $\bm n$ and $\bm{\partial}_t$ in this case. As correctly stated, however, the averaged system of equations in $3+1$ form of \cite{marozzi:aver2} becomes identical to that of Paper II for an irrotational perfect fluid if the fluid rest frames are used to generate the spatial hypersurfaces. This is indeed expected as in this case $\bm n = \bm u$, hence the domain has the same (fluid-comoving) evolution as in Paper II.
}

This choice of propagation was also the one adopted by Smirnov \cite{smirnov:aver} 
and Beltr\'an Jim\'enez \textit{et al.} \cite{dunsby:aver}. 
In these papers, $\bm n$ is assumed to be geodesic and to correspond to the $4-$velocity of an irrotational non-interacting dust contribution
to the stress-energy tensor, in contrast to \cite{marozzi:aver2} where this normal vector was freely specifiable.
The formalism of Smirnov is otherwise close to that of Gasperini \textit{et al.}
\cite{marozzi:aver2}, from which it is directly inspired, with both $4-$covariant and $3+1$ forms of the averaged equations. 
Beltr\'an Jim\'enez \textit{et al.} \cite{dunsby:aver} consider a $3+1$ description, with a 
vanishing shift and a trivial lapse ($N = 1$, allowed by the geodesic assumption on $\bm n$) but still tilted fluid flows, and their domain actually follows both $\bm{\partial}_t$ and $\bm n$ as the vanishing shift makes these two directions identical.

\begin{figure}[!ht]
\center{
		\includegraphics[scale=0.8]{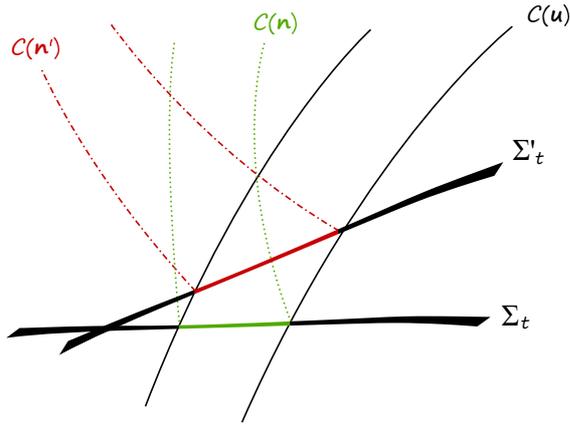}
	}
	\caption{%
		We here illustrate the situation where the propagation of the averaging domain
		is chosen so as to follow the normal of the hypersurfaces at stake. 
		For the foliation of slices $\mathrm{\Sigma}_t$, the domain locus is described by the associated 
		normal congruence $\CC(\bm n)$ (green dotted lines). For another foliation of slices
		$\mathrm{\Sigma}'_{t}$, it is described by the normal congruence $\CC(\bm n')$ (red 
		dash-dotted lines), which differs in general from $\CC(\bm n)$.
		The spatial domains selected in this way by both foliations may be matched to the same set of fluid elements (represented by the fluid congruence $\CC(\bm u)$ in thin continuous black lines) at a given time $t$; such domains are represented by the continuous-line   
		colored sections of the corresponding hypersurfaces, $\mathrm{\Sigma}_t$ and $\mathrm{\Sigma}'_t$. However, at subsequent times the fluid elements collected within both domains will differ.
		Choosing a domain transport along the normal of the hypersurfaces constructs 
		different four-dimensional tubes, corresponding to different physical systems, for different foliations. It will also 
		imply a flow of fluid elements across the domain boundary in general.
	}
	\label{fig:4dtube}
\end{figure}
The choice of a domain transport along $\bm n$ leads to formally simpler averaged equations
in terms of the geometric variables of the foliation
due to the vanishing boundary terms (see section~\ref{subsubsec:literature_global}).
It also makes the propagation of the averaging domain independent of the propagation 
of the spatial coordinates, but this propagation becomes instead dependent on the choice of the foliation 
which defines the vector $\bm n$. One could argue that such a dependence is inherently 
present in any spatial averaging scheme, since the domain of averaging lies by definition 
within the hypersurfaces built from the foliation. However, the dependence we refer to can be 
understood from a spacetime perspective: by changing the foliation, and hence the vector 
$\bm n$, the four-dimensional tube spanned by the domain transported along this vector will 
be different (see figure~\ref{fig:4dtube}). We also note that the 
first drawback mentioned previously for an evolution along $\bm{\partial}_t$
also holds for a transport along $\bm n$: in presence of tilt ($\bm n \neq \bm u$), the particle content of the domain will be altered during its evolution and, as a consequence, the rest mass of the fluid within the domain will
generically not be conserved.

Two generalization schemes to non-fluid-orthogonal foliations have been suggested by R\"as\"anen in \cite{rasanen:lightpropagation} (see also the application \cite{rasanen:lightpropagation_app}), and by Kasai \textit{et al.} in \cite{futamase:aver1} followed by Tanaka \& Futamase in \cite{futamase:aver2}, where such issues related to the propagation
of the domain boundaries are avoided. However, in both cases this requires specific choices of the averaging domain that restrict the scope to large scales and to a class of foliations where the assumptions made in these papers can hold. R\"as\"anen \cite{rasanen:lightpropagation} derives the averaged equations in a $4-$covariant
form for a domain covering the whole (typically non-bounded) hypersurfaces, thus without the need for specifying its propagation. The convergence of the averages for such an infinite domain is ensured by the assumption of statistical homogeneity to hold in these hypersurfaces. In turn, the system of averaged equations obtained by Tanaka \& Futamase in \cite{futamase:aver2} (slightly generalizing that of \cite{futamase:aver1}) requires a domain and foliation where periodic boundary conditions can be assumed.\footnote{%
This system of averaged equations is given in a background-independent scheme 
as a preliminary step in \cite{futamase:aver1,futamase:aver2}. However, the emphasis is subsequently put on linear perturbation theory around a Friedmannian background, on which the main conclusions are based. Accordingly, no or negligible contributions from backreaction are found in this setting, which is expected due to the nonlinear and background-free nature of backreaction. We emphasize that mixing background-dependent applications with a background-free framework may imply strong restrictions, e.g. the small backreaction found by Russ \textit{et al.} \cite{Russ} in second-order perturbation theory with a Friedmannian background must in reality vanish due to the geometric constraints imposed (see the comments in Paper~I \cite{buchert:av_dust}, Sect. 3.4.).} 
The transport of the averaging domain is not specified; this does not affect the results due to the vanishing of any boundary term.
Comparing with \cite{marozzi:aver2} and in view
of the discussion above in section \ref{subsubsec:literature_global}, we conclude that the average equations
obtained in both schemes discussed in this paragraph, \cite{rasanen:lightpropagation,futamase:aver2}, would remain valid in a general foliation, and for any domain, provided it is required that its boundaries propagate along $\bm n$ (which would also be a propagation along $\bm \partial_t$ in \cite{futamase:aver2} in view of the vanishing shift vector choice) in order to prevent the occurrence of additional boundary terms. A wider applicability of the schemes would thus be recovered, but the drawbacks highlighted above for such a propagation would also be retained.

The third choice, which we adopt in the present work, is that of a domain 
comoving with the fluid. As its boundaries follow the fluid flow $\bm u$, the averaging 
domain always sweeps out the same four-dimensional tube of spacetime, whatever the choice of the foliation and spatial coordinates.
This option also ensures, by definition, that the domain encloses the same collection of fluid elements throughout its evolution,
which in turn implies the conservation of the fluid rest mass within $\CD$. 
Choosing such a domain propagation therefore avoids all of the drawbacks mentioned above. 
It should be noted, however, that the advantage of rest mass conservation within the 
spatial domain for a model universe filled with several fluids would not hold (for each fluid), in general. 
A multi-fluid approach would require to 
pick up and follow one preferred fluid congruence, preserving the corresponding rest mass only, while allowing the others 
to flow across the domain boundaries (see, \textit{e.g.}, \cite{ismael:multifluid}), or, to use the joint barycentric velocity for all fluids, ensuring only the preservation of the total combined rest mass.
The rest mass within the domain could be conserved simultaneously for every fluid
only by assuming that the $4-$velocities of all fluids coincide, at least at the domain boundary,\footnote{%
The averaged equations are in general defined for arbitrary domains. If an assumption is adopted that distinct fluid congruences
coincide or ``average out" on the boundary, the arbitrariness of the domain choice has to be given up.%
}
or that the spatial domain is extended to the whole hypersurface. 
In the present work we consider a cosmological model sourced by a single fluid, which 
should satisfactorily account for the description of the main cosmological epochs largely dominated by a 
particular fluid (radiation or dust).


\subsubsection{Relations to the literature: comparison of the final averaged equations}
\label{subsubsec:literature_general}

Most authors cited in the above discussion base their studies either on a direct $3+1$ formulation of the evolution and averaged equations,
or on a formulation using explicitly $4-$covariant terms from which a $3+1$ form is explicitly deduced. This allows for a rather direct comparison
with the formalism and results presented so far in this paper (sections \ref{subsec:av_proc} to \ref{subsec:av_cosmo_extrinsic}).\footnote{%
The averaged energy conservation equation and the integrability condition (see section \ref{subsubsec:av_cond} above) are not always considered. The $3+1$ approach of Beltr\'an Jim\'enez \textit{et al.} \cite{dunsby:aver} differs from the one used here in that it does neither include lapse nor shift, while Tanaka \& Futamase \cite{futamase:aver2} consider a nontrivial lapse along with a vanishing shift.
In the approach of R\"as\"anen \cite{rasanen:lightpropagation},
the formulation is only given in explicitly $4-$covariant terms; also in this case can a $3+1$ formulation be readily deduced, for comparison with the above averaged equations, upon making a coordinate choice including the appropriate time $t$.}

All of the corresponding systems of $3+1$ averaged equations are manifestly different from the one we obtain in section
\ref{subsec:av_cosmo_extrinsic} due to the different propagation of the averaging domain.
However, we notice a formal similarity between the commutation rule
\eqref{eq:com_rule_final} and system of dynamical equations for the effective scale factor \eqref{eq:av_raych}--\eqref{eq:av_hamilt} that
we present, and those of Brown \textit{et al.} \cite{brown:aver}.
The tilt vector weighted by the lapse $N \bm v$ appearing in several terms in our commutation rule and backreaction formulas
would be formally replaced by the shift vector $\bm N$ in the latter paper, both representing the deviation of the vector flow followed by the domain (respectively
$\bm u$ and $\bm{\partial}_t$) to the normal to the slices $\bm n$ in the corresponding framework. Similarly, the time derivative
${\rm d}/{\rm d}t$ along $\bm u$ would be replaced by the time derivative $\partial_t \big|_{x^i}$ along $\bm{\partial}_t$. This allows
to easily see that both systems of equations become equivalent in the special case of a comoving description (within which $N \bm v = \bm N$
and ${\rm d}/{\rm d}t = \partial_t \big|_{x^i}$), as expected since in this case
the spatial coordinates are chosen in such a way that both domains follow the same flow $\bm{\partial}_t \propto \bm u$.

Despite
the same domain propagation choice (also along $\bm{\partial}_t$), the averaged equations of Larena \cite{larena:aver} remain different from the former even in a comoving picture due to a different notion of effective scale factor.\footnote{%
Such additional differences with the results of \cite{larena:aver} arise from a definition of the effective scale factor in this latter study that makes its evolution different from that of the cubic root of the domain's volume. Since the aim of an averaging framework is to investigate the regional dynamics of comoving
domains lying within spatial hypersurfaces, it seems to us to be more appropriate to define the scale factor from the volume of these domains.
The reader may refer to \cite{larena:aver_app} for a comparison of the different averaged energy constraints obtained for different choices
of $a_\CD$, and for an analysis of the backreaction effects obtained for each such choice in a 
Friedmann-Lema\^itre-Robertson-Walker (FLRW) model
perturbed up to second order. Note, however, that in these studies the domain also follows the congruence of the coordinate frames along
$\bm{\partial}_t$, implying the drawbacks already highlighted in section~\ref{subsubsec:literature_transport}.
}
Finally, as already discussed, the choice of a domain propagating along the normal
to the slices (or in the last two cases, the use of global assumptions on the domain that erase boundary terms, yielding the same evolution) made by Gasperini \textit{et al.} \cite{marozzi:aver2},
Beltr\'an Jim\'enez \textit{et al.} \cite{dunsby:aver}, Smirnov \cite{smirnov:aver}, Tanaka \& Futamase \cite{futamase:aver2} and R\"as\"anen \cite{rasanen:lightpropagation}
would require to take either global averages or fluid-orthogonal hypersurfaces (when possible) in each case to make the averaged equations of these papers equivalent to those derived in the above section \ref{subsec:av_cosmo_extrinsic}.

The reader may find a complete comparison of the averaging formalisms discussed above in Appendix~\ref{app:lit_comp} and synthetic tables therein.

\section{Rest mass--preserving scalar averaging: fluid-intrinsic approach}
\label{sec:intrinsic}

In this section we propose an alternative averaging procedure, valid as well for arbitrary spatial foliations, aimed at characterizing average properties that are fully intrinsic to the fluid. 
We start by presenting the motivations for this approach.


\subsection{Motivation for a fluid-intrinsic averaging procedure}

As discussed in the previous section, all results in the literature on the generalization of scalar spatially averaged cosmologies that we compared (see sections~\ref{subsubsec:literature_global}--\ref{subsubsec:literature_general} and Appendix~\ref{app:lit_comp}) abandon the intrinsic fluid averaging approach that was a primary element of  Papers I and II. Instead, the averaging procedures considered are 
built from averaging domains evolving either along the normal congruence of the hypersurfaces of arbitrary foliations, or at constant values of the arbitrary spatial coordinates. We pointed out that such choices raise problems with regards to the foliation- or coordinate-dependent evolution of the domain, and especially the non-conservation of the rest mass of the averaging domain in general situations. These problems are avoided for our choice of a comoving domain of averaging, i.e. of a domain transported along the fluid congruence.

The approach we presented
in section~\ref{sec:extrinsic} complies, however, with the definition of the averaging operation, and with the set of
foliation-related local variables explicitly appearing in the equations, adopted in the aforementioned literature
(although some `mixed' fluid and foliation scalars such as $h^{\mu \nu} \nabla_\mu u_\nu$ have also been used by Larena \cite{larena:aver}).
This \textit{extrinsic} approach could be employed to measure the deviations from the dynamics of a homogeneous-isotropic model universe in a geometric way. It is indeed most naturally expressed in terms of averages of foliation-dependent scalars characterizing the hypersurfaces such as the respective traces of the extrinsic and intrinsic curvatures.
We argue, however, that intrinsic properties of the fluid content such as those quantified by the rest frame kinematic quantities $\Theta$, $\sigma^2$ and $\omega^2$, defined in section~\ref{subsubsec:kin_fluid}, are more relevant for the characterization of an effective cosmological model.

It is not only a philosophical question to consider as a viable cosmology the evolution of an averaged fluid formulated in its own variables, rather than looking at averages `from outside' that mostly focus on geometric properties of the hypersurfaces. The latter point of view risks invoking a quasi-Newtonian understanding of a moving fluid with respect to some fiducial external spacetime. If, as in the aforementioned literature and in our section~\ref{sec:extrinsic} above,  the averaged dynamics and definitions of backreaction terms involve the extrinsic curvature of the slices, the resulting properties depend on derivatives of the normal vector. Even if the tilt measuring the deviation of the normal with respect to the $4-$velocity is small (the Lorentz factor is close to unity), its derivatives can be large. This may lead to a strong foliation dependence of the averaged variables and backreaction terms that is to be considered irrelevant for a cosmological model, since in such an approach these average quantities only characterize properties of a family of extrinsic observers (\textit{cf.} the discussion in \cite{foliationsletter}).
 
Having said this, the reader may point out that focusing on the properties of the fluid congruence is more reminiscent of a $1+3$ (\textit{threading}) point of view. Indeed, we employ in this work a $1+3$ threading formalism, but jointly with a $3+1$ foliation, simply because hypersurfaces are needed for the averaging operation. Going as far as possible toward a fluid-intrinsic description avoids an excessive foliation-dependence of the variables considered. However, this goal
will encounter limitations, since the rest frames of a vortical fluid are not hypersurface-forming. A fully intrinsic construction of effective cosmologies will thus in general require other choices. The foliation at constant fluid proper time, as part of the Lagrangian description (see section \ref{subsubsec:lag_desc}), allows for a spatial averaging over hypersurfaces that are built from the fluid flow itself.
Another possibility that is opened with the intrinsic approach would be to characterize hypersurfaces statistically. This strategy will be discussed in section \ref{subsubsec:beyond}.

As a first step toward an intrinsic approach, we present in Appendix~\ref{app:usual_aver_intrinsic} a re-expression of the extrinsic evolution equations
\eqref{eq:av_raych}--\eqref{eq:av_hamilt} in terms of the intrinsic variables of the fluid. This provides more insight into the contributions of these quantities to the averaged dynamics, in particular
the influence of the vorticity can be better understood, but it also raises additional contributions from the tilt factor $\gamma$.
In the following, we shall go another route aiming at an intrinsic fluid point of view. To this end we 
introduce a slightly different generalization of the fluid-orthogonal averaging formalism of Papers~I and II that will also allow us to derive a more compact form of averaged cosmologies. 
We first motivate this route by contemplating further on the definition and conservation of the rest mass of the fluid within the domain.


\subsubsection{The regional rest mass and its conservation}

We have shown in section \ref{subsec:mass_cons} that the total fluid rest mass within the domain $\CD$, $M_\CD = \int_\CD M^\mu \, \mathrm{d}\sigma_\mu$, with $M^\mu = \varrho u^\mu$ the conserved rest mass flux vector, is preserved in time ($\mathrm{d}M_\CD / \mathrm{d}t = 0$) as a consequence of the domain's fluid-comoving propagation. We have also shown that $M_\CD$ can be expressed in terms of the hypersurface volume and associated averaging 
operator introduced by \eqref{eq:spat_aver} as follows:
\begin{equation}
 	M_\CD = \int_\CD \gamma \varrho \,\sqrt{h} \, \mathrm{d}^3 x = \VD \average{\gamma \varrho} \, .
\end{equation}
The relevant scalar to be integrated over the spatial domain is therefore $\gamma \varrho$, rather than the rest mass density $\varrho$ that could have been expected. Unless the foliation is fluid-orthogonal ($\gamma=1$), the quantity $\int_\CD \varrho \, \sqrt{h} \, \mathrm{d}^3 x = \VD \average{\varrho}$ is not the fluid rest mass within $\CD$ and accordingly is not conserved. Indeed, using the continuity equation 
\eqref{eq:cons_restmass_t} for $\varrho$ as well as the commutation rule \eqref{eq:com_rule_final_theta} and the associated volume evolution rate expression, we have
\begin{equation}
 \ddt{} \left( \VD \average{\varrho} \right) 
		= - \VD \average{\frac{1}{\gamma} \ddt\gamma \varrho} \, ,
\end{equation}
which is nonzero in general. The need to account for the factor $\gamma$ is a consequence of the conserved $\varrho$ being a rest mass density of the fluid in its local rest frames. It is thus a density with respect to the measure of proper volume of the fluid elements, while $\gamma \varrho$ is the corresponding density with respect to the (Lorentz-contracted) normal-frames volume measure $\sqrt{h} \, \mathrm{d}^3 x$ used in the definition of the extrinsic averaging operator $\average{\, \cdot \,}$.

The total fluid rest mass within the domain is alternatively obtained by integrating (still over the domain $\CD$ lying within the arbitrary spatial hypersurfaces) the rest mass density per unit of fluid proper volume, $\varrho$, 
with the corresponding fluid rest frames volume element, $\sqrt{b} \: \mathrm{d}^3x$ with $b := \det(b_{ij})$. Given the relation between the determinants 
$b$ and $h$, 
\begin{align}
 	b 
	& = \det(g_{i j} + u_i u_j) = \det(h_{i j} + u_i u_j) = h \, \det(\delta^i_{\ j} + h^{i k} u_k u_j) \nonumber \\
 	& = h \, (1 + h^{i j} u_i u_j) = h \, (1 + h^{\mu \nu} u_\mu u_\nu) = h \, \gamma^2 \; ,
 	\label{eq:rel_dets}
\end{align} 
we have $\sqrt{b} \: \mathrm{d}^3 x = \gamma \sqrt{h} \, \mathrm{d}^3 x$, and therefore we indeed get
\begin{equation}
 \label{eq:restmass_and_b}
	M_\CD = \int_\CD \varrho \, \sqrt{b} \: \mathrm{d}^3 x \, .
\end{equation}
The rest mass $M_\CD$ of the fluid within $\CD$ is thus more naturally defined in terms of the proper volume measure $\sqrt{b} \: \mathrm{d}^3 x$.

Note that the two covariant\footnote{%
	As $\sqrt{h (t,x^k )} \,{\mathrm d}^3 x$, the fluid-orthogonal volume $3-$form $\sqrt{b (t,x^k )}\, {\mathrm d}^3 x$ is also invariant under a 
	change of spatial coordinates, as can be checked either directly or by rewriting it as $\gamma (t,x^k ) \sqrt{h(t,x^k )} \,{\mathrm d}^3 x$, 
	$\gamma = -n^\mu u_\mu$ being a $4-$scalar. It reads in particular 
	$\sqrt{b (t,x^i )} \,{\mathrm d}^3 x = \sqrt{b (t, f^i(t,\bm X))} \, J(t,X^i) \,{\mathrm d}^3 X$ in comoving spatial coordinates $X^i$,
	with $b (t, f^i(t,\bm X)) \, J(t,X^i)^2$ being the determinant of the spatial components of the fluid rest frame projector $\mathbf{b}$ in the 
	comoving coordinate system $(t,X^i)$.
}
volume measures $\sqrt{h} \, \mathrm{d}^3 x$ and $\sqrt{b} \: \mathrm{d}^3 x$ coincide in the case of a flow-orthogonal foliation (possible for an irrotational 
fluid), which is the situation considered in Papers~I and II. Hence, a degeneracy between both volumes is present in these papers, while they are distinct for any other choice of foliation. This is similar to the difference between hypersurface-orthogonal and fluid-comoving propagation choices for the averaging domain, that emerges outside the fluid-orthogonal foliation framework of Papers~I and II (where both choices can be made simultaneously). We have argued above that once this distinction needs to be done, preserving the comoving character of the domain propagation is the relevant choice for a physical description of average properties of a regional subset of the fluid. Here we also notice that keeping a volume measure that corresponds to a proper volume for the fluid appears to be the most suited to describe the integrated contribution of variables that are primarily defined from the fluid rest frames, as, \textit{e.g.}, for the expression of the total rest mass within the domain from the rest mass density $\varrho$.

We shall accordingly introduce a new volume for the domain and a new averaging operator based on the fluid proper volume element. This will allow us to define these notions intrinsically from the source content, leaving only the integration itself as based on the foliation choice since the spatial integration domain lies within a hypersurface. (However, in sections \ref{subsec:lagrangian_form} and \ref{subsubsec:lagrange_interest} we shall emphasize the choice of a fluid proper time foliation, for which, in particular, the hypersurfaces are themselves also defined intrinsically from the fluid, up to the choice of an initial hypersurface.)
We will also recover the expected relation between total rest mass and averaged rest mass density.


\subsubsection{Intrinsic averaging operator}
\label{subsubsec:int_av_def}

We consider as before a compact domain $\CD$ transported along the fluid flow lines and contained within the arbitrary  spatial hypersurfaces normal to the unit timelike 
vector field $\bm n$. Instead of using the hypersurface domain volume $\VD^h$ as in the previous section, where the superscript $h$ is added here for clarity, we introduce 
the total proper volume of the fluid elements within $\CD$ (or \textit{intrinsic volume} or \textit{fluid volume} of $\CD$) as an integral of the fluid proper volume element over $\CD$: 
\beq \label{eq:vol_D_intrinsic}
	\VDb(t) := \int_{\CD} u^\mu {\rm d}\sigma_\mu = \int_{\CD} \gamma (t,x^i ) \sqrt{h (t,x^i )} \, 
	{\mathrm d}^3 x = \int_{\CD} \sqrt{b (t,x^i )} \, {\mathrm d}^3 x \; .
\eeq
We then define the \textit{intrinsic (or fluid-volume) average} over $\CD$ of any scalar $\psi$ as:
\begin{eqnarray}
 \baverage{\psi} := \frac{1}{\VDb} \int_\CD \psi \, u^\mu {\mathrm d}\sigma_\mu
	& = & \frac{1}{\VDb} \int_\CD \psi (t,x^i ) \, \gamma (t,x^i ) \sqrt{h (t,x^i )} \, {\mathrm d}^3 x \nonumber \\
	& = &\frac{1}{\VDb} \int_\CD \psi (t,x^i ) \, \sqrt{b (t,x^i )} \, {\mathrm d}^3 x \; . \label{eq:spat_aver_intrinsic}
\end{eqnarray}
These volume and averages definitions simply differ from the extrinsic ones introduced in section~\ref{sec:extrinsic} through the volume $3-$form they use, being now built from $\bm u$ instead of $\bm n$. We detail this change in the volume $3-$form further in Appendix~\ref{app:volumes}. We also briefly recall there the possible reformulation of both averaging schemes under a manifestly $4-$covariant form (see \cite{asta1}, extending \cite{marozzi:aver2}), for which the transition from the extrinsic to the intrinsic operators also takes a rather natural form.

Similarly to the extrinsic hypersurface averager of section~\ref{sec:extrinsic}, we recover from \eqref{eq:vol_D_intrinsic} and \eqref{eq:spat_aver_intrinsic} the volume and averager of Papers~I and II when considering a foliation orthogonal to an irrotational fluid flow. The two averaging schemes can be formally related as follows:
\begin{equation}
 \label{eq:extr_intr_avg_rel}
 	\VDb = \VD^h \left\langle \gamma \right\rangle^h_\CD \;\; ; \qquad 
	\baverage{\psi} = \frac{\left \langle \gamma \psi \right\rangle_\CD^h}{\left\langle \gamma \right\rangle_\CD^h} \;\; ,
\end{equation}
for any scalar $\psi$, where we label the extrinsic averaging operator used throughout section \ref{sec:extrinsic} with a superscript $h$ for a more explicit distinction. This shows the identity of both operators in the absence of tilt ($\gamma = 1$), and their approximate identity in the case of a small tilt, \textit{i.e.}, of non-relativistic Eulerian velocities of the fluid in the chosen foliation ($\gamma \simeq 1$). This also shows, on the other hand, that we must have $\VD^h < \VDb$ if the tilt does not identically vanish inside $\CD$, which can be seen as a consequence of the local Lorentz contraction of the volume of each fluid element when measured in the normal frames.

By construction, from the expression \eqref{eq:restmass_and_b} for $M_\CD$, the averaged rest mass density is naturally expressed as the ratio of the total rest mass to the total volume, if the intrinsic averaging operator and fluid volume are used:
\begin{equation}
 \baverage{\varrho} = \frac{M_\CD}{\VDb} \; .
\end{equation}
While such a relation is expected, it is to be compared to the case where the extrinsic averaging operator and volume are used, as in section~\ref{sec:extrinsic}: in this case, $M_\CD / \CV_\CD^h$ is instead equal to $\left \langle \gamma \varrho \right \rangle_\CD^h$, as can be seen from \eqref{eq:rest_mass_scalar}.

Despite the intrinsic averaging operator and associated volume presenting such natural features, we will still denote them with a label ${}^b$ in the following to underline the difference with the extrinsic formalism of section~\ref{sec:extrinsic} above and previous proposals from the literature that are closely related to the latter. Also for this reason, we use no specific label for the extrinsic formalism, unless a label ${}^h$ is needed as hereabove for a clearer distinction between both approaches.


\subsection{Intrinsic effective dynamics of general fluids seen in general foliations}
\label{subsec:av_cosmo_intrinsic}


\subsubsection{Fluid-intrinsic volume and averager: time evolution}

The evolution rate of the fluid volume $\CV_\CD^b$ can be derived in the same way as it was done for the hypersurface Riemannian domain volume in section \ref{subsubsec:vol_lagrange},
changing the spatial coordinates to comoving ones in the integral to commute integration and comoving coordinate-time derivative. Using the invariance
of the fluid rest frame volume form with respect to such a spatial diffeomorphism, we then get:
\begin{equation}
 \label{eq:vol_rate_intrinsic}
 \frac{1}{\VDb} \ddt{}\VDb = \baverage{\frac{N}{\gamma} \Theta} = \baverage{\Thetat} \; ,
\end{equation}
through the first equality of relation \eqref{eq:exp_tensor_comoving} holding in comoving coordinate systems. We have introduced above the rescaled scalar expansion rate $\Thetat := (N / \gamma ) \Theta$. Since $N/ \gamma = \mathrm{d}\tau / \mathrm{d}t$, $\Thetat$ can be seen as the fluid's local expansion rate with respect to the coordinate time $t$, while $\Theta$ expresses this rate with respect to the proper time $\tau$.

An alternative derivation can be obtained by starting from the reformulation of the extrinsic averaging scheme in Appendix \ref{app:usual_aver_intrinsic}. Using the extrinsic volume evolution rate \eqref{eq:vol_rate_II}, the associated commutation rule~\eqref{eq:com_rule_final_II}, and the above relations between both averaging schemes \eqref{eq:extr_intr_avg_rel}, the above evolution rate is recovered. 

Both methods can be equally used to obtain a new commutation rule for the intrinsic averager, which we now express in the form of a \textit{Lemma}.

\begin{lemma}[commutation rule for fluid-intrinsic volume averages]\\
\label{lemma:com_rule_intrinsic} 
\vspace{-7pt}

The commutation rule between fluid-intrinsic averaging on a compact domain $\CD$, lying within the constant-$t$ hypersurfaces and comoving with the fluid, and comoving differentiation with respect to the coordinate time reads, for any $3+1$ foliation of spacetime and for any scalar $\psi$:
\beq
	\frac{\mathrm d}{{\mathrm d}t} \baverage{\phantom{\tilde |}\!\!\psi} 
		= \baverage{ \frac{\mathrm d}{{\mathrm d}t} \psi } 
		- \baverage{ \Thetat } \baverage{\phantom{\tilde |}\!\!\psi} 
		+ \baverage{ \Thetat \, \psi } \,  . 
	\label{eq:com_rule_intrinsic}
\eeq
\end{lemma}

\bigskip

This simple relation is again independent of the shift due to the spatial coor\-dinates-independent definitions of the domain propagation and of the averaging procedure. It only depends on the lapse and the tilt through the threading lapse factor $N/\gamma$ in $\Thetat$, rescaling the proper time evolutions to coordinate-time evolutions (see Appendix~\ref{app:threading}).

From the volume evolution rate \eqref{eq:vol_rate_intrinsic}, the commutation rule \eqref{eq:com_rule_intrinsic} and the local continuity equation $\mathrm{d} \varrho / \mathrm{d}t + \Thetat \varrho = 0$, we obtain $\mathrm{d}(\VDb \, \langle \varrho \rangle_\CD^b)/ \mathrm{d}t = 0$,
which shows again the preservation of the domain rest mass $\MD = \VDb \, \langle \varrho \rangle_\CD^b$.


\subsubsection{Averaged evolution equations}
\label{subsubsec:av_evol_intrinsic}

We define the \textit{intrinsic (or fluid-volume) effective scale factor} of the fluid body within the domain $\CD$
\textit{via} the intrinsic domain volume:  
\begin{equation}
	a^b_\CD (t) := \left( \frac{\VDb (t)}{\CV^b_{\CD_\mathbf{i}}} \right)^{1/3} \; ,
\end{equation}
so that its rate of change yields the averaged fluid expansion rate as seen in coordinate time $t$: 
\begin{equation}
	H_\CD^b \;:=\; \frac{1}{a^b_\CD} \frac{{\mathrm d} a^b_\CD}{{\mathrm d}t} 
		\, = \, \frac{1}{3} \baverage{ \Thetat } \, .
	\label{eq:evol_aD_intrinsic}
\end{equation}
Equivalently, the rate of change of the intrinsic scale factor can be written as
\beq
	H_\CD^b = \baverage{\frac{1}{\ell} \frac{{\mathrm d}\ell}{{\mathrm d}t}} \, ,
\eeq
with $\ell$ being the representative length lying in the rest frames of the fluid, defined in section~\ref{subsubsec:comov_desc} as satisfying $\dot \ell / \ell = \Theta / 3$, \textit{i.e.}, $\ell^{-1} \, \mathrm{d}\ell / \mathrm{d}t = \tilde \Theta / 3$.
Accordingly, $\ell$ represents the spatial isotropic deviation of two neighbouring fluid elements.\footnote{%
The difference to the averager used in section~\ref{sec:extrinsic} can be made explicit by introducing $l$ as the counterpart of $\ell$:
\[
		\frac{1}{a_\CD} \ddt{a_\CD} = \average{\frac{1}{l} \ddt{l}} \;\; \text{with} \;\; \frac{1}{l} \frac{{\mathrm d} l}{{\mathrm d}t} := \frac{1}{3} \left( \frac{N}{\gamma} \Theta - \frac{1}{\gamma} \frac{{\mathrm d} \gamma}{{\mathrm d}t} \right) \, 
		= 	\frac{1}{\ell} \frac{{\mathrm d} \ell}{{\mathrm d}t} - 
			\frac{1}{3} \frac{1}{\gamma} \frac{{\mathrm d} \gamma}{{\mathrm d}t}  \; .
\]
We thus have $l^3 = \ell^3 / \gamma$, i.e. $l$ is an isotropically averaged length (cubic root of a volume) associated with the
volume contraction of $\ell^3$ by the Lorentz factor $\gamma$: lengths are contracted by $\gamma$ in one
spatial direction and are not affected in the other orthogonal two directions, implying a factor $\gamma^{1/3}$ for the isotropically averaged length contraction. In comoving spatial coordinates, one can see from the first relation in \eqref{eq:exp_tensor_comoving} that $\ell$ may be chosen as $\ell \propto b^{1/6}$ from the determinant $b$, so that, accordingly, $l$ may be chosen as $l \propto h^{1/6}$. The continuity equation \eqref{eq:cons_restmass_density} also shows that one may define $\ell$ as $\ell \propto \varrho^{-1/3}$, in any spatial coordinates.
}

Instead of using the Einstein equations projected along $\bm n$, yielding equations \eqref{eq:evol_K} and
\eqref{eq:hamilt_const} (expressed in terms of the intrinsic and extrinsic curvatures of the hypersurfaces), we here express
the local dynamics of the fluid directly through the Raychaudhuri equation, obtained from a projection of the Einstein equations along $\bm u$:
\beq
 \dot\Theta = -\frac{1}{3} \Theta^2 - 2 \sigma^2 + 2 \omega^2 + \nabla_\mu a^\mu - 4 \pi G \left(\epsilon + 3 p \right) + \cc \; .
 \label{eq:raych_eq}
\eeq
This equation relates rest frame kinematic and dynamical scalars of the fluid, and is thus relevant for the present fluid-focussed approach. It can be complemented by an analogue in terms of fluid-intrinsic quantities
of the foliation-related energy constraint \eqref{eq:hamilt_const} by defining a `fluid rest frame $3-$curvature' scalar $\SR$ from the $4-$Ricci tensor $R_{\mu \nu}$ and scalar $R$, following Ellis et al. \cite{ellis:vorticity}, as follows:
\begin{equation}
 \SR := \nabla_\mu u^\nu \,\nabla_\nu u^\mu - \nabla_\mu u^\mu \,\nabla_\nu u^\nu + R + 2 \,R_{\mu \nu} \, u^\mu u^\nu \; .
 \label{eq:def_u_3_curv}
\end{equation}
Noting that the covariant derivatives above can be equivalently replaced by their projections orthogonal to $\bm u$
($\nabla_\rho u^\sigma \mapsto b_\rho^{\ \kappa} b^\sigma_{\ \tau} \nabla_\kappa u^\tau$), the scalar Gauss equation
\cite{alcub:foliation,gourg:foliation}  applied to the $\bm u$-orthogonal hypersurfaces when those exist (i.e.
for vanishing vorticity) shows that $\SR$ corresponds in this case to the scalar intrinsic curvature of these hypersurfaces.
For non-zero vorticity, such hypersurfaces cannot be built, and $\SR$ is not transparently interpreted as a scalar curvature.\footnote{%
However, $\SR$ can indeed arise as the $3-$Ricci scalar associated to a $\bm u$-orthogonal spatial `Riemann-like' tensor
which can be built from the $\bm u$-orthogonal spatial covariant derivative operator (defined for tensors fully orthogonal to $\bm u$
as the projection through $\mathbf{b}$ on every component of their covariant $4-$derivative) as well as from its spacetime embedding
\cite{massa:1+3,ellis:vorticity,roy:1+3}. For non-vanishing vorticity,
this Riemann-like tensor does not possess all the symmetry properties of a true Riemann tensor, and the way of defining
such a spatial curvature tensor is not unique. Despite of this, $\SR$ may be seen as the scalar part of local $3-$curvature associated with this tensor in the $\bm u$-orthogonal
subspace of the tangent space at each spacetime point. Boersma and Dray introduce so-called parametric manifolds to define this quantity as the curvature of the parametric submanifold \cite{boersmadray}. Alternatively, we may see it simply as a definition through equation~\eqref{eq:gauss_eq_intr}.
}
In particular, it should be kept in mind that it does not in general correspond to the intrinsic scalar curvature $\CR$ of the $\bm n$-orthogonal hypersurfaces in which the domain $\CD$ is embedded. 

Inserting the trace of the Einstein equations and their projection along $\bm u$
in the definition \eqref{eq:def_u_3_curv} of $\SR$ allows us to relate it to the fluid rest frame energy density within a constraint equation where the covariant derivative of $\bm u$ has been decomposed into its kinematic parts:
\beq
 \frac{2}{3} \Theta^2 - 2 \sigma^2 + 2 \omega^2 + \SR = 16 \pi G \epsilon + 2 \cc \, .
 \label{eq:gauss_eq_intr}
\eeq
Analogously to what has been done in section~\ref{sec:extrinsic} within the extrinsic averaging scheme, we can now apply the fluid-intrinsic averager
to equations \eqref{eq:gauss_eq_intr} and \eqref{eq:raych_eq} multiplied by $(N/\gamma)^2$ and use expression
\eqref{eq:vol_rate_intrinsic} for the evolution rate of $\VDb$ as well as the commutation rule \eqref{eq:com_rule_intrinsic},
to obtain the effective evolution equations of the intrinsic scale factor $a^b_\CD$. We formulate them in the following \textit{Theorem} in terms of rescaled variables defined similarly to $\Thetat$: rescaled kinematic variables, $\sigmat^2 := (N/\gamma)^2 \sigma^2$
and $\omegat^2 := (N/\gamma)^2 \omega^2$, dynamical variables, $\epsilont := (N/\gamma)^2 \epsilon$
and $\tilde p := (N/\gamma)^2 p$, acceleration $4-$divergence, $\tilde{\CA} := (N/\gamma)^2 \CA$
with $\CA := \nabla_\mu a^\mu$, and fluid $3-$curvature, $\SRt := (N/\gamma)^2 \SR$.
\begin{thgroup}
\label{ths:av_intrinsic}
\begin{theorem}[fluid-intrinsically averaged evolution equations]\\
\label{th:av_evol_intrinsic}
\vspace{-7pt}

The evolution equations for the intrinsic effective scale factor of the fluid body within a compact and comoving regional spatial domain $\CD$ of an inhomogeneous general fluid, and for any $3+1$ foliation of spacetime, read:
\begin{align}
	3 \, \frac{1}{a^b_\CD} \frac{{\mathrm d}^2 a^b_\CD}{{\mathrm d}t^2} = 
		& \; - 4 \pi G \baverage{ \, \epsilont + 3 \tilde p \, } + {\tilde\Lambda}^b_\CD + \tilde{\CQ}_\CD^b
			+ \tilde{\CP}_\CD^b \, ;  
\label{eq:av_raych_intr} \\
	3 \,  \big( H_\CD^b \big)^2  = 
		& \; 8 \pi G \baverage{ \, \epsilont \, } + {\tilde\Lambda}^b_\CD - \frac{1}{2} \tilde{\CQ}_\CD^b - \frac{1}{2} \baverage{\SRt} \, , 
\label{eq:av_hamilt_intr}
\end{align}
with a time- and scale-dependent contribution from the cosmological constant,
\beq
\label{tildelambda}
{\tilde\Lambda}^b_\CD := \cc \baverage{\frac{N^2}{\gamma^2}}\;,
\eeq 
and with $\tilde{\CQ}^b_\CD$ and $\tilde{\CP}^b_\CD$
denoting the intrinsic kinematical and dynamical backreaction terms, respectively, as seen in the $t$-hypersurfaces. They are defined as follows:
\begin{align}
	\tilde{\CQ}^b_\CD 
		& := \frac{2}{3} \baverage{ \left( \Thetat - \baverage{\Thetat} \right)^2 }
			- 2 \baverage{ \sigmat^2 } + 2 \baverage{ \omegat^2 } \, ;
 \label{eq:kin_back_intr} \\ 
	\tilde{\CP}^b_\CD 
		& := \baverage{\tilde\CA} + \baverage{\Thetat \frac{\gamma}{N} \ddt{} \left(\frac{N}{\gamma} \right)}
 \, . 
 \label{eq:dyn_back_intr}
\end{align}
\end{theorem}

As for \textit{Theorem}~\ref{th:av_evol}, the left-hand side of equation \eqref{eq:av_raych_intr} above should not be directly interpreted as a proper-time acceleration of the scale factor, unless a framework such as the Lagrangian picture, that we develop below, is adopted. This is again related to the dependence of individual terms in the above equations on the choice of the time coordinate for a given foliation, while the equations as a whole remain fully covariant (see the discussion and proof in section~\ref{subsubsec:av_evol}).

Note also that the backreaction terms introduced above do not correspond in general to the terms $\QD$ and $\PD$ appearing in the extrinsic averaging scheme. They do coincide, however, in case of a fluid-orthogonal foliation (with, moreover, $\TD = 0$ in this case) as can be seen by direct comparison with the definitions
\eqref{eq:kin_back}--\eqref{eq:dyn_back} of $\QD$ and $\PD$, and by noting that in this case $\CK_{i j} = - \Theta_{i j}$, $\omega^2 = 0$,
and (through relation \eqref{eq:eulerian_acc} between lapse and acceleration of the normal frames), $\CAt = N N^{||i}_{\ \ ||i}$.

The above system of averaged equations can alternatively be derived (through relations \eqref{eq:extr_intr_avg_rel} between
both averaging schemes) from the analogous relations for the extrinsic effective scale factor $a_\CD$, provided the latter
relations are re-expressed in terms of the fluid rest frame local kinematic and dynamical variables, as exposed in
Appendix~\ref{app:usual_aver_intrinsic}. The use of the local dynamical equations \eqref{eq:gauss_eq_intr} and \eqref{eq:raych_eq} is still required in the process since the local quantities to be averaged differ between both schemes by a factor $\gamma$.


\subsubsection{Integrability and energy balance conditions}
\label{subsubsec:av_cond_intr}

As for the extrinsic averaging formalism (see section \ref{subsubsec:av_cond}), a condition of integrability of the system
of averaged equations \eqref{eq:av_raych_intr}--\eqref{eq:av_hamilt_intr} can be obtained by applying the coordinate-time derivative
$\mathrm{d}/\mathrm{d}t$ to the averaged constraint equation \eqref{eq:av_hamilt_intr} and by inserting
$2 \, H_\CD^b \times ( \eqref{eq:av_raych_intr} - \eqref{eq:av_hamilt_intr})$ into the result. The averaged fluid source terms
appearing in the resulting condition are themselves constrained by the local energy balance equation \eqref{eq:en_cl}, which can be rescaled by
a factor $(N/\gamma)^3$ to yield:
\begin{equation}
 \frac{\mathrm d}{{\mathrm d}t} \epsilont + \Thetat \left( \epsilont + \tilde p \right) 
		= 2 \, \epsilont \, \frac{\gamma}{N} \ddt{} \left(\frac{N}{\gamma} \right)
		- \frac{N^3}{\gamma^3} \left( q^\mu a_\mu + \nabla_\mu q^\mu + \pi^{\mu \nu} \sigma_{\mu \nu} \right) \, .
\end{equation}
Applying to it the intrinsic averager, the commutation rule \eqref{eq:com_rule_intrinsic} yields an evolution equation
for $\langle \epsilont \rangle_\CD^b$, which we express along with the integrability condition in a second part of the above \textit{Theorem}.
\vspace{10pt}
\begin{theorem}[integrability and energy balance conditions to \ref{th:av_evol_intrinsic}]\\
\label{th:av_cond_intrinsic}
\vspace{-7pt}

A necessary condition of integrability of equation \eqref{eq:av_raych_intr} to yield equation \eqref{eq:av_hamilt_intr} is given by: 
\color{black}
\begin{eqnarray}
 \ddt{} \tilde{\CQ}_\CD^b + 6 H^b_\CD \tilde{\CQ}_\CD^b + \ddt{} \baverage{\SRt}
			+ 2 H^b_\CD \baverage{\SRt} + 4 H^b_\CD \tilde{\CP}_\CD^b \nonumber \\
	= 16 \pi G \left( \ddt{} \baverage{\epsilont}
			+ 3 H^b_\CD \baverage{\epsilont + \tilde p} \right)
			+ 2 \ddt{} {\tilde\Lambda}_\CD^b \; , \qquad
	\label{eq:int_condition_intr}
\end{eqnarray}
where the source terms on the right-hand side obey an averaged energy balance equation:
\begin{eqnarray}
 \ddt{} \Big\langle \epsilont \Big\rangle_\CD^b \!\! + 3 H^b_\CD \Big\langle \epsilont + \tilde p \Big\rangle_\CD^b \!\!
			= \baverage{\Thetat} \!\! \Big\langle \tilde p \Big\rangle_\CD^b \!\!
			- \baverage{\Thetat \, \tilde p} \quad \qquad \qquad \qquad \nonumber \\
	- \baverage{\frac{N^3}{\gamma^3} (\nabla_\mu q^\mu + q^\mu a_\mu + \pi^{\mu \nu} \sigma_{\mu \nu})}
			+ 2 \baverage{\epsilont \, \frac{\gamma}{N} \ddt{} \left( \frac{N}{\gamma} \right) } \; . 
 \label{eq:av_en_cons_intr}
\end{eqnarray}
This balance equation can be supplemented by the rest mass conservation law $\mathrm{d}\MD/\mathrm{d}t = 0$, which can be equivalently expressed in terms
of the averaged rest mass density $\baverage{\varrho} = \MD/\VDb$:
\begin{equation}
	\frac{\rm d}{{\rm d}t} \baverage{\varrho} + 3 H^b_\CD \baverage{\varrho} = 0 \, .
 \label{eq:av_mass_density_cons}
\end{equation}
\end{theorem}
\end{thgroup}


\subsection{Effective forms of the fluid-intrinsic cosmological equations}
\label{subsec:eff}

We now introduce effective forms of the fluid-intrinsically averaged equations providing compact expressions that are suitable for applications.


\subsubsection{Effective Friedmannian form}
\label{subsubsec:eff_friedmann}

Following the suggestion in Paper II, the set of equations given in \textit{Theorem} \ref{ths:av_intrinsic},
which differs from the standard Friedmann equations, can be seen as a (scale-dependent) Friedmannian dynamics sourced by an \textit{effective} energy-momentum tensor. The corresponding effective, time-dependent energy density and pressure for a given domain $\CD$ are defined as:
\begin{eqnarray}
	\epsilon^b_{\rm eff} (t) & := & \baverage{\epsilont} - \frac{1}{16 \pi G} \tilde\CQ_\CD^b - \frac{1}{16 \pi G} \tilde\CW_\CD^b + \frac{1}{8 \pi G} \tilde\CL^b_\CD   \, ;
 	\label{eq:def_epsilon_eff} \\
	p^b_{\rm eff} (t) & := & \baverage{\tilde p} - \frac{1}{16 \pi G} \tilde\CQ_\CD^b + \frac{1}{48 \pi G} \tilde\CW_\CD^b - \frac{1}{8 \pi G} \tilde\CL^b_\CD - \frac{1}{12 \pi G} \tilde\CP_\CD^b \, , \qquad \qquad
	\label{eq:def_p_eff}
\end{eqnarray}
where we have introduced the backreaction terms $\tilde\CW_\CD^b$, for the deviation of the averaged rescaled fluid $3-$curvature $\langle \SRt \rangle^b_\CD$ from a constant-curvature behaviour, and $\tilde\CL_\CD$ for the deviation of $\tilde \Lambda_\CD^b$ from the cosmological constant $\cc$:\footnote{%
In the standard cosmological model it is assumed that the cosmological constant $\cc$ models Dark Energy; the averaged equations show that we then also have to account for \textit{Dark Energy backreaction} $\tilde \CL^b_\CD$ in cases where $N \ne \gamma$ and $\cc \neq 0$, \textit{cf.} \eqref{tildelambda}. Note that a change of time parameter within a given foliation, $t \mapsto T(t)$, changes the value of this term at each time according to an affine transformation as the lapse factor in $\tilde \Lambda^b_\CD$ gets rescaled.
}
\begin{equation}
\label{deviationfields}
 \tilde\CW_\CD^b := \baverage{\SRt} - 6 \frac{k_{\initial\CD}}{(a^b_\CD)^2} \quad;\quad
  \tilde\CL^b := \tilde\Lambda^b_\CD - \cc\;.
  \end{equation}
$k_{\initial\CD}$ is a domain-dependent constant that can be arbitrarily set for each $\CD$. It may for instance be defined for a given domain as $k_{\initial\CD} = \langle \SRt \rangle^b_\CD (t_{\mathbf i})  / \, 6$ (recall that $a_\CD^b(\initial t) = 1$ by definition), so that $\tilde\CW_\CD^b(t_{\mathbf i}) = 0$, and $\tilde\CW_\CD^b(t)$ represents the deviation of the average rescaled curvature $\langle \tilde \SR \rangle^b_\CD (t)$ from $\langle \tilde \SR \rangle^b_\CD (t_{\mathbf{i}}) /a(t)^2$. It may also be defined instead on a given domain $\CD_H$, corresponding to some large scale of homogeneity, as a global constant $k_{\CD_H}=:k$, \textit{e.g.} for comparison with a given fiducial FLRW model, including the flat, $k=0$ case (see also the remarks of footnote~\ref{fn:int_constants} in section \ref{subsec:newton}).

Equations \eqref{eq:av_raych_intr}--\eqref{eq:av_hamilt_intr} and \eqref{eq:int_condition_intr} can then be written as Friedmann-like equations for the effective sources and
the effective Hubble function $H_\CD^b$, summarized in the following \textit{Corollary} to \textit{Theorem}~\ref{ths:av_intrinsic}.

\begin{corgroup}
\setcounter{corgroupcount}{1}
\refstepcounter{corgroupcount}
\label{cor:intrinsic_effective_all}
\begin{corollary}[Effective Friedmannian form]\\
\label{cor:intrinsic_effective_Friedmann}
\vspace{-7pt}

The set of effective cosmological evolution equations of \textit{Theorem} \ref{th:av_evol_intrinsic} can be written in Friedmannian form for the effective sources \eqref{eq:def_epsilon_eff} and \eqref{eq:def_p_eff}:
\begin{eqnarray}
 3 \, \frac{1}{a^b_\CD} \frac{{\rm d}^2 a^b_\CD}{{\rm d}t^2} & = & - 4 \pi G \,(\epsilon^b_{\rm eff} + 3 \, p^b_{\rm eff}) + \cc \; ;
  \label{eq:friedmann_raych}  \\
 3 \, \big( H_\CD^b \big)^2 & = & 8 \pi G \,\epsilon^b_{\rm eff} - 3 \frac{k_{\initial\CD}}{(a^b_\CD)^2} + \cc \, ,
  \label{eq:friedmann_hamilt} 
\end{eqnarray}
while the integrability condition \eqref{eq:int_condition_intr} reduces to the effective conservation equation:
\begin{equation}
 \ddt{} \epsilon^b_{\rm eff} + 3 \, H_\CD^b \, \left( \epsilon^b_{\rm eff} + p^b_{\rm eff} \right) = 0 \; .
 \label{eq:friedmann_cons}
\end{equation}
\end{corollary}


\subsubsection{Effective scalar field form}
\label{subsubsec:eff_morphon}

Looking at the effective sources \eqref{eq:def_epsilon_eff} and \eqref{eq:def_p_eff}, we appreciate that the
kinematical backreaction term $- \tilde\CQ_\CD^b / (16 \pi G)$ individually obeys an effective \textit{stiff equation of state}, \textit{i.e.}, its contributions $p^b_\CQ $ and $\epsilon^b_\CQ$ to the effective pressure and energy density (respectively) obey $p^b_\CQ = \epsilon^b_\CQ$. The curvature deviation term $- \tilde\CW_\CD^b / (16 \pi G)$, on the other hand, individually obeys an effective \textit{curvature equation of state}, $p^b_\CW = - \epsilon^b_\CW / 3$ (with similar notations), and the Dark Energy backreaction term $\tilde\CL^b_\CD / (8 \pi G)$ obeys an effective \textit{Dark Energy equation of state}, $p^b_\CL = -\epsilon^b_\CL$. The dynamical backreaction term 
$-\tilde\CP_\CD^b / (12 \pi G)$ arises as an additional effective pressure. 
These considerations motivate the introduction of a scalar field language, since a free scalar field in the fluid analogy obeys a stiff equation of state, while the scalar field potential contributes with opposite signs to the expressions for the energy density and pressure.

The backreaction terms (by definition only time-dependent, as spatial averages) can be represented, for each domain $\CD$, by an effective time-dependent scalar field $\tilde \Phi_\CD(t)$, the \textit{morphon field}, as introduced in \cite{morphon}. The 
resulting Friedmann-like equations are sourced in this description by the following effective time-dependent energy density and pressure:\footnote{In the paper introducing the morphon field \cite{morphon}, the possibility of phantom energies has been discussed too, which in this effective picture does not violate energy conditions. We have omitted this possible parametrization here.}
\beq \label{morphonminimal}
	\epsilon^b_{\rm eff} (t)= \baverage{\epsilont}(t) + \epsilon^{\Phi, b}_{\rm eff}(t)\ \ ; \ \ 
	p^b_{\rm eff}(t) = \baverage{\tilde p}(t) + p^{\Phi, b}_{\rm eff}(t) \; ,
\eeq
with the morphon variables (for the simplest choice of a scalar field fluid analogy), 
\beq \label{morphonenergies}
	\epsilon^{\Phi, b}_{\rm eff}: = \frac{1}{2}\left(\ddt{}\tilde\Phi_\CD\right)^2 + U^b_{\rm eff}(\tilde\Phi_\CD) \ \ ; \ \ 
	p^{\Phi, b}_{\rm eff}: = \frac{1}{2}\left(\ddt{}\tilde\Phi_\CD\right)^2 - U^b_{\rm eff}(\tilde\Phi_\CD)\; .
\eeq
The morphon field is therefore defined from the backreaction terms as follows:
\begin{eqnarray}
\label{morphondictionary}
&&24\pi G \left(\ddt{}\tilde\Phi_\CD\right)^2 := - 3\tilde\CQ_\CD^b - 2 \tilde\CP_\CD^b - \tilde\CW_\CD^b  \ ; \\ 
&&24\pi G \ U^b_{\rm eff}(\tilde\Phi_\CD) \; \; := \ \ 3\tilde\CL_\CD^b + \tilde\CP_\CD^b - \tilde\CW_\CD^b   \ . 
\end{eqnarray}
We infer that the \textit{virial condition} on morphon energies is satisfied for the relation:
\beq
\label{virial}
0\,=\,\left(\ddt{}\tilde\Phi_\CD\right)^2 -  U^b_{\rm eff}(\tilde\Phi_\CD) \,=\, - \frac{1}{8 \pi G} \left(
\tilde\CQ_\CD^b + \tilde\CP_\CD^b + \tilde\CL_\CD^b  \right)  \;,
\eeq
\textit{i.e.}, in contrast to the irrotational dust matter model without cosmological constant considered in \cite{morphon} where only the kinematical backreaction has to vanish for this condition, here the dynamical and the Dark Energy backreaction terms also enter the energy density balance. Those terms also determine, together with the curvature deviation term, the effective potential in which the morphon evolves. Vanishing of the backreaction terms would reduce the averaged equations to the standard Friedmann equations that represent a case of the scalar field `virial equilibrium' \eqref{virial}.

We summarize the reformulation of the averaged equations in terms of the morphon field in the following \textit{Corollary}.

\begin{corollary}[Effective Friedmannian form with effective scalar field]\\
\label{cor:intrinsic_effective_morphon}
\vspace{-7pt}

The set of effective cosmological evolution equations of \textit{Theorem} \ref{ths:av_intrinsic} can be written in Friedmannian form for the averaged energy sources and effective scalar field energies:\footnote{%
%
The language of a given effective scalar field theory can be freely specified. We may think of other effective scalar field theories, \textit{e.g.} non-minimally coupled, especially if we set the scalar field analogy within an extrinsic averaging formalism, where another dictionary could be a better choice. In this line, the analogy---here set up for fluid-intrinsic averaging---could have interesting implications for the relation of different effective scalar field theories to different foliation choices. By construction, the scalar field obtained here for any given domain $\CD$ obeys the evolution equations of a homogeneous scalar field, being built from pure functions of $t$. One may, however, define it first (following the above procedure) as a pure function of the time $t$ of a given foliation, and then consider this field in another foliation choice, where it will in general be inhomogeneous. In this way, the scalar field would acquire a nonvanishing spatial gradient and would so allow for a comparison with phenomenological
inhomogeneous scalar fields employed in standard perturbation theory.
}
\begin{eqnarray}
	 3 \, \frac{1}{a^b_\CD} \frac{{\rm d}^2 a^b_\CD}{{\rm d}t^2} & = & - 4 \pi G \,\left(\baverage{\epsilont} + \epsilon^{\Phi, b}_{\rm eff} + 3 \, (\baverage{\tilde p} + p^{\Phi, b}_{\rm eff})\right) + \cc \; ;
 	\label{eq:friedmann_raych_morphon}  \\
	3 \, \big(H_\CD^b \big)^2 & = & 8 \pi G \, \left(
 \baverage{\epsilont} + \epsilon^{\Phi, b}_{\rm eff} \right) - 3 \frac{k_{\initial\CD}}{(a^b_\CD)^2} + \cc \; .
	\label{eq:friedmann_hamilt_morphon}
\end{eqnarray}
The integrability condition \eqref{eq:int_condition_intr}, written in terms of the deviation fields $\tilde\CW_\CD^b$ and $\tilde\CL_\CD$ (see \eqref{deviationfields}),
\begin{eqnarray}
 \ddt{} \tilde{\CQ}_\CD^b + 6 H^b_\CD \tilde{\CQ}_\CD^b + \ddt{} \tilde\CW_\CD 
			+ 2 H^b_\CD \tilde\CW_\CD + 4 H^b_\CD \tilde{\CP}_\CD^b - 2 \ddt{} {\tilde\CL}_\CD^b\nonumber \\
	= 16 \pi G \left( \ddt{} \baverage{\epsilont}
			+ 3 H^b_\CD \baverage{\epsilont + \tilde p} \right)
			 \; , \qquad
\end{eqnarray}
is mapped to a conservation law for the effective time-dependent scalar field energies, equivalent to an effective Klein-Gordon operator applied to $\tilde\Phi_\CD$(t):
\begin{eqnarray}
\ddt{}\epsilon^{\Phi,b}_{\rm eff} + 3 \, H_\CD^b \, \left( \epsilon^{\Phi,b}_{\rm eff} + p^{\Phi,b}_{\rm eff}\right) &+ \ \mathfrak{S}_\CD^b \,=\,0 \; , \\
\mathrm{i.e.,} \quad \frac{{\rm d} \tilde\Phi_\CD}{{\rm d}t}\left(\frac{{\rm d}^2 \tilde\Phi_\CD}{{\rm d}t^2}  + 3 H_\CD^b \frac{{\rm d} \tilde\Phi_\CD}{{\rm d}t} + \frac{\partial U^b_{\rm eff}(\tilde\Phi_\CD)}{\partial \tilde\Phi_\CD}\right) &+\ \mathfrak{S}_\CD^b \,=\,0\; .
\label{kleingordon}
\end{eqnarray}
This is balanced by the averaged conservation law for the sources (\textit{cf.} \eqref{eq:av_en_cons_intr}):
\begin{eqnarray}
\label{conservationbalance}
\mathfrak{S}_\CD^b (t):= \ddt{}\baverage{\epsilont} + 3 \, H_\CD^b \, \left( \baverage{\epsilont} + \baverage{\tilde p}\right) \,=\, \baverage{\Thetat} \!\! \Big\langle \tilde p \Big\rangle_\CD^b \!\!
			- \baverage{\Thetat \, \tilde p} \!\! \qquad\qquad\nonumber \\
	- \baverage{\frac{N^3}{\gamma^3} (\nabla_\mu q^\mu + q^\mu a_\mu + \pi^{\mu \nu} \sigma_{\mu \nu})}
			+ 2 \baverage{\epsilont \, \frac{\gamma}{N} \ddt{} \left( \frac{N}{\gamma} \right) } \; , \quad\qquad \;
\end{eqnarray}
so that in total the conservation law for the total effective energy densities \eqref{eq:friedmann_cons} holds.
\end{corollary}
\end{corgroup}

We note the important property that $\mathfrak{S}_\CD^b (t)$ only identically vanishes in general in the case of dust
matter, then separating the individually satisfied material sources conservation law from a Klein-Gordon equation for $\tilde\Phi_\CD (t)$. In more general cases, $\mathfrak{S}_\CD^b (t)$ is not identically zero, and the dynamics of the averaged material sources, energy density and pressure, is consequently affected by its coupling to the effective morphon field.


\subsection{Lagrangian effective forms}
\label{subsec:lagrangian_form}

The averaged evolution equations derived in the fluid-intrinsic approach can be further simplified by moving to a Lagrangian picture, where the rescaled variables ($\Thetat$, $\epsilont$, ...)  reduce to the original variables ($\Theta$, $\epsilon$, ...) since $N=\gamma$.
We recall that a Lagrangian picture requires both a foliation choice of hypersurfaces at constant fluid proper time $\tau$, and the natural adapted spacetime coordinates choice $(\tau, X^i)$.
We shall list below (section \ref{subsubsec:lagrange_interest}) arguments why we consider this choice as the most adapted one, both to the geometric structure and to cosmological applications. The choice of fluid-comoving spatial coordinates $X^i$ actually remains optional in the following, as we have seen that the average equations do not depend on the shift, all terms that they feature being invariant under a change of spatial coordinates.

Within this picture, the commutation rule \eqref{eq:com_rule_intrinsic} and scale factor evolution rate \eqref{eq:evol_aD_intrinsic} become respectively:
\begin{equation}
\baveragedot{\overset{\phantom{.}}{\psi}} = \baverage{\dot\psi}
 		- \Rbaverage{\Theta} \Rbaverage{\psi} + \Rbaverage{\Theta \, \psi} \quad  ; \quad 
\frac{\dot a^b_\CD}{a^b_\CD} \, = \, \frac{1}{3} \baverage{ \Theta} \; ,
\end{equation}
where the overdot and $\mathrm{d}/\mathrm{d}t$ are here equivalent operators for scalars.


\subsubsection{Lagrangian effective cosmological equations}

We summarize the Lagrangian formulation of the averaged cosmological equations, \textit{i.e.}, their expression in a Lagrangian picture, in the following \textit{Corollary}. 

\begin{corgroup}
\label{cor:intrinsic_Lagrange_all}
\begin{corollary}[Lagrangian effective cosmological equations]\\
\label{cor:intrinsic_Lagrange_averages}
\vspace{-7pt}

The evolution equations for the fluid-volume scale factor $a^b_\CD$ \eqref{eq:av_raych_intr}--\eqref{eq:av_hamilt_intr} for a choice of foliation by constant fluid proper time slices, parametrized by $t=\tau$, read:
\begin{eqnarray}
 3 \, \frac{\ddot a^b_\CD}{a^b_\CD} & = & \, 
		- 4 \pi G \baverage{\, \epsilon + 3 p \,} + \cc + \CQ^b_\CD + \CP^b_\CD \ ; \\
 3 \, \big(H_\CD^b \big)^2  & = & \,
 		8 \pi G \baverage{\,\epsilon\,} + \cc -\frac{1}{2} \QD^b -\frac{1}{2} \baverage{\SR} \ ,
\end{eqnarray}
with $H_\CD^b = \dot a^b_\CD / a^b_\CD$, and the backreaction terms being reduced to the following forms:
\begin{eqnarray}
 \label{eq:kin_back_lagrange}
 \CQ^b_\CD & = & \frac{2}{3} \baverage{\left(\Theta - \baverage{\Theta} \right)^2 }
			- 2 \baverage{\sigma^2} + 2 \baverage{\omega^2} \, ; \\ 
 \label{eq:dyn_back_lagrange}
 \CP^b_\CD & = & \baverage{\CA} \, .
\end{eqnarray}
The corresponding integrability condition (\textit{cf.} equation~\eqref{eq:int_condition_intr}) now becomes:
\begin{eqnarray}
 {\dot\CQ}^b_\CD + 6 H_\CD^b \QD^b + \baveragedot{\SR}
		+ 2 H_\CD^b \baverage{\SR} + 4 H_\CD^b \PD^b \qquad \nonumber\\
		= 16 \pi G \left(\baveragedot{\epsilon}
		+ 3 H_\CD^b \baverage{\epsilon + p} \right) ,
 \label{eq:int_cond_lagrange}
\end{eqnarray}
with the right-hand side satisfying the averaged energy conservation equation \eqref{eq:av_en_cons_intr} under the following simpler form:
\begin{equation}
\label{conservation_lagrange}
\baveragedot{\epsilon}  + \, 3 H_\CD^b \baverage{\epsilon + p} = \baverage{\Theta} \baverage{p} - \baverage{\Theta \, p} 
		\, - \baverage{\nabla_\mu q^\mu + q^\mu a_\mu + \pi^{\mu \nu} \sigma_{\mu \nu}} \, ._{\phantom{\big|}}
\end{equation}
\end{corollary}


\subsubsection{Effective Friedmannian and Lagrangian form}

Combining the above form with the effective Friedmannian form provides the most compact writing of the averaged cosmological equations, summarized in the following second part of the \textit{Corollary}.

\begin{corollary}[Compact form of Lagrangian cosmologies]\\
\label{cor:intrinsic_Lagrange_Friedmann}
\vspace{-7pt}

In a Lagrangian picture (which implies $N=\gamma$), the effective Friedmann equations \eqref{eq:friedmann_hamilt}--\eqref{eq:friedmann_cons} reduce to the following form: 
\begin{align}
3 \frac{\ddot a^b_\CD}{a^b_\CD}  =  - 4 \pi G \left(\epsilon^b_{\rm eff} + 3 \, p^b_{\rm eff}\right) + \cc \ ;
  \label{eq:friedmann_lagrange2} \\
 3 \, \big(H_\CD^b \big)^2  =  8 \pi G \,\epsilon^b_{\rm eff} - \frac{3 \, k_{\initial\CD}}{(a^b_\CD)^2} + \cc \ ;
  \label{eq:friedmann_lagrange1} \\
{\dot\epsilon}^{\,b}_{\rm eff} + 3 \, H_\CD^b \, \left( \epsilon^b_{\rm eff} + p^b_{\rm eff} \right) = 0 \ ,
 \label{eq:friedmann_lagrange3}  
\end{align}
with $H^b_\CD = {\dot a^b}_\CD / a^b_\CD$. The effective energy density $\epsilon^b_{\rm eff} (t)$ and effective pressure $p^b_{\rm eff} (t)$,
as defined in \eqref{eq:def_epsilon_eff} and \eqref{eq:def_p_eff}, are here simplified to the following expressions:
\begin{align}
\epsilon^b_{\rm eff}  (t) & {} = \baverage{\epsilon} - \frac{1}{16 \pi G} \QD^b - \frac{1}{16 \pi G} \WD^b \ ; 
\label{effectivesources_lagrange1}\\
p^b_{\rm eff} (t) & {} = \baverage{p} - \frac{1}{16 \pi G} \QD^b + \frac{1}{48 \pi G} \WD^b - \frac{1}{12 \pi G} \PD^b  \ , \quad \qquad
\label{effectivesources_lagrange2}
\end{align}
with $\QD^b$ and $\PD^b$ as given by \eqref{eq:kin_back_lagrange} and \eqref{eq:dyn_back_lagrange},
and with the curvature deviation term $\tilde\CW_\CD^b$ reduced to $\WD^b = \baverage{\SR} - 6 \, k_{\initial\CD} / (a^b_\CD)^2$. 
\end{corollary}
Note that in a Lagrangian picture the \textit{Dark Energy backreaction} $\tilde\CL_\CD^b$ vanishes as a consequence of $N=\gamma$. 

Such a picture also allows for a rewriting of the scalar (morphon) field analogy of \textit{Corollary}~\ref{cor:intrinsic_effective_morphon} under a simplified form.

\begin{corollary}[Compact form of Lagrangian cosmologies with morphon]\\
\label{cor:intrinsic_Lagrange_morphon}
\vspace{-7pt}

The effective Friedmann equations in a Lagrangian picture \eqref{eq:friedmann_lagrange1}--\eqref{eq:friedmann_lagrange3} can be interpreted as being sourced by time-dependent morphon energy densities through reformulation of the backreaction terms in \eqref{effectivesources_lagrange1} and \eqref{effectivesources_lagrange2}:
\begin{align} \label{effectivesources_morphon_lagrange}
	\epsilon^b_{\rm eff}(t)  =  \baverage{\epsilon}(t) + \epsilon^{\Phi, b}_{\rm eff}(t) \quad & ; \quad 
	\epsilon^{\Phi, b}_{\rm eff} (t) := \frac{1}{2}{\dot\Phi}_\CD^2 + U^b_{\rm eff}(\Phi_\CD) \; ; \\
	p^b_{\rm eff}(t) = \baverage{p}(t) + p^{\Phi, b}_{\rm eff}(t) \quad & ; \quad 
	p^{\Phi, b}_{\rm eff} (t) := \frac{1}{2}{\dot\Phi}_\CD^2 - U^b_{\rm eff}(\Phi_\CD)\; , 
\end{align}
with the simplified morphonic dictionary:
\begin{align}
\label{morphondictionary_lagrange}
24\pi G \ {\dot\Phi}_\CD^2 & {} := -3 \CQ_\CD^b -2 \CP_\CD^b - \CW_\CD^b \ ; \\ 
24\pi G \ U^b_{\rm eff}(\Phi_\CD) & {} := \, \CP_\CD^b - \CW_\CD^b  \ ,\quad\qquad\qquad 
\end{align}
yielding the following \textit{`virial condition'} on morphon energy densities:
\beq
\label{virial_lagrange}
0\,=\,{\dot\Phi}_\CD^2 -  U^b_{\rm eff}(\Phi_\CD) \,=\, - \frac{1}{8 \pi G} \left(
\CQ_\CD^b + \CP_\CD^b \right)  \;.
\eeq
The conservation law \eqref{eq:friedmann_lagrange3} couples the conservation law for the material sources \eqref{conservation_lagrange} to an effective Klein-Gordon operator applied to $\Phi_\CD(t)$:
\beq
{\dot\Phi}_\CD \left({\ddot\Phi}_\CD  + 3 H_\CD^b {\dot\Phi}_\CD + \frac{\partial U^b_{\rm eff}(\Phi_\CD)}{\partial \Phi_\CD}\right) + \mathfrak{S}_\CD^b (t) \,=\,0 \; ,
\eeq
with $\mathfrak{S}_\CD^b$ here reduced to 
\beq
\mathfrak{S}_\CD^b =  \baveragedot{\epsilon}  + \, 3 H_\CD^b \baverage{\epsilon + p} = \baverage{\Theta} \baverage{p} - \baverage{\Theta \, p} - \baverage{\nabla_\mu q^\mu + q^\mu a_\mu + \pi^{\mu\nu} \sigma_{\mu\nu}} \ .\nonumber
\eeq
\end{corollary}
\end{corgroup}

As for the more general situation considered in \textit{Corollary}~\ref{cor:intrinsic_effective_morphon}, the above $\mathfrak{S}_\CD^b (t)$ is still not identically zero in general for a non-dust fluid. As an example it was pointed out in Paper~II, as a consequence of this property, that the spatially averaged inhomogeneous (irrotational) radiation fluid does not follow the volume expansion law of the homogeneous-isotropic radiation-dominated cosmos.

Useful characteristics for cosmological models such as dimensionless effective cosmological `parameters' or effective state finders can also be defined from the above compact forms along the lines explained in \cite[sect. 2.4]{buchert:review}.

\section{Discussion and concluding remarks}
\label{sec:discussion_intr}

In this article we have introduced the \textit{fluid-extrinsic} (section~\ref{sec:extrinsic}) and \textit{fluid-intrinsic} (section~\ref{sec:intrinsic}) averaging procedures. We have applied them to different scalar parts of 
the Einstein equations in the setting of an arbitrary spacelike $3+1$ splitting of spacetime, and for a general fluid with a congruence tilted with respect to the normal to the hypersurfaces. 
All past generalization proposals of Papers~I and II to arbitrary spatial foliations found in the literature have focussed on fluid-extrinsic formalisms. We have compared in detail these investigations to our fluid-extrinsic averaging approach, which, unlike the previous proposals, is formulated for comoving, i.e. rest mass--preserving, domains of averaging.
The fluid-intrinsic approach, on the other hand, forms our new proposal of constructing effective cosmological equations. We have already briefly introduced this formalism including a discussion on foliation dependence of cosmological backreaction in \cite{foliationsletter}, and have put a manifestly $4-$covariant reformulation of the averaged equations into perspective in \cite{asta1}. In the present work we provide all the details that will be necessary for the concrete application of the fluid-intrinsic framework to various fluid equations of state and cosmological models, as well as to the analysis of relativistic numerical simulations.


\subsection{Recovering the results of Paper I and Paper II}
\label{subsec:recovery}

The fluid-intrinsic scalar averaging generalizes, in letter and spirit, the original effective cosmological equations for flow-orthogonal foliations of spacetime, as given in Paper I \cite{buchert:av_dust} for irrotational dust (case I below) and Paper II \cite{buchert:av_pf} for irrotational perfect fluids (case II-C below). We show below how these results can be recovered as special cases of the present general approach. We have also highlighted how the fluid-intrinsic averaging is adapted to a Lagrangian picture. We comment below, in addition, on its application within this picture to irrotational perfect fluids with pressure as a new volume-averaged description of such fluids (case II-L below).
\begin{itemize}
\item[I:] \textit{Irrotational, non-tilted dust in Lagrangian form.} We set $\omega = 0$, $p=0$, $q^{\mu}=0$, $\pi_{\mu\nu}=0$, and $\epsilon = \varrho$. We assume the global existence of a fluid-orthogonal foliation, the local existence of which is guaranteed by the irrotationality assumption and Frobenius' theorem, and we will work within this foliation choice (that is, we set $\bm n = \bm u$). The acceleration expression \eqref{eq:mom_cl} shows that the dust fluid is geodesic, $a_\mu = 0$. This implies that $\CA = 0$ and, from the application of \eqref{eq:eulerian_acc} to this fluid-orthogonal and geodesic case, that the lapse is a pure function of time. The time coordinate $t$ can thus be chosen such that the lapse reduces to $N=1$. This $t$ is then a proper time $\tau$ for the fluid (see \eqref{eq:def_tau}, with here $\gamma=1$), and we can apply the Lagrangian form of the averaged equations, \textit{Corollary}~\ref{cor:intrinsic_Lagrange_averages}.
Moreover, $a^b_\CD = a_\CD$ and extrinsic and intrinsic averaging operators become equivalent, $\baverage{\cdots} = \average{\cdots}$. The index $b$ becomes redundant for all expressions, and we directly recover the cosmological equations of Paper~I.
\item[II-C:] \textit{Irrotational, non-tilted perfect fluids in comoving form.} In Paper~II the fluid-orthogonal choice ($\gamma = 1$ with a non-constant lapse function $N$) was adopted. This does not correspond to a Lagrangian picture, \textit{i.e.}, $\tau$ does not reduce to the coordinate time $t$ (see \eqref{eq:def_tau}). To recover the same form we have to use the equations of 
\textit{Theorem}~\ref{ths:av_intrinsic}, in which we can select a fluid-orthogonal foliation assuming an irrotational perfect fluid content. We can then omit the index $b$ for the same reasons as in case I.
Setting  $\gamma = 1$, $\omega = 0$, $q^{\mu}=0$ and $\pi_{\mu\nu}=0$ in these equations, we recover the equations of Paper~II (with the additional contribution of $\cc$). As mentioned in the same context for our extrinsic approach in section~\ref{sec:extrinsic}, these equations are also independent of the shift, and thus they are obtained without the need to assume a vanishing shift as in Paper~II.
\item[II-L:] \textit{Irrotational perfect fluids in Lagrangian form.} We can set $\omega = 0$, $q^{\mu}=0$, $\pi_{\mu\nu}=0$ and consider a fluid proper time foliation, with $t=\tau$ (hence the equations of \textit{Corollary}~\ref{cor:intrinsic_Lagrange_averages} can be used). For nonvanishing pressure gradients in the fluid local rest frames, this foliation is not fluid-orthogonal, $\gamma > 1$, and we get different, simpler averaged equations with respect to Paper~II, with an intrinsic averaging operator that is distinct from the extrinsic one.
\end{itemize}


\subsection{Recovering the Newtonian form of the effective cosmological equations}
\label{subsec:newton}

The compact form of the cosmological equations of \textit{Corollary}~\ref{cor:intrinsic_Lagrange_all}
is directly reminiscent of, and can be transformed into, the corresponding equations that arise in Newtonian Cosmology
\cite{buchert:av_Newton}.
Such a transformation does not require a `Newtonian limit' (referring to a limit where the inverse of a causality constant goes to zero, $1/c\rightarrow 0$; for a detailed discussion of this notion, see \cite{newtonlimit}). It rather requires a restriction of the fluid deformations to integrable fluid deformations, named \textit{Minkowski Restriction} (henceforth MR), defined and executed for various variables in a fluid-orthogonal framework in the series of papers \cite{rza1,rza2,rza3,rza4,rza5}, see especially \cite{rza3,rza4}. 

The cosmological equations presented in this paper do not assume any particular spatial metric or any specific explicit form for the twice-covariant fluid-orthogonal projecting tensor $\mathbf{b} = b_{\mu \nu} \, \mathbf{d}x^\mu \otimes \mathbf{d}x^\nu$: they only depend functionally on these tensors. 
As a step toward extending the definition of the MR to the current general framework, we first consider a Lagrangian picture. We thus in particular set a $3+1$ foliation by constant fluid proper time hypersurfaces. For this foliation, we then write in all generality the restriction $\mathbf{b}_{|\mathrm{\Sigma}}$ of $\mathbf{b}$ to the constant-$\tau$ hypersurfaces\footnote{%
More precisely, the notation $\mathbf{b}_{|\mathrm{\Sigma}}$ denotes the restriction of the type-(0,2) tensor $\mathbf{b}$ to $T\mathrm{\Sigma} \otimes T\mathrm{\Sigma}$ where $T\mathrm{\Sigma}$ is the tangent space to the hypersurface; \textit{i.e.}, it corresponds to the type-(0,2) tensor on the submanifold $\mathrm{\Sigma}$ obtained from the evaluation of $\mathbf{b}$ on vectors of $T\mathrm{\Sigma}$ only. We will use the same notation in Appendix~\ref{app:volumes} for $3-$forms, meaning their restriction to $T\mathrm{\Sigma} \otimes T\mathrm{\Sigma} \otimes T\mathrm{\Sigma}$ in this case.
} $\mathrm{\Sigma}$ in terms of three Cartan coframes $\boldsymbol{\eta}^a$ (basis of $1-$form fields on the hypersurfaces): 
\beq
\mathbf{b}_{|\mathrm{\Sigma}} = \delta_{ab} \, \boldsymbol{\eta}^a \otimes \boldsymbol{\eta}^b \ , 
\eeq
where the labels $a,b=1,2,3$ count the spatial coframes. 

The MR is an assumption, within the fluid proper time foliation choice adopted here, of \textit{integrability} of the spatial coframes within the slices. It thus restricts these general spatial $1-$forms to exact forms on the spatial slices: $\boldsymbol{\eta}^a \mapsto \spatiald f^a$, where $\mapsto$ denotes the application of the MR, and $\spatiald$ is the exterior derivative on the hypersurfaces. The three scalar functions $f^a$, considered as functions of space parametrized by $\tau$, define Eulerian spatial coordinates $x^a = f^a (X^i; \, \tau)$, where $X^i$ are the comoving (Lagrangian) spatial coordinates. In particular, this condition of integrability requires the global existence of coordinates on the spatial slices, and consequently imposes that the hypersurfaces are diffeomorphic to the Euclidean space $\mathbb{R}^3$. The set of functions $f^a$ moreover defines for each slice the transition function---a diffeomorphism---between the Lagrangian and Eulerian maps.

In the exact spatial basis $\{ \spatiald X^i\}$, associated with the Lagrangian coordinates $X^i$, the coefficients of the Cartan coframes reduce to the Newtonian deformation matrix ($\partial f^a / \partial X^i$) in Lagrangian coordinates, 
\begin{equation}
\boldsymbol{\eta}^a = \eta^a_{\ i} \, \spatiald X^i \ \  \mapsto \;\; \spatiald f^a = \frac{\partial f^a}{\partial X^i }\, \spatiald X^i \ .
\end{equation}
The restriction to the hypersurfaces $\mathbf{b}_{|\mathrm{\Sigma}}$ of the fluid-orthogonal projector reduces to an Euclidean metric within the slices, as can be seen through the coordinate transformation from Lagrangian to Eulerian coordinates, 
\begin{equation}
\mathbf{b}_{|\mathrm{\Sigma}} = b_{ij}\, \spatiald X^i \otimes \spatiald X^j \mapsto \delta_{ab} \, \frac{\partial f^a}{\partial X^i} \frac{\partial f^b}{\partial X^j} \, \spatiald X^i \otimes \spatiald X^j = \delta_{ab}\,
\spatiald x^a \otimes \spatiald x^b \; .
\end{equation}
For a dust fluid, the $3+1$ Einstein equations then reduce to the Newtonian equations in Lagrangian form (for an explicit demonstration in the case of an irrotational dust matter model, see \cite{rza1} for the Einstein equations and \cite{buchert:lagrange,ehlersbuchert} for the Newtonian equations in Lagrangian form, admitting non-zero vorticity; see also the related discussions in \cite[sect.7]{buchert:focus} and \cite{newtonlimit}). 
In presence of acceleration due to a non-dust energy-momentum tensor, the induced special-relativistic effects---that can be interpreted as local time dilations, in addition to the direct contribution of pressure to the energy sources---will maintain differences with respect to the Newtonian dynamics. A nonrelativistic limit (in the sense of special relativity) $c \rightarrow \infty$ would then be further required to fully recover the Newtonian equations.

We illustrate this transformation to the Newtonian equations for the case of a \textit{rotational dust matter model}. We thus consider a (geodesic) perfect fluid with zero pressure, but with \textit{a priori} $\bm \omega \neq {\bf 0}$, and we adopt a Lagrangian description. The effective cosmological equations of \textit{Corollary}~\ref{cor:intrinsic_Lagrange_averages} thus hold and reduce to the following for this case:
\begin{align}
\label{eq:avgRaych_prenewton}
 		3 \, \frac{\ddot a^b_\CD}{a^b_\CD} & {} =  \, 
		- 4 \pi G \baverage{\varrho} + \cc + \CQ^b_\CD \; ; \\
\label{eq:avgHamilton_prenewton}
  3 \, \big(H_\CD^b \big)^2  & {} =  \,
 		8 \pi G \baverage{\varrho} - 3 \frac{k_{\initial\CD}}{\left( a_\CD^b \right)^2} + \cc -\frac{1}{2} \QD^b -\frac{1}{2} \WD^b \; ,
\end{align}
with $H_\CD^b = \dot a_\CD^b / a_\CD^b$, and $\WD^b = \baverage{\SR} - 6 \, k_{\initial\CD} / \left(a_\CD^b \right)^2$. For dust, the energy density $\epsilon$ coincides with the rest mass density $\varrho$, and the averaged energy conservation law \eqref{conservation_lagrange} reduces to the continuity equation for the average density:
\beq
\label{eq:avgMasscons_prenewton}
 \baveragedot{\varrho}  + \, 3 H_\CD^b \baverage{\varrho} \,=\,0 \; .
\eeq
Applying the MR to this situation, the fluid proper volume measure on the hypersurfaces reduces to an Euclidean volume element, with $\int_{\CD_{\bm X}} \psi \, \sqrt{\det(b_{ij})} \, \mathrm{d}^3 X = \int_{\CD_{\bm x}} \psi \, \mathrm{d}^3 x$ for any scalar $\psi$. Hence, the fluid volume $\VDb$ and effective scale factor $a_\CD^b$ reduce to Euclidean-space expressions in the Eulerian coordinates $x^a$, and the intrinsic averaging operator reduces to an Euclidean volume average $\Eaverage{\cdot}$. The intrinsic kinematical backreaction reduces to
\begin{equation}
 \label{eq:kin_back_prenewton}
 \QD^b = \CQ_\CD^E :=  \frac{2}{3} \Eaverage{\left(\Theta - \Eaverage{\Theta} \right)^2 }
			- 2 \Eaverage{\sigma^2} + 2 \Eaverage{\omega^2} \ , 
\end{equation}
where $\Theta$, $\sigma$ and $\omega$ are the scalar kinematic variables of the now integrable expansion and vorticity tensors, \textit{i.e.}, they are associated with the kinematic invariants of the velocity gradient field $\spatiald  \dot f^a = (\partial \dot f^a / \partial x^b) \; \spatiald x^b$ with coefficients expressed in the Eulerian spatial coordinate basis $\{\spatiald x^a\}$.

In Newtonian theory, $k_{\initial\CD}$ is solely a constant of integration (not associated with a constant-curvature term), and the equivalent of the `curvature deviation' term $\CW^b_\CD$ also loses its relation to a Riemannian curvature.
The curvature deviation term appearing in equation \eqref{eq:avgHamilton_prenewton} above can as well be defined without an explicit averaged curvature term. The integrability condition \eqref{eq:int_cond_lagrange} reduces for the dust case to
\beq
 \dot{\CQ}^b_\CD + 6 H_\CD^b \, \CQ^b_\CD + \dot{\CW}^b_\CD
		+ 2 H_\CD^b \, \CW^b_\CD \,=\,0 \ ,
\eeq
which can also be directly deduced from \eqref{eq:avgRaych_prenewton}--\eqref{eq:avgMasscons_prenewton}. Within the MR, noting respectively $\WD^b$ and $a_\CD^b$ as $\WD^E$ and $a_\CD^E$ in this case, this can be integrated to yield (\textit{cf.} \cite[sect.2.3.1]{buchert:review}, with arbitrary initial data being included hereunder for both backreaction variables):
\beq
\label{eq:integrabilityintegral}
\CW^E_\CD (t)  \,= - \CQ^E_\CD (t) 
+ \frac{1}{a_\CD^E \big( t \big)^2} \left( \CW^E_\CD (t_{\mathbf{i}})+\CQ^E_\CD (t_{\mathbf{i}}) - 4 \int_{t_\mathbf{i}}^t   \! \CQ^E_\CD(t') \, a_\CD^E(t') \, \dot{a}_\CD^E(t') \, {\mathrm d}t' \right) \; ,
\eeq
since $a^E_\CD (\initial t) = {}  a^b_\CD (\initial t) = {} 1$. This gives an expression
for $\WD^E$ in terms of $\QD^E$ and $a_\CD^E$, and their time-history,\footnote{%
The integration constant $\CW_\CD^E(t_{\mathbf{i}}) + \CQ_\CD^E(t_{\mathbf{i}})$ is fully determined
by the---arbitrary---choice of the constant $k_{\initial\CD}$ and by the initial conditions on $H_\CD^b$ and $\langle \varrho \rangle^b_\CD$ \textit{via} the averaged energy constraint equation \eqref{eq:avgHamilton_prenewton} at $t = \initial t$. With the choice $k_{\initial{\CD}} = 0$, the initial data of a reference Einstein-de Sitter FLRW model would, for instance, correspond to $\CW_\CD^E(t_{\mathbf{i}}) + \CQ_\CD^E(t_{\mathbf{i}})=0$.
\label{fn:int_constants}
} which are all defined from Euclidean-space volume integration. The above formula can thus equally be used to define the term corresponding to $\WD^E$ in a Newtonian setting; with this substitution, the system of averaged equations \eqref{eq:avgRaych_prenewton}--\eqref{eq:avgMasscons_prenewton} within the MR reduces to the Newtonian equations\footnote{%
In Newtonian cosmology we have to abandon the background-free character of general relativity. In order to obtain unique solutions, we have to introduce a background in terms of a linear reference velocity field, $V_i = H_{ij} \, x_j$, with homogeneous expansion, shear and vorticity, $H_{ij} := (\Theta_H (t)/3) \, \delta_{ij} + \Sigma_{ij}(t) + \Omega_{ij}(t)$, with $\Sigma_{ij}$ symmetric and traceless and $\Omega_{ij}$ antisymmetric. Deviations thereof are to be bound to a $3-$torus topology. As a result of the integrability of the Newtonian variables on flat space sections, the kinematical backreaction (which does not depend on the background variables) can be written in terms of full divergences of deviation vector fields. Hence, this backreaction has to vanish on the boundary-free $3-$torus, see \cite{buchert:av_Newton} and the recent discussion \cite{buchert:av_Newton_Rev}.%
}%
(\textit{cf}. \cite{buchert:av_Newton}). The spatial slices chosen here, where all fluid elements have experienced the same proper time $\tau$ elapsed from a reference initial hypersurface, play the role of Newton's absolute space, with associated Eulerian coordinates $x^a$, and $\tau$ plays the role of Newton's absolute time.

Note that these results for dust, using the intrinsic average operator of section~\ref{sec:intrinsic}, do not require the additional Newtonian limit assumption $c \rightarrow \infty$. With this additional assumption, the special-relativistic effects of differences of clock rates would become negligible and one would be able to neglect the local tilt (although not its spatial variations) between the fluid rest frames and the constant-$\tau$ hypersurfaces considered, $\gamma \simeq 1$. The extrinsic averages of section~\ref{sec:extrinsic} would then coincide with the intrinsic averages, with $\sqrt{h} \simeq \sqrt{b}$, and thus they would also become Euclidean volume averages.


\subsection{Summarizing remarks on the fluid-instrinsic and Lagrangian approaches}


\subsubsection{Interest of the fluid-intrinsic averaged equations}

\textit{Theorem}~\ref{ths:av_intrinsic} shows that, in the fluid-intrinsic approach, the effective averaged system is not affected by contributions from the tilt, apart from the local time rate factor $N/\gamma$, and the \textit{stress-energy backreaction} vanishes.
This is different from the previous extension proposals of the averaging approach of Papers~I and II to general foliations found in the literature and in our section~\ref{sec:extrinsic} above (extrinsic approach), which we have put into perspective in subsection \ref{subsec:discussion}. Tilt effects may, however, be important for the observational interpretation, since the observer may be tilted with respect to the cosmic fluid. For effective cosmologies we advocate the fluid-intrinsic point of view, focusing on the effective evolution of the \textit{model universe} and eliminating wherever possible observer-specific issues.
It is then an entirely different question, well-separated from the model universe, how the variables of these cosmologies are related to observables. This would require the study of light cone averages, not considered in this work (however, see \textit{e.g.} \cite{lightconeav1,lightconeav2} for proposals of light cone averaging formalisms).

The fluid-intrinsic approach focuses on the proper volume of the fluid as a local volume measure. As seen above, this allows for a reduced dependence of the volume, averages and backreaction terms on the choice of spatial foliation, and yields simple and compact expressions for the effective scale factor evolution equations in terms of the foliation-independent intrinsic kinematic and dynamical variables of the fluid. It also allows for a direct, transparent relation between the conserved total fluid rest mass within the domain and the averaged rest mass density. It can be viewed as an alternative extension of the formalism of Papers~I and II, reducing to these results for an irrotational fluid in the fluid-orthogonal foliation, as does the extrinsic approach, but preserving in all foliations the relation between the volume measure and the local fluid rest frames.
We have also shown that this approach is especially suited for a Lagrangian picture; we shall now comment further on this combined description of averaged dynamics. 


\subsubsection{Interest and limitations of the Lagrangian picture choice}
\label{subsubsec:lagrange_interest}

\textit{Corollary}~\ref{cor:intrinsic_Lagrange_averages} shows, in addition to \textit{Theorem}~\ref{ths:av_intrinsic}, that choosing a constant fluid proper time foliation parametrized by $t=\tau$ (implying $N=\gamma$), makes the \textit{Dark Energy backreaction} $\tilde \CL_\CD^b$ vanish and removes the need to account for a difference of time rates in the dynamical backreaction and in the rescaling of all variables to be averaged. 
Despite these simplifications, we emphasize that the corresponding set of 
effective cosmological equations still holds for a general fluid.
The difference between \textit{Theorem}~\ref{ths:av_intrinsic} (presenting the effective cosmological equations for general fluids \textit{and} general foliations) and \textit{Corollary}~\ref{cor:intrinsic_Lagrange_averages} (applying them specifically to a fluid proper time foliation) may serve for a discussion of the robustness of the averaged equations with respect to the foliation choice (see \cite{foliationsletter}).

We here summarize the important features resulting from the use of the Lagrangian picture in combination with the fluid-intrinsic averaging framework:
\begin{itemize}
 \item[(i)] It links the foliation itself to the fluid (in addition to  already making such a link for the domain propagation, volume measure and local variables to be averaged, due to the fluid-intrinsic formalism),
	in a way alternative to the fluid-orthogonal choice, but in contrast to the latter it comes with a unique 
	time-normalization and holds for any fluid;
 \item[(ii)] It allows for an especially simple and compact form of the averaged equations, 
	removing the need for rescalings and extra terms due to different clock rates;
 \item[(iii)] The corresponding choice of time is formally unique up to a constant along each flow line (see, however, our remarks on effective times below), and it enjoys a clear physical interpretation. The associated time derivatives, in particular the scale factor expansion and acceleration rates, are well-defined as proper-time rates for the fluid elements;
 \item[(iv)] It directly reduces to the usual (fluid-orthogonal and Lagrangian) approach for irrotational inhomogeneous dust
	and for homogeneous perfect fluids (FLRW);
 \item[(v)] It allows for a recovery of the Newtonian averaged equations under the assumption of the reduction of the fluid-orthogonal projector to an Euclidean metric within constant fluid proper time spatial slices (Minkowski Restriction, as defined in section~\ref{subsec:newton} above), for negligible $4-$acceleration and pressure;
 \item[(vi)] It also allows for simple, transparent, Newtonian-like formulas for the local kinematic variables and acceleration (especially the components of the tensor variables).
\end{itemize}

It is important to note that in the case of several fluid components, the advantages of a Lagrangian description can in general only be preserved for one fluid, from which the proper time will be defined, unless a Lagrangian description is considered applicable to the common motion of a fluid mixture (see also the related discussion on the multi-fluid case at the end of section~\ref{subsubsec:literature_transport}). A further, already mentioned, important potential limitation of the proper time foliation choice has to be checked in individual applications: the hypersurfaces built in such a foliation from a given, spacelike, initial hypersurface may, over time, become strongly tilted with respect to the fluid, and might even not all remain everywhere spacelike. Along a given flow line, strong acceleration and/or vorticity may induce an infinite tilt after a finite time. If this occurs, the averaging formalism becomes ill-defined in the spacetime regions where the hypersurfaces are null or timelike and lose their interpretation as `spatial' slices.

\subsubsection{Is there interest to go beyond this work? --- an outlook}
\label{subsubsec:beyond}

The results of this work are \textit{general} in various respects, culminating in compact forms of effective cosmologies, especially within a Lagrangian description. However, a few issues remain and are worth addressing. We highlight some of them in what follows and suggest some procedures towards solutions or alternative constructions.

\paragraph{The issue of closure.}

The presented sets of averaged equations and compact cosmological equations are not closed. 
This known issue is already obvious from the approach of performing averages of only the scalar parts of the Einstein equations.
It is also expected given that a balance equation on averages will not allow to reconstruct the local inhomogeneous metric
(similarly to, in Newtonian contexts, the virial relations not allowing for the reconstruction of the orbits in phase space).
We do not enter the issue of averaging or smoothing tensor variables here, but we emphasize that even averaging further scalar equations would result in a hierarchy of equations that would not close
(similar to the hierarchy of moment equations in kinetic theory).
As in the standard Friedmannian framework, where closure conditions have to be imposed in terms of \textit{equations of state} determining fluid properties,
closure conditions may here be represented as \textit{effective equations of state} in the effective Friedmannian and Lagrangian forms.
Such effective relations encode inhomogeneous properties and evolution details of the fluid and, hence, they are dynamical and not simply derivable from thermodynamical properties.
Closure conditions can be studied in terms of exact scaling solutions \cite{buchert:review,morphon,roy:instability}, global assumptions on model universes \cite{buchert:darkenergy,buchert:static}, exact solutions of the Einstein equations \cite{krazinski,bolejko,sussman1,sussman2,korzynski_embedding,sussman:szekeres,zimdahl:LTB,bolejko:szekeres,sussman3,cliftonsussman} (see also \cite[sect.7]{buchert:focus}), or generic but approximate models for inhomogeneities. The latter may be based on relativistic Lagrangian perturbation theory, \textit{e.g.} \cite{rza2}; on specifications of potentials in the morphon analogy of backreaction acting as quintessence, inflation, or dark matter \cite{morphon,inflation,quentin:darkmatter}; or on other assumptions of equations of state or models for bulk viscosity \cite{chaplygin,zimdahl:viscosity}. Closure may also be approached by topological considerations, as discussed in \cite{buchert:av_dust} and implemented for $2+1$ spacetimes \cite{magni,GBC}, and 
for $3+1$ spacetimes leading to evolution equations for backreaction \cite{GBC}. With this latter we enter a deeper level of the hierarchy of scalar-averaged equations. The above closure methods can also be applied on this deeper level  and they may lead to more refined descriptions of the average dynamics.

\vspace{-5pt}
\paragraph{Remarks on robustness of effective cosmologies.}

We may ask two questions regarding the consequences of the foliation choice on the averaged dynamical equations and on the backreaction terms. One question is, how much a change of foliation quantitatively affects the result on the magnitudes of the backreaction terms. This question is considered in an accompanying Letter \cite{foliationsletter} and will be addressed more explicitly in a forthcoming work \cite{asta2}. 
A second question is, whether a foliation choice can make individual backreaction terms disappear completely. 
We have demonstrated in this paper that the explicit foliation dependence in the intrinsic-averaged equations (and associated backreaction terms) only lies in contributions from the local ratio of time rates $N/\gamma$ between the coordinate time and the proper time of the fluid (the threading lapse, \textit{cf.} Appendix~\ref{app:threading}). We have also seen that, in a Lagrangian picture, the \textit{Dark Energy backreaction} disappears, as well as the explicit dependence of the final equations on the residual freedom in the foliation choice. But, we do not expect a generic situation to exist where the remaining \textit{kinematical} and \textit{dynamical backreactions} would simultaneously disappear.

\vspace{-5pt}
\paragraph{Effective times and the coarse-graining of geometry.}

Looking at the presented effective equations it is evident that spatial properties of the fluid are averaged, while the involved times are not. From a spacetime perspective it is to be expected that a domain may feature its own `averaged time' as an effective time. Note that the choice of a Lagrangian picture allows one to synchronize fluid elements within a domain, but this does not provide an answer to the question of whether an average time would make more physical sense, \textit{e.g.} in terms of an averaged time rates ratio $N/\gamma$. Such an effective time would be expected to be different for domains located in over-- or underdense regions due to the different time dilations; these time differences could play a major role in inhomogeneity-induced dynamics and biases in the interpretation of observations, as suggested in \cite{wiltshire:clocks,wiltshire:dust,wiltshire:lecture}. 

This is connected to the issue of the coarse-graining of geometry: averaging provides integral properties of variables in the inhomogeneous geometry, but the geometry itself could be coarse-grained too to provide an effective geometry. A possible implementation has been discussed in terms of Ricci-flow smoothing of cosmological hypersurfaces and the resulting `dressing' of scalar variables as they are interpreted on the smoothed-out geometry \cite{buchertcarfora:dressing}. 

Finally, a related important question is on which scale the Einstein equations hold \cite{wiltshire:dust,coleywiltshire}.
The assumption behind the averaging operation of the present work is that the Einstein equations hold on the scale where the fluid approximation is a sensible model for the stress-energy tensor, linked to the geometry through the Einstein equations on that scale. Consequently, the averaged dynamics does not obey the Einstein equations, and this is expected too for a coarse-graining of the geometry. An open question is whether spatial averaging of the locally evolved Einstein flow would commute with the dynamics of the spatially averaged variables that include nonlocal terms, \textit{i.e.}, all orders of the correlations of locally evolved variables. A prominent example is the nonlocal kinematical backreaction term.
Note that the assumption of commutativity is made whenever it is assumed (in simulations or exact solutions) that taking the spatial average of a $3+1$ solution of the Einstein equations on the fluid elements' scale provides the correct answer for the average dynamics on a given regional domain.

\paragraph{Statistical hypersurfaces of averaging.}

The framework of Papers~I and II allows for averaging on fluid-comoving domains and on hypersurfaces formed by the fluid itself, but only in cases where the fluid is irrotational and non-tilted.
We proposed here a way of dealing with rotational and tilted cases by introducing a fluid-intrinsic averaging procedure that reduces to the standard volume average in the case of irrotational fluids in their fluid-orthogonal foliations.
We also suggested the fluid proper time foliations as a possible way of tying the average properties even further to the fluid.
Alternatively, we can take a statistical point of view by investigating hypersurfaces of `statistically averaged' geometries, a notion that has yet to be formalized.  
The example of vorticity may illustrate the physical idea behind such a concept:
if we view vorticity as arising on small scales only, while expecting that by going to larger scales vorticity becomes unimportant, we may wonder whether vorticity `averages out' (in a statistical sense) on some scale of averaging.
On scales larger than this, a potential flow is expected, and the fluid can be described as hypersurface-forming, while `averaged-out' scales may still imply a statistical `dressed' 
\cite{buchertcarfora:dressing} contribution from vorticity. 
The idea of viewing averaged equations as providing a definition of `statistical hypersurfaces of averaging' has been advocated (\textit{e.g.} \cite{rasanen:lightpropagation}) and, in Newtonian theory, assumptions such as homogeneous-isotropic turbulence have been advanced to construct statistical averages \cite{olsonsachs}. Statistical averaging is a natural further step to construct effective cosmologies, also since a single realization follows all details of the inhomogeneities, while a statistical average over realizations inherently leads to smoothing out those details. This could potentially provide a solution to both the averaging and the fitting problems on cosmological scales \cite{ellis:proceedings,ellis:fitting,clarkson:review}.

\begin{acknowledgements}
This work is part of a project that has received funding from the European Research Council (ERC) under the European Union's Horizon 2020 research and innovation programme (grant agreement ERC advanced grant 740021--ARTHUS, PI: T. Buchert). PM was supported by
a `sp\'ecifique Normalien' Ph.D. grant for most of this work, and is grateful to the Erwin Schr\"odinger Institute's 2017 Summer school `Between Geometry and Relativity' in Vienna for providing insights that were of use for this work.
XR acknowledges support from the Claude Leon Foundation when this work started in 2012 at the Cosmology and Gravity group of the University of Cape Town. XR thanks Chris Clarkson and Obinna Umeh for fruitful discussions and also acknowledges hospitality at the \textit{Centre de Recherche Astrophysique de Lyon}, and thanks TB for his invitations. We wish to thank Fosca Al Roumi for valuable comments on an earlier manuscript, L\'eo Brunswic, Mauro Carfora, Asta Heinesen, \'Etienne Jaupart, Jan Ostrowski, Roberto Sussman, Nezihe Uzun, Quentin Vigneron, and David Wiltshire for many valuable discussions, and Julien Larena, Syksy R\"as\"anen for useful comments on the final manuscript.
\end{acknowledgements}
\appendix
\section{$3+1$ evolution equations along the congruence of the fluid}
\label{app:3+1_lag}

The evolution equations for $h_{ij}$ and $\CK^i_{\phantom{i} j}$ along the congruence of the fluid 
are obtained from expressions \eqref{eq:evol_h} and \eqref{eq:evol_K}, by relating the derivative 
$\partial_t |_{x^i}$ to ${\mathrm d}/{\mathrm d}t$ with the help of \eqref{eq:der_rel_1}. They read: 
\begin{align}
	\frac{\mathrm d}{{\mathrm d}t} h_{ij} = 
		& - 2 N \CK_{ij} + N_{i || j} + N_{j || i} + V^k \partial_k h_{ij} \ ; 
		\label{eq:evol_h_lag} \\
	\frac{\mathrm d}{{\mathrm d}t} \CK^i_{\phantom{i} j} = 
		& \; N \, \Big( \CR^i_{\phantom{i} j} + \CK \CK^i_{\phantom{i} j} 
			+ 4 \pi G \, \big[ \left( S - E \right) \delta^i_{\phantom{i} j} - 2\, S^i_{\phantom{i} j} \big]
			- \cc \,\delta^i_{\phantom{i} j} \Big) \nonumber \\ 
		& - N^{|| i}_{\ \; || i} + N^k \CK^i_{\phantom{i} j || k} + \CK^i_{\phantom{i} k} N^k_{\ || j}
			- \CK^k_{\phantom{k} j} N^i_{\ || k} + V^k \partial_k \CK^i_{\phantom{i} j} \ . 
		\label{eq:evol_K_lag}
\end{align}

\noindent
\textit{Comoving coordinate system.}
In the comoving picture, as described in section \ref{subsec:lagrange}, we have $\bm N = N \bm v$, or equivalently $\bm V = \bm 0$. 
Equations \eqref{eq:evol_h_lag} and \eqref{eq:evol_K_lag} hence read: 
\begin{align}
	\frac{\mathrm d}{{\mathrm d}t} h_{ij} = 
		& - 2 N \CK_{ij} + N ( v_{i || j} + v_{j || i})  +  v_{i} N_{ || j} + v_{j} N_{ || i}  \ ; 
		\label{eq:evol_h_lag_I} \\
	\frac{\mathrm d}{{\mathrm d}t} \CK^i_{\phantom{i} j} = 
		& \; N \, \Big( \CR^i_{\phantom{i} j} + \CK \CK^i_{\phantom{i} j} 
			+ 4 \pi G \, \big[ \left( S - E \right) \delta^i_{\phantom{i} j} - 2\, S^i_{\phantom{i} j} \big]
			- \cc \,\delta^i_{\phantom{i} j} \Big) - N^{|| i}_{\ \; || i}  \nonumber \\ 
		& + N v^k \CK^i_{\phantom{i} j || k} + N \CK^i_{\phantom{i} k} v^k_{\ || j} + \CK^i_{\phantom{i} k} v^k N_{|| j} 
			- N \CK^k_{\phantom{k} j} v^i_{\ || k} - \CK^k_{\phantom{k} j} v^i N_{|| k} \ . 
		\label{eq:evol_K_lag_I}
\end{align}
For these equations, focusing on the local foliation-related variables $\mathbf{h}$, $\CK^i_{\ j}$ rather than fluid variables, the additional assumption of a fluid proper time foliation and of the choice $t=\tau$ (implying $N=\gamma$), corresponding to the Lagrangian picture, does not affect the comoving picture equations \eqref{eq:evol_h_lag_I}--\eqref{eq:evol_K_lag_I} above, except for the possibility of replacing each occurrence of $N$ by $\gamma$. We shall accordingly not rewrite the equations for this case.

The local equations of evolution along the fluid flow for arbitrary slices and coordinates \eqref{eq:evol_h_lag}--\eqref{eq:evol_K_lag} allow for an alternative derivation of the coordinate-time derivative of the extrinsic volume, Eq.~\eqref{eq:vol_D_final}:
\beq
	\frac{\mathrm d}{{\mathrm d}t} \CV_\CD
		= \int_{\CD_{\bm x}} \left( -N \CK + \left(N v^i \right)_{|| i} \right) \sqrt{h} \, {\mathrm d}^3 x \ , 
	\label{eq:vol_D_appendix}
\eeq
by starting over from \eqref{eq:vol_D_var_III} and expanding its integrand as 
\beq
	\frac{\mathrm d}{{\mathrm d}t} \CV_\CD 
		= \int_{\CD_{\bm x}} \left( \frac{1}{2} h^{ij} \frac{\mathrm d}{{\mathrm d}t} h_{ij}
			+ J^{-1} \frac{\mathrm d}{{\mathrm d}t} J \right) \sqrt{h} \, {\mathrm d}^3 x \ . 
	\label{eq:vol_D_appendix_II}
\eeq
The trace of \eqref{eq:evol_h_lag} can then be used together with \eqref{eq:der_jacob} to obtain: 
\beq
	\frac{\mathrm d}{{\mathrm d}t} \CV_\CD 
		= \int_{\CD_{\bm x}} \left( - N \CK + N^i_{\ || i} + \frac{1}{2} h^{ij} V^k \partial_k h_{ij}
			+ \partial_k V^k \right) \sqrt{h} \, {\mathrm d}^3 x \ . 
		\label{eq:vol_D_appendix_III}
\eeq
This expression then allows to catch up with the end of the derivation given in section~\ref{subsubsec:vol_lagrange},
so that a similar use of relations \eqref{eq:vol_D_var_V} and \eqref{eq:relat_vel} again gives
the evolution of the volume \eqref{eq:vol_D_appendix}. 

\section{Extrinsic averaging operator in fluid-intrinsic variables}
\label{app:usual_aver_intrinsic}

The system of equations for extrinsic averages on $\CD$ derived in section \ref{sec:extrinsic} (\textit{Theorem}~\ref{ths:av_extrinsic}) is mostly expressed in terms of geometric
variables of the $\bm n$-orthogonal hypersurfaces, such as their extrinsic curvature. We here present an alternative formulation of the same
equations focusing instead on the intrinsic, rest-frame kinematic and dynamical quantities of the fluid (see section \ref{subsubsec:kin_fluid}) which do not depend on the foliation choice.

We can first rewrite the volume expansion rate \eqref{eq:vol_rate} and the commutation rule \eqref{eq:com_rule_final} in terms of the intrinsic local expansion rate of the fluid $\Theta$ by re-expressing the three-divergence of $N \bm v$ as 
\beq
	\big( N v^i \big)_{|| i}
		= N v^i_{\ || i} + v^i N_{|| i} 
		= N \nabla_\mu v^\mu \ , 
	\label{eq:relar_div_vel}
\eeq
where we have employed \eqref{eq:eulerian_acc} for the last equality. Noticing that $\CK = - \nabla_\mu n^\mu$ 
and making use of expression \eqref{eq:four_vel}, we get:
\begin{equation}
- N \CK + \big( N v^i \big)_{|| i} = \frac{N}{\gamma} \Theta - \frac{1}{\gamma} \ddt\gamma = \Thetat - \frac{1}{\gamma} \ddt\gamma := \Thetat^T  \ ,
\label{eq:tilted_exp_rate_app}
\end{equation}
where we have defined a \textit{tilted and scaled expansion rate} $\Thetat^T$ out of the scaled rate $\Thetat = (N/\gamma) \, \Theta$. The factor $N / \gamma$ in $\Thetat^T$ above adjusts the local clock rates between the proper time of the fluid and the coordinate time. This can also be seen upon writing: 
\beq
	\Thetat^T = \frac{N}{\gamma} \Theta - \frac{1}{\gamma} \frac{\rm d \gamma}{{\rm d}t}
		= \frac{N}{\gamma} \left( \Theta - \frac{1}{\gamma} \frac{\rm d \gamma}{{\rm d}\tau} \right) 
		= \frac{{\rm d}\tau}{{\rm d}t} \left( \Theta - \frac{1}{\gamma} \frac{\rm d \gamma}{{\rm d}\tau} \right) \ , 
\eeq
where we have used the relation \eqref{eq:der_rel_2} 
between $\mathrm{d}/\mathrm{d}t$ and $\mathrm{d}/\mathrm{d}\tau$.
The additional tilt term $- \gamma^{-1} \, \mathrm{d}\gamma / \mathrm{d}t$ can be understood as the effect of the evolving mutual tilt (Lorentz boost) between the hypersurfaces
in which $\CD$ is embedded, and the fluid flow. This affects the local volume measure so that the evolution of the volume is not only due to the fluid's intrinsic expansion.

The above rewriting \eqref{eq:tilted_exp_rate_app} allows us to recast the volume and scale factor evolution rates, and the commutation rule, respectively, into the following expressions:
\begin{align}
 \frac{3}{a_\CD} \ddt{a_\CD} & {} = \frac{1}{\VD} \ddt{\VD}
 		= \average{\Thetat^T} \; ;
 \label{eq:vol_rate_II} \\
 \frac{\mathrm d}{{\mathrm d}t} \average{\psi} 
		& {} = \average{ \frac{\mathrm d}{{\mathrm d}t} \psi } 
		- \average{ \Thetat^T } \average{\psi} 
		+ \average{ \Thetat^T \psi } \; .
 \label{eq:com_rule_final_II}
\end{align}
We notice that, even for the general configuration we are investigating 
(see figure~\ref{fig:schem_vect_3D}), the commutation rule, as well as the domain volume expansion rate, can be cast into a simple form
with respect to the fluid quantities, although this extrinsic averaging framework requires the explicit additional contribution from the evolving tilt.

The use of the Raychaudhuri equation \eqref{eq:raych_eq} and the energy constraint \eqref{eq:gauss_eq_intr} (instead of the scalar parts of the extrinsic
$3+1$ Einstein equations \eqref{eq:evol_K}, \eqref{eq:hamilt_const}), together with the above alternative form of the commutation rule,
allows for a rewriting of the evolution equations for the extrinsic effective scale factor $a_\CD$. This yields the following equivalent formulation
of \textit{Theorem}~\ref{ths:av_extrinsic}, in terms of rescaled fluid-intrinsic kinematic and dynamical variables,
$\sigmat^2 = (N^2/\gamma^2) \, \sigma^2$,
$\omegat^2 = (N^2/\gamma^2) \, \omega^2$, $\SRt = (N^2/\gamma^2) \, \SR$, $\epsilont = (N^2/\gamma^2) \, \epsilon$,
$\tilde p = (N^2/\gamma^2) \, p$, and $\tilde{\CA} = (N^2/\gamma^2) \, \CA$
(with $\CA = \nabla_\mu a^\mu$), as well as $\tilde \cc := (N^2/\gamma^2) \,\cc$:

\begin{corgroup}
\setcounter{corgroupcount}{0}
\refstepcounter{corgroupcount}
\label{cor:extrinsic_all}
\begin{corollary}[extrinsically averaged evolution equations in fluid variables]\\
\label{cor:extrinsic_evol}
\vspace{-7pt}

\small{The averaged evolution equations for the extrinsic effective scale factor $a_\CD$ can also be written under the following form:
\begin{eqnarray}
 \label{eq:AvgHamilton_2}
 3 \left(\frac{1}{a_\CD} \ddt{a_\CD} \right)^2 & = & 8 \pi G \average{\tilde\epsilon} + \average{\tilde\cc}
		- \frac{1}{2} \average{\tilde\SR} - \frac{1}{2} \tilde \CQ^{\rm T}_\CD \ ; \\
 \label{eq:AvgRaych_2}
 \frac{3}{a_\CD} \frac{{\rm d}^2 a_\CD}{{\rm d}t^2} & = & -4 \pi G \average{(\tilde\epsilon + 3 \tilde p)}
		+ \average{\tilde\cc} + \tilde \CQ^{\rm T}_\CD + \tilde \CP^{\rm T}_\CD \ ,
\end{eqnarray}
with alternative, `tilted' kinematical and dynamical backreactions, reading respectively:
\begin{align}
\tilde \CQ^{\rm T}_\CD & := \frac{2}{3} \left[ \average{\left(\Thetat^T \right)^2} \! - \average{\Thetat^T}^2 \right]
- 2 \average{\tilde\sigma^2} + 2 \average{\tilde\omega^2} 
+ \frac{2}{3} \average{ 2 \, \Thetat^T \frac{1}{\gamma} \ddt\gamma + \left(\frac{1}{\gamma} \ddt\gamma \right)^2} \; ; \\
\tilde \CP^{\rm T}_\CD & := \average{\tilde{\mathcal A}}
\! + \average{\frac{\gamma}{N} \ddt{}\!\! \left(\frac{N}{\gamma} \right) \Thetat^T }
\! - \average{ 2 \, \Thetat^T \! \frac{1}{\gamma} \ddt\gamma + \left(\frac{1}{\gamma}\ddt\gamma \right)^2
\!\! + \frac{N}{\gamma} \ddt{} \left(\frac{\gamma}{N} \frac{1}{\gamma} \ddt{\gamma} \right) \! } .
\end{align}}
\end{corollary}

We recall that, as in \textit{Theorem}~\ref{th:av_evol}, the left-hand sides in the above equations should be seen as derivatives with respect to the
chosen parameter $t$, and be interpreted according to the physical meaning of the latter. In particular,
the term $3 \, a_\CD^{-1} \, \mathrm{d}^2 a_\CD / \mathrm{d}t^2$ in equation \eqref{eq:AvgRaych_2} corresponds in a Lagrangian picture to the proper time scale factor acceleration, but not in general.

Under this form, only two backreaction terms appear, $\tilde \CQ^{\rm T}_\CD$ and $\tilde \CP^{\rm T}_\CD$, as the tilt only contributes under these combinations. These backreaction terms will not in general be directly related to the terms $\QD$ and $\PD$
appearing in \textit{Theorem}~\ref{th:av_evol}, as they do not collect the same local terms in their expressions. They do coincide, however, for a fluid-orthogonal foliation assuming an irrotational fluid
(with in this case $\tilde \CQ^{\rm T}_\CD = \QD$ and $\tilde \CP^{\rm T}_\CD = \PD$, while $\TD = 0$).

Note that there is no explicit non-perfect fluid contribution
to these effective evolution equations, although the non-perfect fluid components of the energy-momentum tensor
do have an influence on the dynamics \textit{via} the local (and average, see below) evolution of $\epsilon$ and $p$.

As before, this set of equations goes along with an integrability condition, and must be complemented by the evolution equation for the averaged energy density and pressure.

\begin{corollary}[integrability and energy balance conditions to \textit{Corollary}~\ref{cor:extrinsic_evol}]\\
\label{cor:extrinsic_intcond}
\vspace{-7pt}

\small{The integrability condition corresponding to \textit{Corollary}~\ref{cor:extrinsic_evol} reads:
\begin{align}
\lefteqn{\ddt{} \tilde \CQ^{\rm T}_\CD + \frac{6}{a_\CD} \ddt{a_\CD} \, \tilde \CQ^{\rm T}_\CD + \ddt{} \average{\tilde\SR}
		+ \frac{2}{a_\CD} \ddt{a_\CD} \average{\tilde\SR} + \frac{4}{a_\CD} \ddt{a_\CD} \, \tilde \CP^{\rm T}_\CD} \nonumber \\
\label{IntCond_1}
& & = 16 \pi G \left(\ddt{} \average{\tilde\epsilon}
		+ \frac{3}{a_\CD} \ddt{a_\CD} \average{\tilde\epsilon + \tilde p} \right) 
		+ 2 \, \ddt{} \average{\tilde\cc} \ ,
\end{align}
while the associated averaged conservation equation for the scaled energy density $\epsilont$ and pressure $\tilde p$ becomes:
\begin{align}
 \ddt{} \average{\tilde\epsilon} + \frac{3}{a_\CD} \ddt{a_\CD} \average{\tilde\epsilon + \tilde p} =
		\left(\average{\Thetat^T} \average{\tilde p}
		- \average{\Thetat^T \tilde p}\right)
		- \average{\frac{1}{\gamma} \ddt{\gamma} \, \left(\epsilont + \tilde p \right)} \nonumber \\
\label{EnergyMomentum_WithNGamma}
		- \average{\frac{N^3}{\gamma^3} \Big(a_\mu q^\mu + \nabla_\mu q^\mu + \pi^{\mu \nu} \sigma_{\mu \nu} \Big)}
		+ 2 \average{\tilde\epsilon \, \frac{\gamma}{N} \ddt{}\! \left(\frac{N}{\gamma} \right) }  . \quad \;\;
\end{align}}
\end{corollary}
\end{corgroup}
(From the commutation rule, the expression $\langle{\Thetat^T}\rangle_\CD \average{\tilde p} - \langle{\Thetat^T \tilde p}\rangle_\CD$ can also be written as $\average{\mathrm{d} \tilde p / \mathrm{d}t} - \mathrm{d}\average{\tilde p}/\mathrm{d}t$.)

\section{Summarized literature comparison}
\label{app:lit_comp}

We present in Table \ref{table:lit_comp} a comparative overview of the various formalisms used in the existing generalization proposals
of the system of averaged scalar equations of Papers I and II to general foliations, discussed in sections~\ref{subsubsec:literature_global}--\ref{subsubsec:literature_general}, including the extrinsic averaging formalism of section~\ref{sec:extrinsic} of this work.

In this table we express all notations in terms of those
used in this work to make comparisons easier. In particular, when considering the $4-$scalar expressions of
\cite{marozzi:aver2,rasanen:lightpropagation,smirnov:aver}, we define the lapse $N$ as $(- \partial_\mu T \, \partial^\mu T)^{-1/2}$, where $T$
is the scalar function which labels the hypersurfaces. This quantity (noted $\rm \Gamma$ in \cite{rasanen:lightpropagation}) satisfies $n_\mu = - N \partial_\mu T$, thus playing an analogous role
to the $3+1$ lapse, and it indeed coincides with the usual lapse if the 
$4-$scalar formalism
is split into a $3+1$ description with the natural choice of $T$ as the time coordinate. Since the domain propagation varies between these papers, in terms of the present notations the averaging domain should be noted $\CE_{\bm \partial_t}$ or $\CE_{\bm n}$, and $\CD$ for the present paper. Rather than indicating the corresponding domain for each average and backreaction term, we remove any domain-labelling subscript for these objects in the table below for notational ease.

\clearpage

\ctable[
doinside= \renewcommand{\arraystretch}{1.6}, 
caption= {Summary of the main differences between the various generalization proposals discussed in sections~\ref{subsubsec:literature_global}--\ref{subsubsec:literature_general} for the scalar averaging of the $3+1$ Einstein equations
in arbitrary foliations. This table is split into three parts, considering respectively the setup,
the equations presented (and the corresponding effective Hubble parameter), and the main variables used.
},
botcap,
label = table:lit_comp
]{>{\centering}m{0.15\linewidth} 
>{\centering}m{0.21\linewidth}
>{\centering}m{0.168\linewidth}
>{\centering}m{0.15\linewidth}
>{\centering}m{0.15\linewidth}}{  
\tnote[a]{In \cite{futamase:aver2}, boundary terms are removed by an (\textit{a priori} background-dependent) assumption of periodic boundary conditions on the large enough but still compact domain. As discussed in section \ref{subsubsec:literature_global}, this implies equivalent results to the more general case (not considered in \cite{futamase:aver2}) of an arbitrary compact domain propagating along $\bm n$, at the expense of rest mass preservation. As the shift vector is chosen to be zero in \cite{futamase:aver2}, this would also amount to a propagation along $\bm \partial_t$.}
\tnote[b]{Formally, the boundaries of the domain are assumed to be determined by some scalar function, in which case the averages and equations
are truly covariant. However, the authors mention the difficulty of finding such a scalar on physical grounds, which may
constrain one to choose a function of the coordinates instead of a scalar, hence inducing deviations from general covariance
in the averages. The follow-up paper \cite{marozzi:aver3} makes these deviations explicit at second order in perturbation theory; Smirnov \cite{smirnov:aver} suggests a general procedure to construct the function defining the boundary as a scalar.} 
\tnote[c]{The equations would formally still hold without change if a regional domain propagating along $\bm n$ were considered instead.
However, it would not be mass-preserving in this case.}
\tnote[d]{In both cases (Beltr\'an Jim\'enez \textit{et al.} \cite{dunsby:aver} and Smirnov \cite{smirnov:aver}) it is assumed
that there are `natural' observers corresponding to some irrotational dust as part of the fluid content of the model universe,
not interacting with the rest, and the corresponding geodesic irrotational normalized velocity field
is used as the normal vector $\bm n$ to build the hypersurfaces. In \cite{dunsby:aver} it is assumed to represent Dark Matter and baryonic matter
on large scales and hence it is a well-defined part of $T_{\mu \nu}$ (whereas the remaining parts can account for other fluids such as radiation or a
Dark Energy fluid, or for effective terms due to a departure from General Relativity). In \cite{smirnov:aver}, it can either be some component
intrinsically contained within $T_{\mu \nu}$, or some `test observers' that are added to the fluid content with an assumed negligible source
contribution.}
}{  \FL[0.05em] \toprule
Reference &
 Domain boundary flow (Mass-preserving ?) &
 Fluid content &
 Formulation &
 Foliation vector $\bm n$
\ML
Tanaka \& Futamase \cite{futamase:aver2} &
 Equivalent to having $\bm n$, and $\bm \partial_t$ (No)\tmark[a] &
 General $T_{\mu \nu}$ &
 $3+1$ with $\bm N = \bm 0$ &
 General
\NN 
Larena \cite{larena:aver} &
Implicitly $\bm{\partial}_t$ (No) &
 One perfect fluid &
 $3+1$ &
 General
\NN 
Brown \textit{et al.} \cite{brown:aver} &
Implicitly $\bm{\partial}_t$ (No) &
 Sum of general fluids &
 $3+1$ &
 General
\NN 
Gasperini \textit{et al.} \cite{marozzi:aver2} &
 $\bm n$ (No) &
 One perfect fluid &
 Both (mostly)\tmark[b] $4-$scalar and $3+1$ &
 General
\NN 
R\"as\"anen \cite{rasanen:lightpropagation} &
Global domain (Yes)\tmark[c] &
 General $T_{\mu \nu}$ &
 $4-$scalar &
 General
\NN 
Beltr\'an Jim\'enez \textit{et al.} \cite{dunsby:aver} &
 $\bm n$, $ = \bm{\partial}_t$ (No) &
 General $T_{\mu \nu}$ with a dust part\tmark[d] &
 $3+1$ with $\bm N = \bm 0$ and $N = 1$ &
 Geodesic (dust velocity)\tmark[d]
\NN 
Smirnov \cite{smirnov:aver} &
 $\bm n$ (No) &
 General $T_{\mu \nu}$ plus a dust part\tmark[d] &
 Both $4-$scalar and $3+1$ &
 Geodesic (dust velocity)\tmark[d]
\NN 
Section \ref{sec:extrinsic} of this work &
 $\bm u$ (Yes) &
 One general fluid &
 $3+1$ &
 General
\LL[0.05em] \bottomrule
}

\ctable[
doinside=\renewcommand{\arraystretch}{2}, 
caption={Summary of the main differences between the various generalization proposals for the scalar averaging of the $3+1$ Einstein equations
in arbitrary foliations, discussed in sections~\ref{subsubsec:literature_global}--\ref{subsubsec:literature_general}.}, 
botcap,
continued
]{>{\centering}m{0.13\linewidth} 
>{\centering}m{0.305\linewidth}
>{\centering}m{0.132\linewidth}
>{\centering}m{0.133\linewidth}
>{\centering}m{0.127\linewidth}}{
\tnote[e]{The application paper \cite{larena:aver_app} introduces instead five possible definitions of the effective Hubble parameter $H$
in order to compare them, and derives the averaged energy constraint for each of them. The first four of these definitions are respectively:
$3 H = \CV^{-1} \, \mathrm{d}\CV / \mathrm{d}t = \langle - N \CK + N^i_{\ ||i}\rangle$ (which becomes simply
 $\naverage{- N \CK}$ later in the paper as the shift is set to zero); $3 H = \naverage{- \CK}$;
$3 H = \naverage{N h^{\mu \nu} \nabla_\mu u_\nu}$; $3H = \naverage{h^{\mu \nu} \nabla_\mu u_\nu}$; where all averages
are taken over a domain on the $\bm n$-orthogonal hypersurfaces. The last proposal for $3 H$ consists in averaging simply
the intrinsic fluid expansion rate $\Theta$ (without any lapse factor) over a domain on $\bm u$-orthogonal hypersurfaces, in case $\bm u$ is
irrotational. In all of the $\bm n$-orthogonal cases, the domain still implicitly evolves along $\bm{\partial}_t$, whereas in the last case
the averaged (dust) constraint equation is recalled from Paper~I, meaning that in this case the domain must be assumed to be fluid-comoving. (Note the corrections in \cite[sec.~II.1]{larena:aver_app} of some erroneous decompositions of velocity gradients in \cite{larena:aver} that, however, had no impact on the effective Hubble rates.)}
\tnote[f]{In the first application \cite{brown:aver_app1}, the average of $h^{\mu \nu} \nabla_\mu u_\nu$ is also considered,
while the second application \cite{brown:aver_app2} focuses on the average of $\Theta$; however, in both cases, the corresponding
averaged equations are not made explicit.}
\tnote[] {}
}{  \FL[0.05em] \toprule
Reference &
Effective Hubble parameter ($\times 3$) &
 Integrability condition &
 Averaged energy conservation &
 Inclusion of $\cc$
\ML
Tanaka \& Futamase \cite{futamase:aver2} &
 $\frac{1}{\CV} \ddt{\CV}$ $= \naverage{-N \CK}$ &
 No &
 No &
 Yes
\NN 
Larena \cite{larena:aver} &
 $\naverage{N h^{\mu \nu} \nabla_{\!\mu} u_\nu}$ $\neq \frac{1}{\CV} \ddt{\CV}$ \tmark[e] &
 Yes &
 Yes &
 Yes
\NN 
Brown \textit{et al.} \cite{brown:aver} &
 $\frac{1}{\CV} \ddt{\CV}$ $= \left\langle -N \CK + N^i_{\ ||i} \right\rangle$ \tmark[f] &
 No &
 No &
 Yes
\NN 
Gasperini \textit{et al.} \cite{marozzi:aver2} &
 $\frac{1}{\CV} \ddt{\CV}$ $= \naverage{- N \CK}$ &
 No &
 Only in case $\bm n = \bm u$ &
 No
\NN 
R\"as\"anen \cite{rasanen:lightpropagation} &
 $\frac{1}{\CV} \ddt{\CV}$ $= \naverage{-N \CK}$ &
 No &
 Yes &
 Implicit (can be included in $T_{\mu \nu}$)
\NN 
Beltr\'an Jim\'enez \textit{et al.} \cite{dunsby:aver} &
 $\frac{1}{\CV} \ddt{\CV}$ $= \naverage{-\CK}\;$ ($N=1$) &
 Yes &
 Yes &
 Implicit (can be included in $T_{\mu \nu}$)
\NN 
Smirnov \cite{smirnov:aver} &
 $\frac{1}{\CV} \ddt{\CV}$ $= \naverage{- N \CK}$ &
 No &
 No &
 Implicit (can be included in $T_{\mu \nu}$)
\NN 
Section \ref{sec:extrinsic} of this work &
 $\frac{1}{\VD} \ddt{\VD}$ $= \left\langle - N \CK \!+\! (N v^i)_{||i} \right\rangle_{\!\CD}$ &
 Yes &
 Yes &
 Yes
\LL[0.05em] \bottomrule
}

\ctable[
doinside= \renewcommand{\arraystretch}{1.4}, 
caption={Summary of the main differences between the various generalization proposals for the scalar averaging of the $3+1$ Einstein equations
in arbitrary foliations, discussed in sections~\ref{subsubsec:literature_global}--\ref{subsubsec:literature_general}.},
botcap,
continued
]{>{\centering}m{0.130\linewidth} 
>{\centering}m{0.153\linewidth}
>{\centering}m{0.187\linewidth}
>{\centering}m{0.174\linewidth}
>{\centering}m{0.184\linewidth}}{      
\tnote[g]{However, in contrast to other papers, the averages of the intrinsic dynamical quantities alone (multiplied by $N^2$)
do not appear explicitly: the dynamical variables appearing in the averaged equations are actually averages
of the local normal-frame variables as expressed in terms of the local intrinsic ones through the tilt. }
\tnote[h]{In the application paper \cite{larena:aver_app}, where the averaged energy constraint is derived for five proposals of effective
Hubble parameter choices (see footnote ${}^e$ above), the kinematic variables appearing in the equations are the best-suited for each case:
based on the normal frames in the first two cases, mixed in the third and fourth cases, and intrinsic in the last case. 
The backreaction terms introduced there also depend on the Hubble parameter choice and can be either
only $\CQ$, $\CQ$ and $\CP$, or $\CQ$ and another backreaction denoted $\CL$,
with a different expression for $\CQ$ in each case. }
\tnote[i]{ The equations in $4-$covariant form involve the normal-frame variables. In the $3+1$ form of the equations, the intrinsic dynamical variables $\epsilon$ and $p$ are used, which
allows for an explicit separation of the difference to the average of the normal-frame variables, corresponding to the `stress-energy backreaction' of the present work.}
\tnote[j]{Two forms of the equations are given, with an explicit separation of the contribution of the dynamical variables as seen
either in the normal frames, or in an independent, general frame.}
}{  \FL[0.05em] \toprule
Reference &
 Dynamical variables (from $T_{\mu \nu}$) &
 Kinematic variables and acceleration &
 Explicitly identified backreaction terms &
 Main local time derivative
\ML
Tanaka \& Futamase \cite{futamase:aver2} &
 Normal-frame &
 Normal-frame (\textit{i.e.}, built from the extrinsic curvature and $N$) &
 None &
 $\partial_t \big|_{x^i}$ ($= N n^\mu \partial_\mu$) 
\NN 
Larena \cite{larena:aver} &
 Intrinsic\tmark[g] &
 Mixed (\textit{e.g.}, $h^{\mu \nu} \nabla_\mu u_\nu$)\tmark[h] &
 $\CQ$, $\CP$, and three more\tmark[h] &
 $\partial_t \big|_{x^i}$
\NN 
Brown \textit{et al.} \cite{brown:aver} &
 Intrinsic (for each fluid) &
 Normal-frame &
 $\CQ$, $\CP$, and one $\CT$ term per fluid &
 $\partial_t \big|_{x^i}$ 
\NN 
Gasperini \textit{et al.} \cite{marozzi:aver2} &
 Normal-frame /intrinsic\tmark[i] &
 Normal-frame &
 None &
 $N n^\mu \partial_\mu$
\NN 
R\"as\"anen \cite{rasanen:lightpropagation} &
 Either normal-frame or in a general frame\tmark[j] &
 Normal-frame &
 Only $\CQ$ &
 $N n^\mu \partial_\mu$ 
\NN 
Beltr\'an Jim\'enez \textit{et al.} \cite{dunsby:aver} &
 Normal-frame &
 Normal-frame &
 Only $\CQ$; there, $\CP = 0$ &
 $\partial_t \big|_{x^i}$ ($= n^\mu \partial_\mu$) 
\NN 
Smirnov \cite{smirnov:aver} &
 Normal-frame (plus $T^\mu_{\ \mu}$) &
 Normal-frame &
 Only $\CQ$, not in all equations &
 $N n^\mu \partial_\mu$
\NN 
Section \ref{sec:extrinsic} of this work &
 Intrinsic &
 Normal-frame &
 $\QD$, $\PD$, $\CT_\CD$ &
 $\ddt{} = $ $\frac{N}{\gamma} u^\mu \partial_\mu$
\LL[0.05em] \bottomrule
}

\clearpage
\noindent \textbf{Additional specificities of some of the papers compared here:\\}

\smallskip
\begin{itemize}
 \item[\textbullet] Tanaka \& Futamase \cite{futamase:aver2}: The commutation rule obtained in the corresponding framework is not explicitly provided or derived. It thus remains implicit as a necessary intermediate step for the derivation of the evolution equations for the effective scale factor from averages of the local dynamical equations.
We note, moreover, that the main aim of the paper, as for \cite{futamase:aver1}, is to set up a link to a 
background-dependent perturbation approach.\\
 \item[\textbullet] R\"as\"anen \cite{rasanen:lightpropagation}: The velocity field $\bm u$ that is introduced in addition to $\bm n$
is fully general and is not related either to $\bm n$ nor to the content of the model universe (it could be chosen to be the $4-$velocity of some fluid
as in the present approach, but this would be a restriction of generality). It is supposed to represent the $4-$velocity field
of the observers. In the application paper \cite{rasanen:lightpropagation_app}, this field is restricted to be everywhere very close to $\bm n$
(and so has a small vorticity), whereas $\bm n$ is assumed to be chosen such that it builds hypersurfaces of statistical homogeneity
and isotropy. These restrictions are already both suggested in the original paper \cite{rasanen:lightpropagation} but the equations are kept general.\\
 \item[\textbullet] Beltr\'an Jim\'enez \textit{et al.} \cite{dunsby:aver}: The main objective of using a general $T_{\mu \nu}$ in this paper
is to account for theories beyond General Relativity whose differences are transferred into $T_{\mu \nu}$ as effective terms. Note also that this paper
features an additional average equation giving the evolution of the average squared rate of shear $\partial_t \!\naverage{\sigma^2}$, as well as
the corresponding local equation; these equations are absent from the other papers quoted in this Appendix, including the present work
(the reason being that the resulting system is still not closed by adding this equation; work about looking deeper into the hierarchy of equations will be published in \cite{GBC}, where the next level of the hierarchy is explored through a topological approach).\\
 \item[\textbullet] Smirnov \cite{smirnov:aver}: Not only the choice of hypersurfaces (or of $\bm n$) and the choice of the time
that parametrizes them obey specific criteria in this paper, but this is also the case for the averaging domain, although this is not explicit in the averaged equations
and it could as well be any domain propagated along the chosen $\bm n$. Indeed, the domain is there chosen as a `sphere'
in some $\bm n$-comoving coordinates $Z^i$ on the hypersurfaces, as defined by $H_{i j} Z^i Z^j \leq r_0$ for some $r_0 > 0$
and with $H_{i j}$ the components of the spatial metric in these coordinates. This choice was a response to the series of papers
of Gasperini \textit{et al.} and Marozzi \cite{marozzi:aver1,marozzi:aver2,marozzi:aver3} to show how one may determine
the boundary of the domain \textit{via} a scalar function (here in the sense that the $Z^i$ are fixed \textit{a priori} without any link to the actual spacetime coordinates choices).\\
 \item[\textbullet] In the present work, we also introduce, in section~\ref{sec:intrinsic}, a different (intrinsic) averaging
formalism that measures scalar quantities and volume in the local rest frames of the fluid, even if they are integrated over a domain lying
in the not necessarily fluid-orthogonal hypersurfaces. We then obtain the corresponding commutation rule and averaged dynamics under rather
simple forms in terms of the \textit{intrinsic} kinematic and dynamical quantities of the fluid ---for instance, the effective
Hubble parameter, still defined as $1/3$ of the volume expansion rate, can be simply expressed as the average of $(N/\gamma) \, (\Theta/3)$---
and only two backreactions, kinematical and dynamical, distinct in general from the terms $\QD$, $\PD$ of section~\ref{sec:extrinsic}. This
formalism and this system of equations clearly differ from the literature compared in this Appendix (including our section~\ref{sec:extrinsic},
although it otherwise follows the same setup), due to the different volume and averaging operator definition.

We also give in the present work, in Appendix~\ref{app:usual_aver_intrinsic}, a re-expression of the averaged equations arising from the extrinsic averaging operator of section~\ref{sec:extrinsic}, in terms of averages of fluid-intrinsic kinematic and dynamical variables. This is obtained at the expense of the appearance of additional contributions from the evolution of the tilt factor.
\end{itemize}
\newpage
\section{Remarks on volume $3-$forms and manifestly covariant rewritings}
\label{app:volumes}
Let us consider the two different volume $3-$forms obtained as the respective Hodge duals $\star \underline{\bm n}$ , $\star \underline{\bm u}$ to the $1-$forms 
$\underline{\bm n} = n_\mu \, \mathbf{d}x^\mu$ (metric-dual to the normal vector to the slices $\bm n$) and $\underline{\bm u} = u_\mu \, \mathbf{d}x^\mu$ (metric-dual to the fluid $4-$velocity vector $\bm u$). We can decompose these $3-$forms as follows in the exact basis $(\mathbf{d}x^\mu) = (\mathbf{d}t,\mathbf{d}x^i)$, associated with the reference coordinates $(t,x^i)$:
\begin{align}
\star \underline{\bm n} &= \frac{1}{6} n^\mu \varepsilon_{\mu \nu \rho \sigma} \, \mathbf{d}x^\nu \wedge \mathbf{d}x^\rho \wedge \mathbf{d}x^\sigma \nonumber \\
					&= \frac{1}{6} \sqrt{h} \,\epsilon_{ijk} (\mathbf{d}x^i + N^i \mathbf{d}t) \wedge (\mathbf{d}x^j + N^j \mathbf{d}t) \wedge (\mathbf{d}x^k + N^k \mathbf{d}t)  \; ; \\
\star \underline{\bm u} &= \frac{1}{6} u^\mu \varepsilon_{\mu \nu \rho \sigma} \, \mathbf{d}x^\nu \wedge \mathbf{d}x^\rho \wedge \mathbf{d}x^\sigma \nonumber \\
					&= \frac{1}{6} \sqrt{b} \,\epsilon_{ijk} (\mathbf{d}x^i - V^i \mathbf{d}t) \wedge (\mathbf{d}x^j - V^j \mathbf{d}t) \wedge (\mathbf{d}x^k - V^k \mathbf{d}t) \; ,					
\end{align}
where $\varepsilon_{\mu \nu \rho \sigma}$ denote the components of the four-dimensional Levi-Civita tensor,
while $\epsilon_{ijk}$ is simply the Levi-Civita \textit{symbol} with $3$ (spatial) indices. We have used above the components expressions \eqref{eq:n_vec} and \eqref{eq:comp_u} of $\bm n$ and $\bm u$, respectively, and we have related the modulus $g$ of the $4-$metric determinant to the spatial determinants $h$ and $b$ \textit{via} $h = g/N^2$ (see, \textit{e.g.}, \cite[sec.~5.2.3]{gourg:foliation}, or \cite[sec. 3]{marozzi:aver2}) and $b = \gamma^2 h = (\gamma/N)^2 g$ (from Eq.~\eqref{eq:rel_volume_rates}).
Note that in the expression for $\star \underline{\bm n}$, $\mathbf{d}x^i + N^i \mathbf{d}t$ is the hypersurface-projected element $h^i_{\ \mu}\, \mathbf{d}x^\mu$, related to the rewriting of the four-dimensional line element \eqref{eq:line_elem} in all generality as $\mathrm{d}s^2 = - N^2 \mathrm{d}t^2 + h_{ij} (\mathrm{d}x^i + N^i \mathrm{d}t) (\mathrm{d}x^j + N^j \mathrm{d}t)$.

The restrictions of these $3-$forms to the tangent spaces to the spatial hypersurfaces thus read $( \star \underline{\bm n} )_{|\mathrm{\Sigma}} = \sqrt{h} \, \mathrm{d}^3 x$ and $( \star \underline{\bm u} )_{|\mathrm{\Sigma}} = \sqrt{b} \, \mathrm{d}^3 x$, respectively, since the integration element on the slices $\mathrm{d}^3 x$ previously introduced corresponds to $(1 / 6) \, \epsilon_{ijk} \, (\mathbf{d}x^i \wedge \mathbf{d}x^j \wedge \mathbf{d}x^k )_{|\mathrm{\Sigma}}$. We thus get, for any scalar $\psi$, using here again the label ${}^h$ for the extrinsic volume and averages for a clearer distinction:
\beq
\VD^h \left\langle \psi \right\rangle_\CD^h = \int_\CD \psi \, \sqrt{h} \,\mathrm{d}^3 x  = \int_\CD \psi \, \left(\star \underline{\bm n} \right)_{|\mathrm{\Sigma}}  \; , \;\; \text{and,} \;\; \VDb \baverage{\psi} = \int_\CD \psi \,\sqrt{b} \,\mathrm{d}^3 x = \int_\CD \psi \, \left(\star \underline{\bm u} \right)_{|\mathrm{\Sigma}} \; , 
\eeq
the $\psi = 1$ case giving the volumes $\VD^h$ and $\VDb$.
We thus recover under this form the extrinsic and intrinsic operators of sections~\ref{sec:extrinsic} and \ref{sec:intrinsic} as natural averages (and volumes) definitions based, respectively, on the hypersurface normal, or on the fluid $4-$velocity.

The similar roles played by $\bm n$ in the extrinsic averaging formalism and by $\bm u$ in the intrinsic formalism can alternatively be highlighted by rewriting the corresponding volumes and averages under a manifestly $4-$covariant form using a $4-$scalar window function \cite{marozzi:aver2,asta1}. The framework of Gasperini \textit{et al.} \cite{marozzi:aver2} indeed gives a manifestly covariant form of the extrinsic volume and of the associated averaging operator as (for any scalar $\psi$):
\begin{equation}
\VD^h = \int_{\CM^4} \mathfrak{W}_\CD^h \, \sqrt{g} \, \mathrm{d}^4 x \quad ; \quad \left\langle \psi \right\rangle_\CD^h = \frac{1}{\VD^h} \int_{\CM^4} \mathfrak{W}_\CD^h \, \psi \sqrt{g} \, \mathrm{d}^4 x \; ,
\label{eq:vol_avg_cov_form}
\end{equation}
respectively. Here, $\mathfrak{W}_\CD^h$ is the window function selecting the averaging domain:
\begin{equation}
\mathfrak{W}_\CD^h(x^\alpha) := n^\mu \nabla_\mu \big({\mathcal H}(A(x^\alpha) - A_0) \big) {\mathcal H}(r_0 - B(x^\alpha)) \; ,
\end{equation}
where $\mathcal{H}$ is the Heaviside step function. The scalar functions $A$, $B$ respectively define the foliation and the four-dimensional tube spanned by the domain, with the parameters $A_0$, $r_0$ selecting respectively a specific slice as $\{ x^\alpha \vert \; A(x^\alpha) = A_0 \}$ and the domain boundaries from the condition $B(x^\alpha) \leq r_0$. (In \cite{marozzi:aver2}, a hypersurface-orthogonal propagation is then set for the averaging domain; this is done by requiring $n^\mu \partial_\mu B = 0$.) As pointed out and analyzed in more detail in \cite{asta1}, the fluid volume \eqref{eq:vol_D_intrinsic} and the intrinsic averager \eqref{eq:spat_aver_intrinsic} can be written under the same form by simply replacing $\bm n$ by $\bm u$ in the window function, \textit{i.e.}, by replacing the window function $\mathfrak{W}_\CD^h$ by a window function $\mathfrak{W}_\CD^b$ in \eqref{eq:vol_avg_cov_form}, with
\begin{equation}
\mathfrak{W}_\CD^b(x^\alpha) := u^\mu \nabla_\mu \big({\mathcal H}(A(x^\alpha) - A_0) \big) {\mathcal H}(r_0 - B(x^\alpha)) \; .
\end{equation}
With the addition of the constraint $u^\mu \partial_\mu B = 0$, which defines a comoving domain propagation, this yields the same intrinsic averaging framework as that considered in our section~\ref{sec:intrinsic} above.

\section{Remark on the threading lapse}
\label{app:threading}

We here add a technical remark. In the fluid-intrinsic approach we can borrow one element from the $1+3$ formalism to foliate spacetime, the so-called \textit{spacetime threading}, although spatial volume averaging only makes sense on hypersurfaces. We recall that in the $1+3$ decomposition, and in comoving spatial coordinates, the four-dimensional line element \eqref{eq:line_elem_wLag} reads (see, \textit{e.g.}, \cite[sec. 10]{jantzen:GEM}):
\beq
{\mathrm d}s^2 = - \CM^2 \; {\mathrm d}t^2 + 2 \, \CM^2 \CM_i  \; {\mathrm d}t \, {\mathrm d}X^i + 
	( b_{ij} - \CM^2 \CM_i \CM_j) \; {\mathrm d}X^i {\mathrm d}X^j \ ,
\eeq
with $\CM$ the threading lapse and $\underline{\bm \CM}$ the threading shift, which relate to the lapse and to the Eulerian velocity dual $1-$form $\underline{\bm v}$ as follows:
\beq
\CM := \frac{N}{\gamma} \;\; ; \qquad 
\underline{\bm \CM} := \frac{\gamma^2}{N} \left( 
\underline{\bm v} + v^k v_k \, \underline{\bm n} \right) 
\; : \; 
\CM \CM_\mu = \gamma (0, v_i) = (0, u_i) \ .
\eeq
In the Lagrangian description we have: 
\beq
\CM = 1 \;\; ; \qquad 
\underline{\bm \CM} = \gamma \left( 
\underline{\bm v} + v^k v_k \, \underline{\bm n} \right) 
\; : \; \CM_\mu = \gamma (0, v_i) = (0, u_i) \ .
\eeq
Note that in the intrinsic form of the general averaged equations (see section~\ref{subsec:av_cosmo_intrinsic}), we only deal with occurrences of $N / \gamma=\CM$.

\section{Erratum}
\label{app:erratum}

We wish to point out a small mistake in Paper II \cite{buchert:av_pf} (repeated in the appendix of \cite{buchert:review} after equations (A23) and (A28) therein). For this we recall
the spatial components of the $4-$acceleration, $a_i = N_{|| i} / N$, its $4-$divergence
$\mathcal{A} := \nabla_{\mu}a^{\mu} = a^i_{|| i} + a^i a_i$, and the expression of the latter in terms of the lapse $N$ or the injection energy per fluid element $h$ (related to the relativistic enthalpy),
\[
\mathcal{A} = \left(\frac{N^{|| i}}{N}\right)_{|| i} + \frac{N^{||i} N_{||i}}{N^2} =  \frac{N^{|| i}_{\ \ || i}}{N} =  h \left( \frac{1}{h} \right)^{||i}_{\ \ ||i} = - \frac{h^{||i}_{\ \ ||i}}{h} + 2 \frac{h^{||i} h_{||i}}{h^2}\;\;,
\]
which are correctly written (for such a fluid-orthogonal framework). However, the first equality in equation (10a) of Paper II (and of its review in \cite{buchert:review}) is incorrect, ${\cal A} \ne \left(N^{|| i} /N \right)_{|| i}$, due to an omission of the $a^i a_i$ contribution to $\mathcal{A}$ here. 

There is also an imprecise statement: in Paper II, in footnote 3, it is stated that for scalars the operator $||$ amounts to a partial derivative. This statement is only true for spatial components (for a scalar, ${ }_{||i} = \partial_i$, but ${ }_{||0} {} = N^i \partial_i \neq \partial_t$; ${ }_{||0}$ was identically zero for scalars in Paper~II due to the vanishing shift). 

\clearpage



\end{document}